\documentclass[a4paper]{article}
\usepackage{enumitem}
\usepackage{calc}

\usepackage{lmodern} 

\usepackage{graphicx} 

\usepackage{amsmath, accents}
\usepackage{amssymb,tikz}
\usepackage{pict2e}
\usepackage{amsthm}
\usepackage{amscd}
\usepackage{mathrsfs}
\usepackage{todonotes}
\usetikzlibrary{calc} 

\usepackage{yfonts}

\usepackage{marvosym}
\usepackage[percent]{overpic}

\newtheorem{theorem}{Theorem} [section]
	
\newtheorem{corollary}[theorem]{Corollary}	
\newtheorem{lemma}[theorem]{Lemma}

\newtheorem{remark}[theorem]{Remark}
\theoremstyle{definition}

\DeclareMathOperator{\erfc}{erfc}

\newcommand{\C}{\mathbb{C}}
\newcommand{\sfz}{\mathsf{z}}

\newcommand{\Int}{\operatorname{Int}}
\newcommand{\R}{\mathbb{R}}

\newcommand{\re}{\text{\upshape Re\,}}

\newcommand{\res}{\text{\upshape Res\,}}

\tikzset{
	master/.style={
		execute at end picture={
			\coordinate (lower right) at (current bounding box.south east);
			\coordinate (upper left) at (current bounding box.north west);
		}
	},
	slave/.style={
		execute at end picture={
			\pgfresetboundingbox
			\path (upper left) rectangle (lower right);
		}
	}
}

\let\oldbibliography\thebibliography
\renewcommand{\thebibliography}[1]{\oldbibliography{#1}
\setlength{\itemsep}{-0.5pt}}

\def\XXint#1#2#3{{\setbox0=\hbox{$#1{#2#3}{\int}$}
\vcenter{\hbox{$#2#3$}}\kern-.5\wd0}}

\usepackage{tikz}
\usetikzlibrary{arrows}
\usetikzlibrary{decorations.pathmorphing}
\usetikzlibrary{decorations.markings}
\usetikzlibrary{patterns}
\usetikzlibrary{automata}
\usetikzlibrary{positioning}
\usepackage{tikz-cd}
\tikzset{->-/.style={decoration={
				markings,
				mark=at position #1 with {\arrow{latex}}},postaction={decorate}}}
	
	\tikzset{-<-/.style={decoration={
				markings,
				mark=at position #1 with {\arrowreversed{latex}}},postaction={decorate}}}

\usetikzlibrary{shapes.misc}\tikzset{cross/.style={cross out, draw,
         minimum size=2*(#1-\pgflinewidth),
         inner sep=0pt, outer sep=0pt}}

\usepackage{pgfplots}
	
\allowdisplaybreaks	
\newcommand{\dbar}{\bar{\partial}}
\numberwithin{equation}{section}

\def\ds{\displaystyle}
\def\bigO{{\cal O}}

\usepackage[colorlinks=true]{hyperref}
\hypersetup{urlcolor=blue, citecolor=red, linkcolor=blue}

\begin{document}
\title{\vspace*{-1.5cm} Exponential moments for disk counting statistics \\ at the hard edge of random normal matrices}
\author{Yacin Ameur$^*$, Christophe Charlier$^*$, Joakim Cronvall\footnote{Centre for Mathematical Sciences, Lund University, 22100 Lund, Sweden. e-mails: yacin.ameur@math.lu.se,  christophe.charlier@math.lu.se, joakim.cronvall@math.lu.se
} \, and Jonatan Lenells\footnote{Department of Mathematics, KTH Royal Institute of Technology, 10044 Stockholm, Sweden. e-mail: jlenells@kth.se}}

\maketitle

\begin{abstract}

We consider the multivariate moment generating function of the disk counting statistics of a model Mittag-Leffler ensemble in the presence of a hard wall. Let $n$ be the number of points. We focus on two regimes: (a) the ``hard edge regime" where all disk boundaries are at a distance of order $\frac{1}{n}$ from the hard wall, and (b) the ``semi-hard edge regime" where all disk boundaries are at a distance of order $\frac{1}{\sqrt{n}}$ from the hard wall. As $n \to + \infty$, we prove that the moment generating function enjoys asymptotics of the form
\begin{align*}
& \exp \bigg(C_{1}n + C_{2}\ln n + C_{3} + \frac{C_{4}}{\sqrt{n}} + \bigO(n^{-\frac{3}{5}})\bigg), & & \mbox{for the hard edge}, \\
& \exp \bigg(C_{1}n + C_{2}\sqrt{n} \hspace{0.12cm} + C_{3} + \frac{C_{4}}{\sqrt{n}} + \bigO\bigg(\frac{(\ln n)^{4}}{n}\bigg)\bigg), & & \mbox{for the semi-hard edge}.
\end{align*}
In both cases, we determine the constants $C_{1},\dots,C_{4}$ explicitly. We also derive precise asymptotic formulas for all joint cumulants of the disk counting function, and establish several central limit theorems. Surprisingly, and in contrast to the ``bulk", ``soft edge" and ``semi-hard edge" regimes, the second and higher order cumulants of the disk counting function in the ``hard edge" regime are proportional to $n$ and not to $\sqrt{n}$.
\end{abstract}


\section{Introduction and statement of results}\label{section: introduction}

\subsection{Hard wall constraints in random matrix theory} 

In this work we study random normal matrix eigenvalues on subsets of the plane which are obtained by imposing a hard wall constraint. These eigenvalues can also be seen as repelling Coulomb gas particles at the inverse temperature $\beta=2$. While we shall soon specialize to a class of Mittag-Leffler ensembles, it
is convenient to start out from a broader perspective.


 Thus we fix an arbitrary lower semi-continuous function $Q_0: \mathbb{C} \to \mathbb{R}\cup \{+\infty\}$.
Along with $Q_0$ we fix a suitable closed subset $C$ of $\mathbb{C}$ and consider the modification (``external potential''):
\begin{align}\label{potential QK}
Q(z) = \begin{cases}
Q_0(z), & \mbox{if } z \in C, \\
+\infty, & \mbox{otherwise}.
\end{cases}
\end{align}
The external potential is assumed to be finite on some set of positive capacity and to satisfy the basic growth constraint
\begin{align}\label{growth cond}
{Q(z)}-\ln |z|^{2}\to +\infty,\qquad \text{as}\qquad z\to\infty.
\end{align}

Observe that $Q$ may satisfy the growth condition \eqref{growth cond} even if $Q_0$ fails to do so. In particular, this
is the case if $Q_0$ is a constant, or if $Q_0$ is an Elbau-Felder potential \cite{ElbauFelder,IT,LM,AKS2018}:
$$Q_0(z)=\tfrac 1 {t_0}(|z|^2-2\re (t_1z+\cdots+t_kz^k)).$$
Another basic class of hard walls is obtained by taking $C=\mathbb{R}$, which leads to the Hermitian random matrix theory.

Given a confining potential $Q$, we associate Coulomb gas ensembles in the following way (as mentioned, we will only consider the inverse temperature $\beta=2$).
We consider configurations of $n$ points $\{\mathsf{z}_{j}\}_{j=1}^{n}\subset \C$. The total energy, or Hamiltonian of the configuration, is defined by
\begin{align*}
H_{n} = \sum_{\substack{j,k=1 \\ j \neq k }}^{n} \ln \frac{1}{|\mathsf{z}_{j}-\mathsf{z}_{k}|} + n \sum_{j=1}^{n} Q(\mathsf{z}_{j}),
\end{align*}
and the associated Boltzmann-Gibbs measure on $\mathbb{C}^{n}$ is
\begin{align}\label{bogg}
d \mathsf{P}_{n} = \frac{1}{\mathsf{Z}_{n}} e^{-H_{n}} \prod_{j=1}^{n} d^{2}\mathsf{z}_{j},
\end{align}
where $d^{2}\mathsf{z}$ is the two-dimensional Lebesgue measure. The Coulomb gas ensemble (or ``system'') $\{\mathsf{z}_{j}\}_{j=1}^{n}$ corresponding to the external potential $Q$ is a configuration picked randomly with respect to this measure.

To a first order approximation, the system tends to follow Frostman's equilibrium measure $\mu$ associated to the potential $Q$. This is the unique minimizer of the weighted logarithmic energy functional
\begin{align*}
I_Q[\nu] = \iint_{\mathbb{C}^{2}} \ln \frac{1}{|z-w|}\, d\nu(z) d\nu(w) + \int_{\mathbb{C}} Q(z)\, d \nu(z)
\end{align*}
among all compactly supported Borel probability measures on $\mathbb{C}$. The support of $\mu$ is called the droplet and is denoted $S=S[Q]$.
If the potential is $C^{2}$-smooth in a neighborhood of $S$, then the equilibrium measure is absolutely continuous with respect to
the two-dimensional Lebesgue measure $d^2 z$ and takes the form (see \cite{SaTo})
\begin{align}\label{eqform}
d\mu(z) = \frac{1}{4\pi} \Delta Q(z) \chi_{S}(z) \,d^{2}z,
\end{align}
where $\chi_{S}$ is the indicator function of $S$ and $\Delta$ is the standard Laplacian.

It is known that the system $\{\sfz_j\}_1^n$ tends to condensate on the droplet under quite general conditions \cite{PH, Johansson, Deift, KS1999, HM2013, CGZ2014, A2021}, in the sense that as $n\to\infty$ the empirical measures $\frac{1}{n} \sum_{j=1}^{n} \delta_{\sfz_{j}}$ converge weakly to $\mu$ with high probability.

Consider now a smooth confining potential $Q_0$ on the plane whose droplet is $S_{0}$.
A case of some interest is obtained by placing the hard wall exactly along the edge of the droplet, i.e., we take $C= S_{0}$, where the equilibrium measure is still absolutely  continuous and of the form \eqref{eqform}.
 In this case, we obtain a so-called \textit{local
droplet} with a soft/hard edge. Such droplets have been studied in for example \cite{AKMW2020,HM2013,LM} and references therein. While the equilibrium measure is unchanged, the soft/hard edge produces some statistical effects near the edge. Interestingly, the concept of local droplets permits us to define some new and nontrivial ensembles,
such as the ``deltoid'' - a droplet with three maximal cusps which arises for the cubic potential $|z|^2+c\,\re(z^3)$ for a certain critical value of the constant $c$, see e.g. \cite{BK2012}.

However, the main case of interest for the present investigation is that of a hard wall in the bulk of the droplet. To study this case, we choose an external potential $Q_0$ giving rise to
a well-defined droplet $S_0$ and a closed subset $C\subset \Int S_{0}$, and we modify $Q_0$ to a potential $Q$ by defining it as $+\infty$ outside $C$. This has an effect even at the level of the equilibrium measure. Indeed, if the potential
 $Q_0$ is $C^2$-smooth in a neighborhood of $S_0$, then this effect is given by a balayage process which we briefly recall.

Let $\mu_0$ be the equilibrium measure with respect to the potential $Q_0$, given in \eqref{eqform} (with ``$S$'' and ``$Q$" replaced by ``$S_0$'' and ``$Q_{0}$"). Assuming some regularity of the boundary $\partial C$, the equilibrium measure $\mu_h$ corresponding to the potential $Q$ is then given by the formula (see \cite[Theorem II.5.12]{SaTo})
\begin{align}\label{balform}
\mu_h = \mu_{0}\cdot \chi_{C} + \mbox{Bal} \big( \mu_{0}|_{S_{0} \setminus C} , \partial C \big),
\end{align}
where $\mbox{Bal} \big( \mu_{0}|_{S_{0} \setminus C} , \partial C \big)$ is the balayage of $\mu_{0}|_{S_{0} \setminus C}$ onto the boundary $\partial C $. The formula \eqref{balform} expresses the fact that the portion $\mu_{0}|_{S_{0} \setminus C}$ is swept onto the boundary $\partial C$ according to the balayage operation, which preserves (up to a constant) the exterior logarithmic potential in the exterior of the droplet $S_{0}$. See \cite[Sections II.4 and II.5]{SaTo} as well as \cite{CFLV,Seo2021,JLM1993} for more details about the balayage.

The balayage part of \eqref{balform} is a density on the curve $\partial C$, so this part is singular with respect to the two-dimensional Lebesgue measure. We think of this balayage as a first approximation of the density for the particles which would have occupied the forbidden region outside of $C$, were it not for the hard wall. On a statistical level, in the generic case where $\Delta Q(z)>0$ for all $z \in \partial C$, the particles which are swept out of the forbidden region are expected to occupy a very narrow interface about the boundary $\partial C$ of width of order $1/n$. We call this interface the ``hard edge regime". The width $1/n$ is substantially smaller than the two-dimensional microscopic scale $1/\sqrt{n}$. We shall find below that on a $1/\sqrt{n}$-scale from $\partial C$, we obtain a transitional regime between hard edge and bulk statistics, which we call ``semi-hard edge regime". The three regimes (bulk, semi-hard edge, and hard edge) each gives rise to different kinds of statistical behavior, which we study below for a class of radially symmetric potentials.

 We remark that point-processes $\{\sfz_j\}_1^n$ of the above type can be identified
 with the eigenvalues of an $n \times n$ random normal matrix $M$, picked randomly according to the probability measure proportional to $e^{-n\, \mathrm{tr}\, Q(M)}dM$, where ``$\mathrm{tr}$" is the trace and $dM$ is the measure on the set of $n \times n$ normal matrices induced by the flat Euclidian metric of $\mathbb{C}^{n\times n}$  \cite{Mehta, CZ1998, ElbauFelder}. (Note that this makes precise the identification between eigenvalues and $\beta=2$ Coulomb gas processes mentioned above.)
 
 The process $\{\sfz_j\}_1^n$ can be thought of as a conditional process where the eigenvalue process associated with $Q_0$ is conditioned on the event that none of the eigenvalues fall outside of the closed set $C$.
 If $C\subset\Int S$, we are conditioning on a rare event.

  We mention in passing that for other conditional point processes, such as the zeros of Gaussian analytic functions conditioned on a hole event, the situation is drastically different because of the presence of a forbidden region around the singular part of the equilibrium measure \cite{GN2019, NishryWennman}.

\begin{remark}
Hard wall ensembles from Hermitian random matrix theory have been well-studied in the literature, see for example \cite{Fo, DeiftKrasVasi, CD2019, CharlierBessel, LCX2022, DXZ2022 bis, DZ2022}; see also \cite{CK} for a soft/hard edge. We remark that imposing a hard wall
in the interior of a one-dimensional droplet has a well-known global effect on the equilibrium measure, in contrast to \eqref{balform} which just alters the measure locally at the edge. However, this apparent contradiction is quickly dispelled if we note that a one-dimensional droplet consists of only edge and no interior
(regarded as a subset of $\C$).
\end{remark}


\subsection{Mittag-Leffler ensembles with a hard wall constraint} For what follows we will restrict our attention to radially symmetric potentials of the form
\begin{equation}\label{qml} Q_0(z)=|z|^{2b}-\tfrac {2\alpha} n\ln|z|,\end{equation}
where $b>0$ and $\alpha > -1$ are fixed parameters.
The unconstrained model Mittag-Leffler ensemble
is a configuration $\{\zeta_j\}_1^n$ picked randomly with respect to
the following joint probability density function
\begin{align}\label{def of point process}
\frac{1}{n!Z_{n}} \prod_{1 \leq j < k \leq n} |\zeta_{k} -\zeta_{j}|^{2} \prod_{j=1}^{n}|\zeta_{j}|^{2\alpha}e^{-n |\zeta_{j}|^{2b}}, \qquad \zeta_{1},\dots,\zeta_{n} \in \mathbb{C},
\end{align}
where $Z_{n}$ is the normalization constant.
It is well-known that the droplet $S_0$ corresponding to the potential \eqref{qml} is the disk of radius $b^{-\smash{\frac{1}{2b}}}$ centered at $0$; the density
 is given
 according to \eqref{eqform} by
 \begin{align}\label{cmsd}d\mu_0(z) =  \smash{\frac{b^{\smash{2}}}{\pi}}|z|^{\smash{2b-2}}d^{\smash{2}}z.
  \end{align}

\begin{remark}
The logarithmic and power-like singularities of \eqref{qml} at the origin are not strong enough to affect the equilibrium measure. The term ``Mittag-Leffler potential'' is from \cite{AK2013} and refers to a much broader class of potentials having similar kinds of singularities at the origin. The motivation for the terminology is that under some conditions, the local statistics near the origin can be described by a two-parametric Mittag-Leffler function \cite{AKS2018}.
\end{remark}

\medskip \medskip We now fix a parameter $\rho$ with $0<\rho<b^{-\frac{1}{2b}}$ and place a hard wall outside the circle $|z|=\rho$.
More precisely, we consider the probability density
\begin{align}\label{def of point process hard}
\frac{1}{n!\mathcal{Z}_{n}} \prod_{1 \leq j < k \leq n} |z_{k} -z_{j}|^{2} \prod_{j=1}^{n} e^{-nQ(z_{j})}, \qquad z_{1},\dots,z_{n} \in \mathbb{C},
\end{align}
where $\mathcal{Z}_{n}$ is the normalizing partition function and
\begin{align}\label{qdee}
Q(z) = \begin{cases}
|z|^{2b} - \frac{2\alpha}{n}\ln |z|, & \mbox{if } |z| \leq \rho, \\
+\infty, & \mbox{if } |z| > \rho.
\end{cases}
\end{align}
This gives the hard-wall Mittag-Leffler process $\{z_j\}_1^n$, conditioned on the forbidden region $\{|z|>\rho\}$. For brevity, we shall in the sequel refer to $\{z_j\}_1^n$ corresponding to the potential \eqref{qdee} as the \textit{restricted Mittag-Leffler process}.

The equilibrium  measure $\mu_h$ corresponding to the potential \eqref{qdee} can be easily computed using standard balayage techniques \cite{SaTo} (see also \cite[Section 4.1]{CFLV} or \cite{Seo2021} for details)
and is given by
\begin{align}\label{eq mes}
& \mu_{h}(d^{2}z) = \mu_{\mathrm{reg}}(d^{2}z) + \mu_{\mathrm{sing}}(d^{2}z),\nonumber
\\ & \mu_{\mathrm{reg}}(d^{2}z) := 2b^{2}r^{2b-1}dr\frac{d\theta}{2\pi}, \quad \mu_{\mathrm{sing}}(d^{2}z) := c_{\rho} \delta_{\rho}(r) dr \frac{d\theta}{2\pi},
\end{align}
where $z=re^{i\theta}$, $r>0$, $\theta \in (-\pi,\pi]$ and
\begin{align}\label{crho}c_{\rho} := \int_{\rho}^{b^{-\frac{1}{2b}}}  \hspace{-0.1cm} 2b^{2}r^{2b-1}dr = 1-b\rho^{2b}.\end{align}

Standard arguments \cite{Johansson,HM2013,A2021} show that with large probability, the empirical measures $\frac 1 n\sum \delta_{z_j}$ converge weakly to $\mu_h$ as $n\to\infty$.

Clearly, the restricted Mittag-Leffler process is an example of a rotation invariant ensemble, i.e., the joint probability density function \eqref{def of point process hard} remains unchanged if all $z_{j}$ are multiplied by the same unimodular constant $e^{i\beta}$, $\beta \in \mathbb{R}$.

\begin{figure}
\begin{center}
\begin{tikzpicture}[master]
\node at (0,0) {\includegraphics[width=4.2cm]{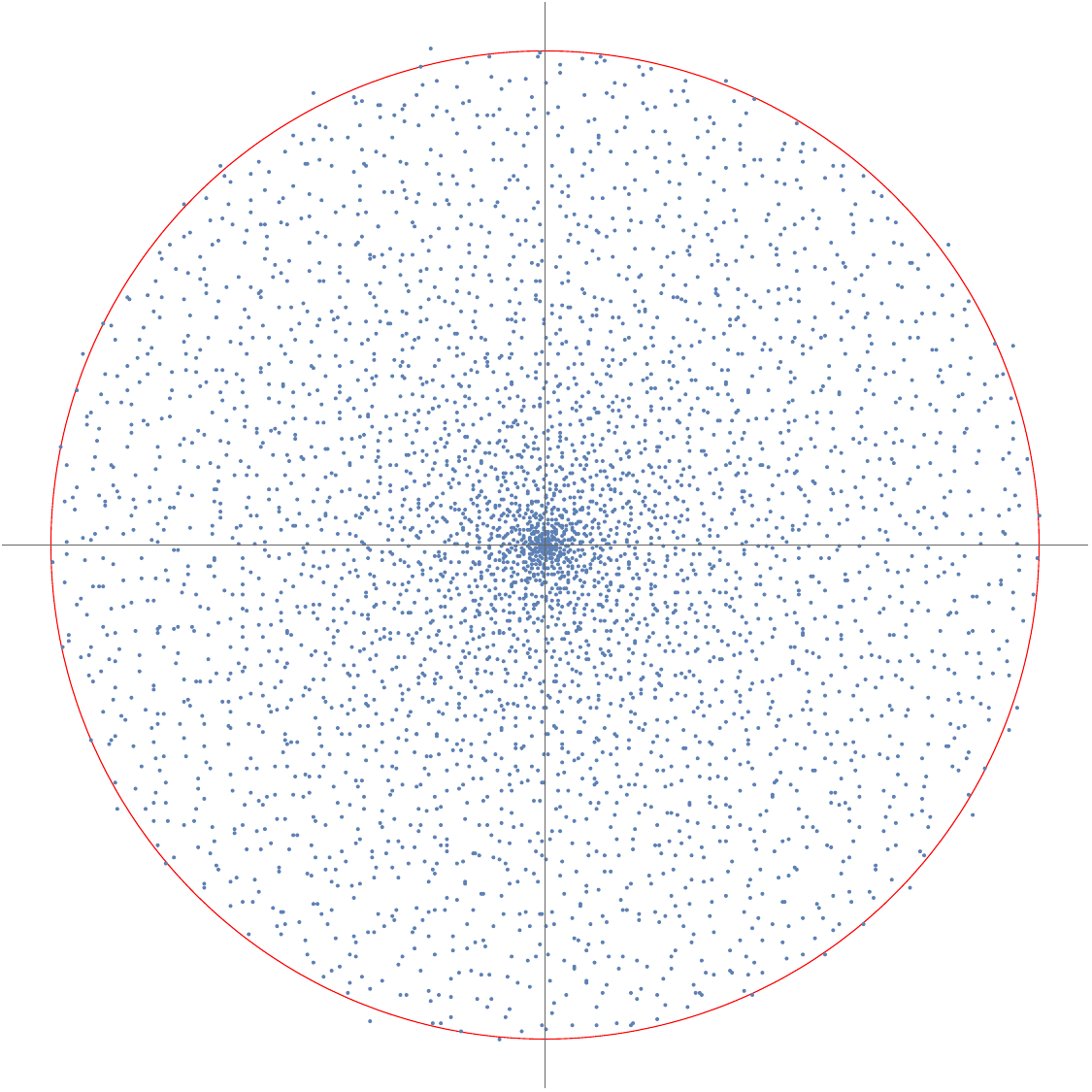}};
\node at (0,-2.3) {$b=\frac{1}{2}$};
\end{tikzpicture}
\begin{tikzpicture}[slave]
\node at (0,0) {\includegraphics[width=4.2cm]{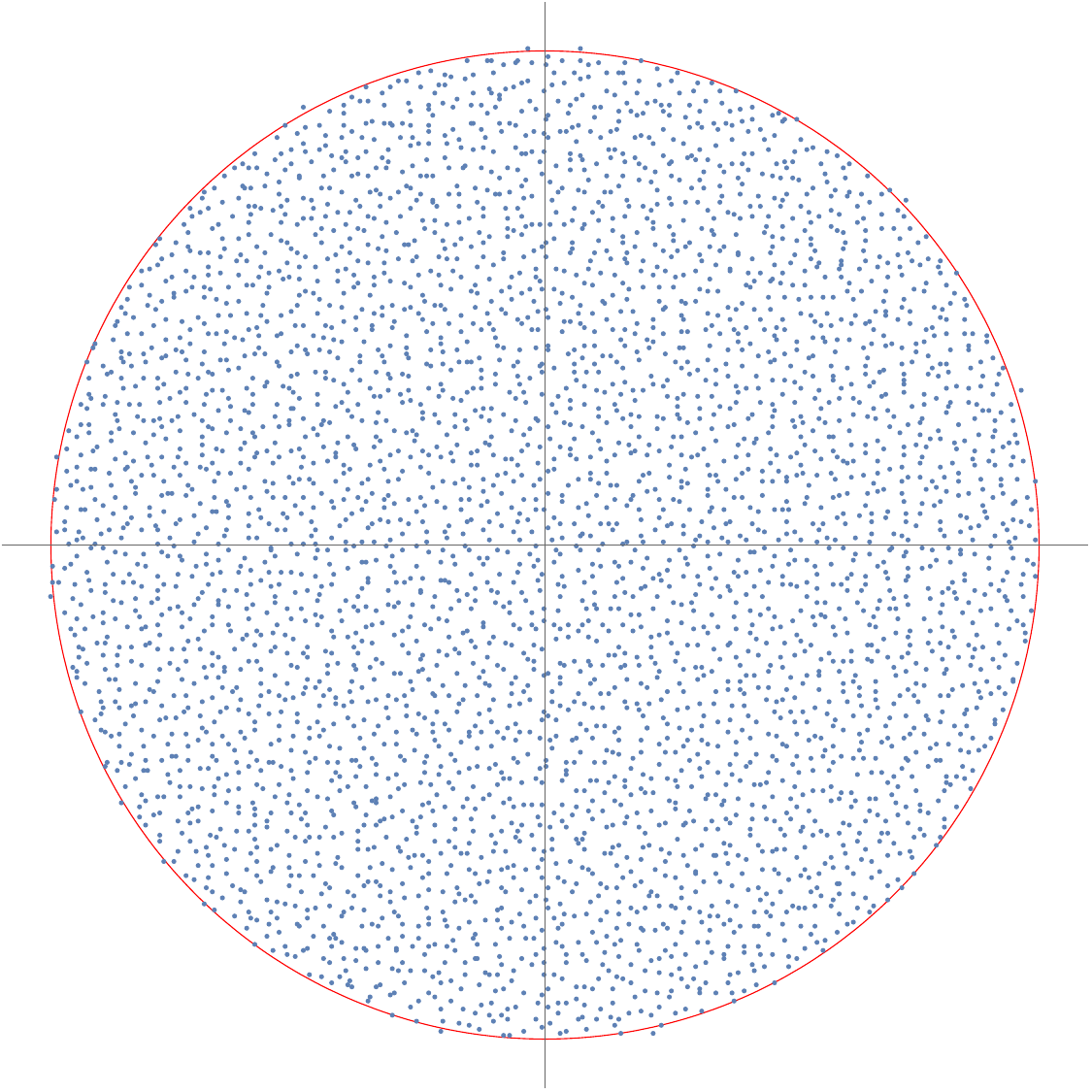}};
\node at (0,-2.3) {$b=1$};
\end{tikzpicture}
\begin{tikzpicture}[slave]
\node at (0,0) {\includegraphics[width=4.2cm]{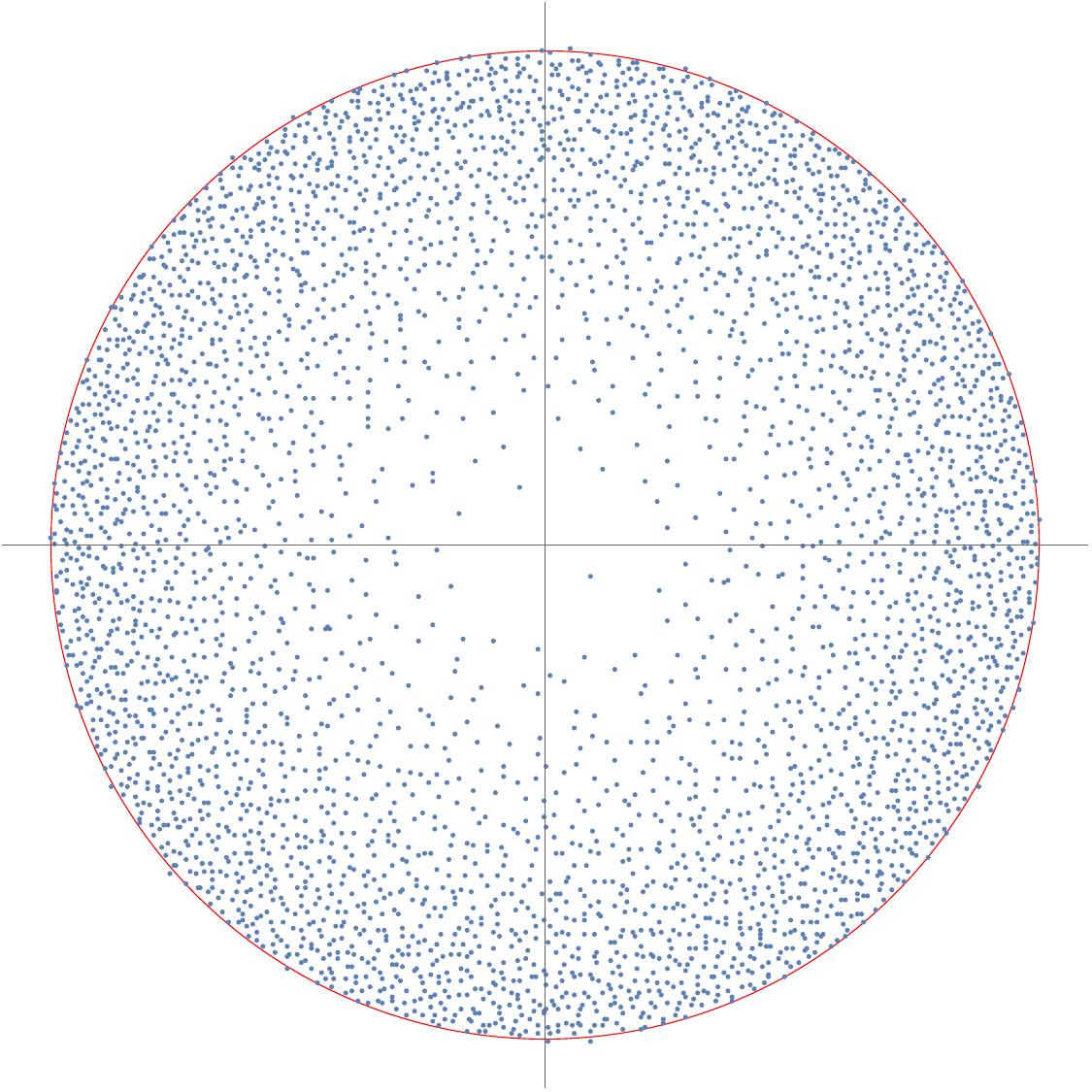}};
\node at (0,-2.3) {$b=2$};
\end{tikzpicture}
\end{center}
\begin{center}
\begin{tikzpicture}[master]
\node at (0,0) {\includegraphics[width=4.2cm]{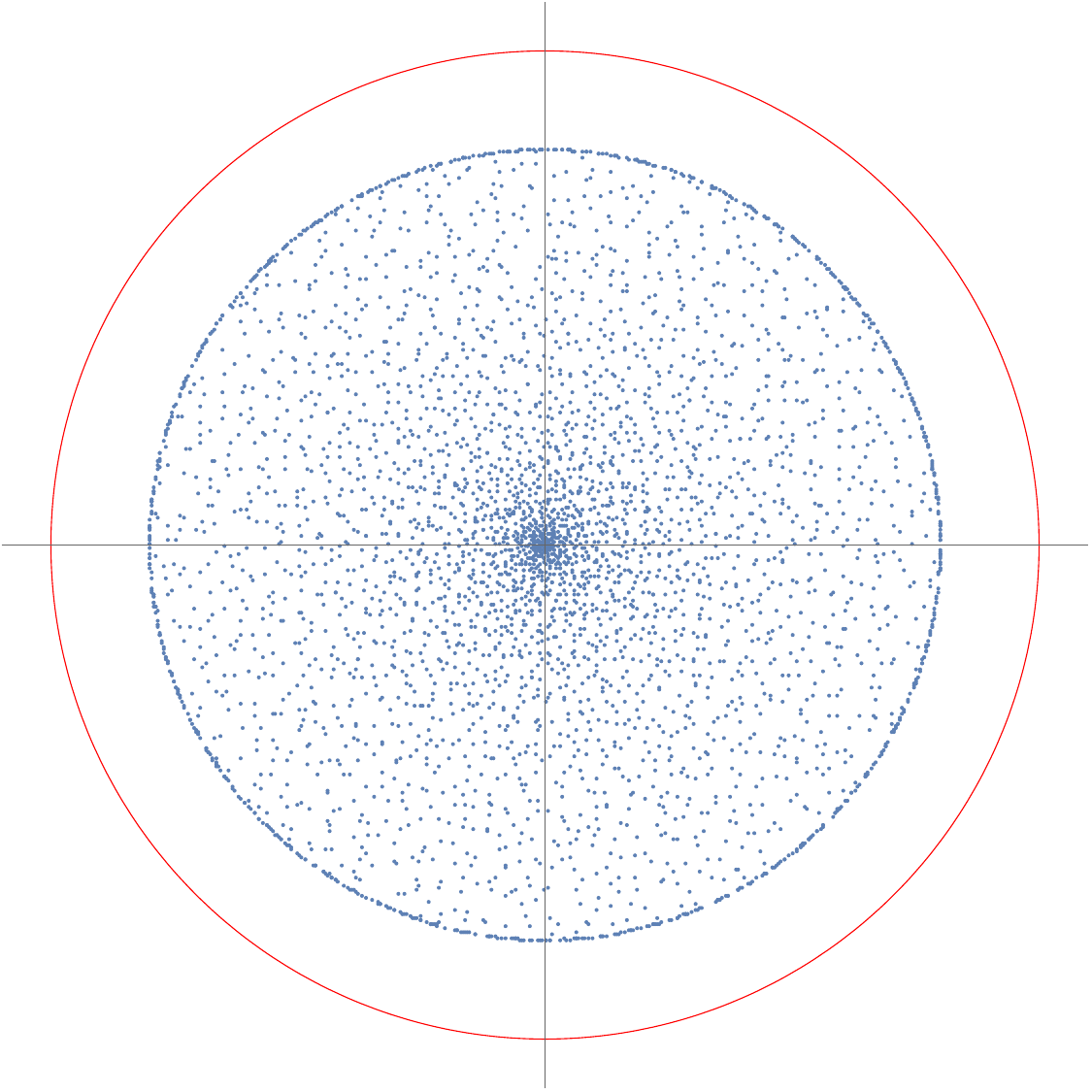}};
\node at (0,-2.3) {$b=\frac{1}{2}$};
\end{tikzpicture}
\begin{tikzpicture}[slave]
\node at (0,0) {\includegraphics[width=4.2cm]{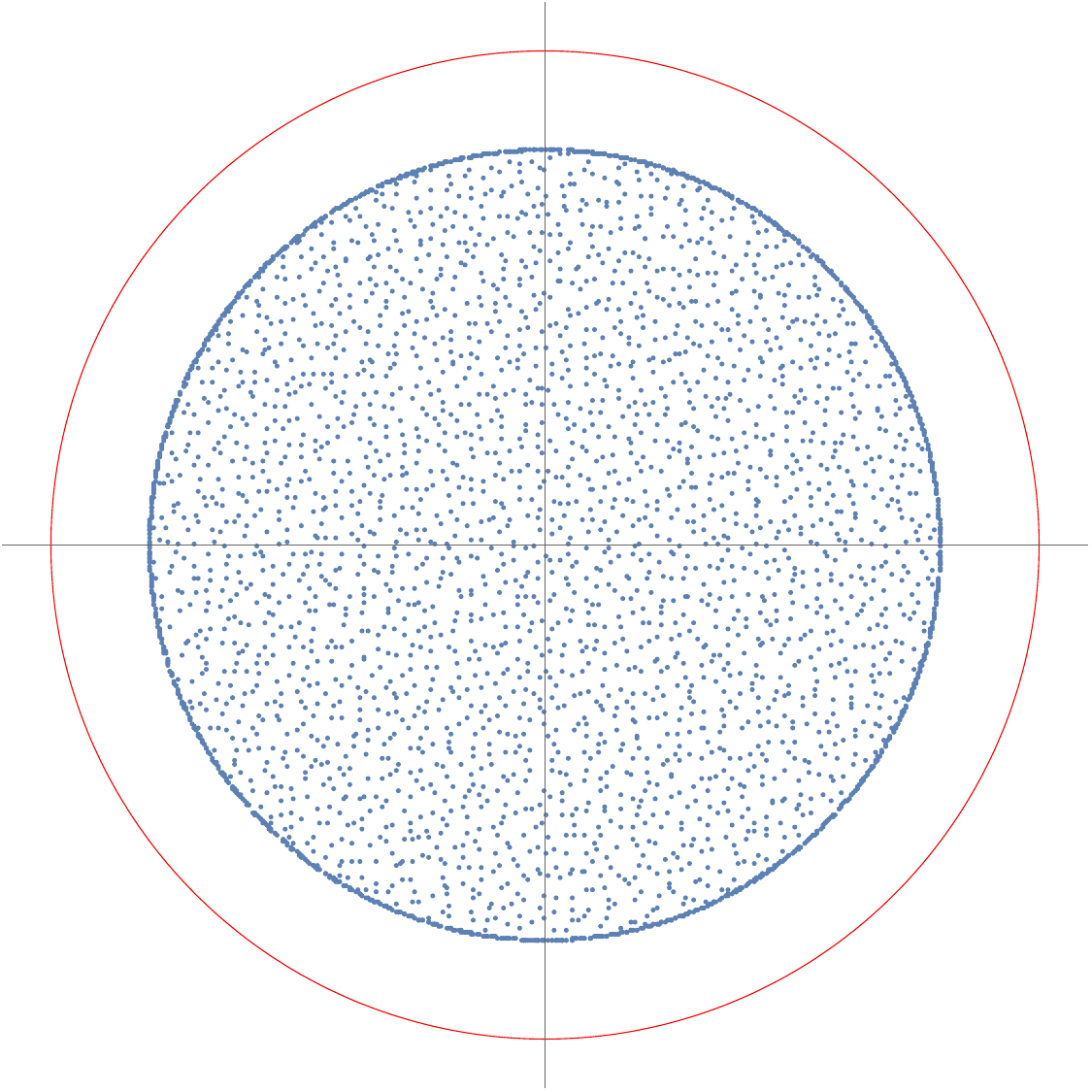}};
\node at (0,-2.3) {$b=1$};
\end{tikzpicture}
\begin{tikzpicture}[slave]
\node at (0,0) {\includegraphics[width=4.2cm]{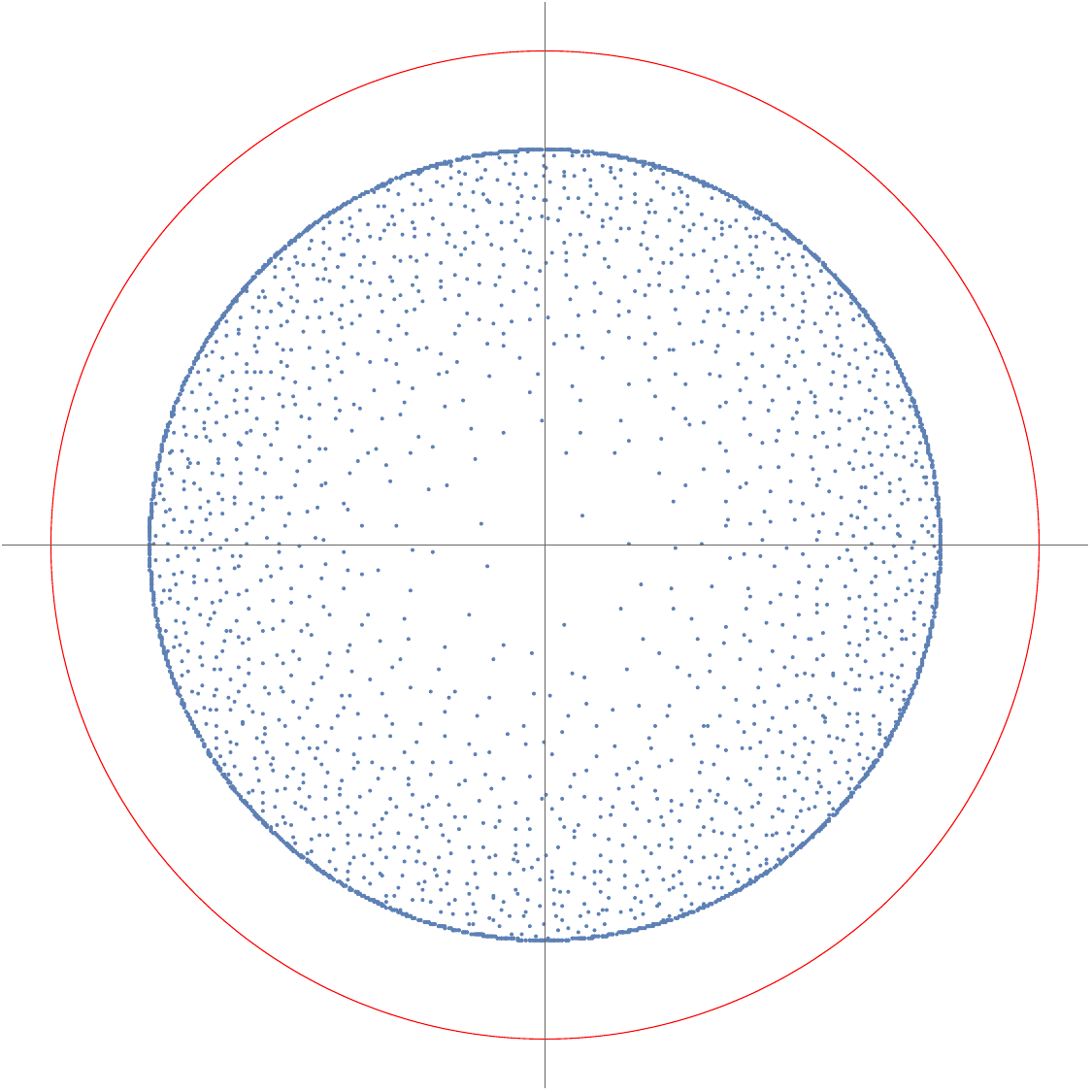}};
\node at (0,-2.3) {$b=2$};
\end{tikzpicture}
\end{center}
\vspace{-0.5cm}\caption{\label{fig: ML with hard wall}  Illustration of the point processes corresponding to \eqref{def of point process} (first row) and \eqref{def of point process hard} (second row) with $n=4096$, $\rho=\frac{4}{5}b^{-\frac{1}{2b}}$, $\alpha=0$ and the indicated values of $b$. In each plot, the red circle is $\{z \in \mathbb{C}: |z| = b^{-\frac{1}{2b}}\}$. A narrow interface about the
hard wall $|z|=\rho$, of width roughly $1/n$, accommodates the roughly $c_\rho n$ particles swept out from the forbidden region. The semi-hard regime of width roughly $1/\sqrt{n}$ is transitional between the hard edge and the bulk.}
\end{figure}
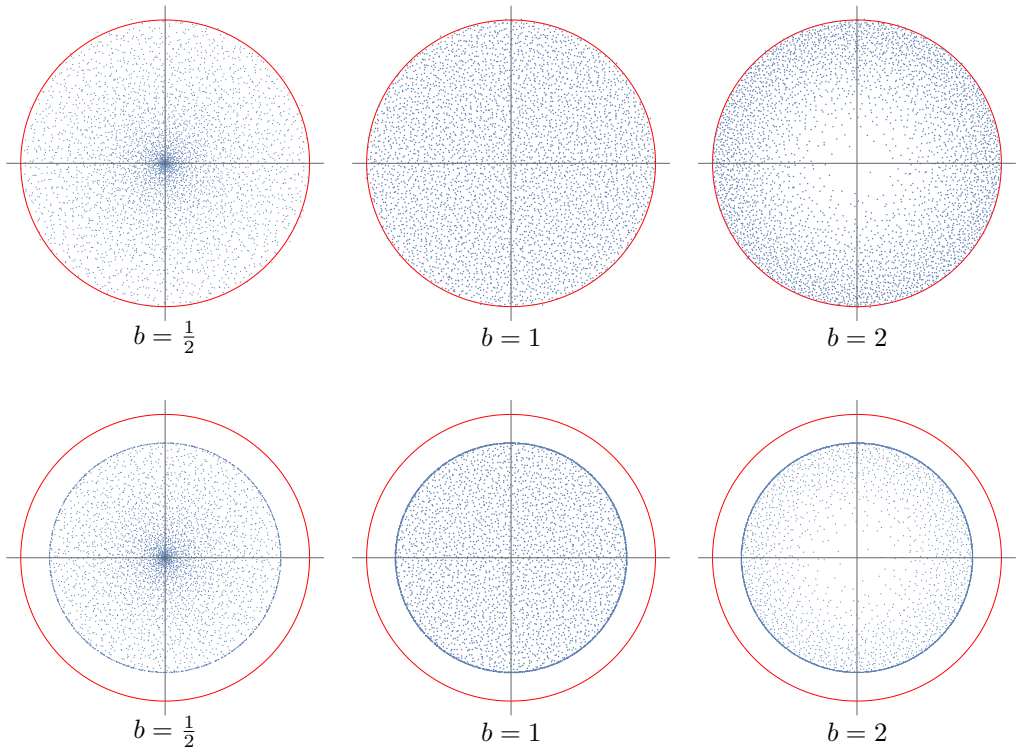
In this work we focus on the case $\rho < b^{-\frac{1}{2b}}$, which means that we are studying a hard wall in the bulk of the droplet $S_{0}$. The case of a soft/hard edge, i.e.,
$\rho=b^{-\frac{1}{2b}}$ could be included as well, but would require a somewhat different (and much simpler) analysis. We shall therefore omit this case.

Coulomb gas ensembles in the presence of a hard wall have previously been considered in the literature, but so far the focus has been on large gap probabilities (or partition functions) \cite{GHS1988, ForresterHoleProba, JLM1993, APS gap 2009, AK hole 2013, AIE gap 2014, AR Infinite Ginibre, GN2018, Charlier 2d gap} and on the local statistics \cite{ZS2000, NAKP2020, Seo2021}. We refer to \cite{AKM2019, Seo0, AKMW2020, BS2021,HM2013,LM} for studies of local droplets and local statistics near soft/hard edges.

In recent years, a lot of works dealing with the counting statistics of two dimensional point processes have appeared \cite{LeeRiser2016, CE2020, LMS2018, L et al 2019, ES2020, FenzlLambert, SDMS2020, Charlier 2d jumps, SDMS2021, AkemannSungsoo, ChLe2022}, see also \cite{Shirai} for an earlier work. A common feature of these works is that they all deal exclusively with either ``the bulk regime" or with ``the soft edge regime".

In this paper we study disk counting statistics of \eqref{def of point process hard} near the hard edge $\{|z|=\rho\}$. To be specific, let $\mathrm{N}(y):=\#\{z_{j}: |z_{j}| < y\}$ be the random variable that counts the number of points of \eqref{def of point process hard} in the disk of radius $y$ centered at $0$. Our main result is a precise asymptotic formula as $n \to + \infty$ for the multivariate moment generating function (MGF)
\begin{align}\label{moment generating function intro}
\mathbb{E}\bigg[ \prod_{j=1}^{m} e^{u_{j}\mathrm{N}(r_{j})} \bigg]
\end{align}
where $m \in \mathbb{N}_{>0}$ is arbitrary (but fixed), $u_{1},\dots,u_{m} \in \mathbb{R}$, and the radii $r_{1},\dots,r_{m}$ are merging at a critical speed. We consider several regimes:
\begin{align}
& \mbox{Hard edge: } & & 0 < r_{1} < \dots <r_{m}, \quad r_{\ell} = \rho \bigg( 1-\frac{t_{\ell}}{n} \bigg)^{\frac{1}{2b}}, & & t_{1}>\dots>t_{m}\geq 0, \label{def of rell hard} \\
& \mbox{Semi-hard edge: } & & 0 < r_{1} < \dots <r_{m}, \quad r_{\ell} = \rho \bigg( 1-\frac{\sqrt{2}\, \mathfrak{s}_{\ell}}{\rho^{b}\sqrt{n}} \bigg)^{\frac{1}{2b}}, & & \mathfrak{s}_{1}>\dots>\mathfrak{s}_{m}>0, \label{def of rell semi-hard} \\
& \mbox{Bulk: } & & 0 < r_{1} < \dots <r_{m}, \quad r_{\ell} = r \bigg( 1 + \frac{\sqrt{2}\, \mathfrak{s}_{\ell}}{r^{b}\sqrt{n}} \bigg)^{\frac{1}{2b}}, & & \mathfrak{s}_{1}<\dots<\mathfrak{s}_{m} \in \mathbb{R}, \; r < \rho. \label{def of rell bulk}
\end{align}
We emphasize that $\mathfrak{s}_{m}\neq 0$ in \eqref{def of rell semi-hard}. 

We shall prove that, as $n \to + \infty$, the joint MGF $\mathbb{E}\big[ \prod_{j=1}^{m} e^{u_{j}\mathrm{N}(r_{j})} \big]$ enjoys asymptotic expansions of the form
\begin{align}
& \exp \bigg(C_{1}n + C_{2}\ln n + C_{3} + \frac{C_{4}}{\sqrt{n}} + \bigO(n^{-\frac{3}{5}})\bigg), & & \mbox{for the hard edge}, \label{shape of asymp in hard edge} \\
& \exp \bigg(C_{1}n + C_{2}\sqrt{n} \hspace{0.14cm} + C_{3} + \frac{C_{4}}{\sqrt{n}} + \bigO\bigg(\frac{(\ln n)^{4}}{n}\bigg)\bigg), & & \mbox{for the semi-hard edge}, \label{shape of asymp in semi-hard edge} \\
& \exp \bigg(C_{1}n + C_{2}\sqrt{n} \hspace{0.14cm} + C_{3} + \frac{C_{4}}{\sqrt{n}} + \bigO\bigg(\frac{(\ln n)^{2}}{n}\bigg)\bigg), & & \mbox{for the bulk}. \label{shape of asymp in bulk}
\end{align}
For each of these three regimes, we determine $C_{1},\dots,C_{4}$ explicitly.

As can be seen from \eqref{shape of asymp in hard edge}--\eqref{shape of asymp in bulk}, the counting statistics in the hard edge regime are drastically different from the counting statistics in the bulk and semi-hard edge regimes (and also very different from the counting statistics in the soft edge regime \cite{Charlier 2d jumps, ChLe2022}). Indeed, at the hard edge the subleading term is proportional to $\ln n$, while in all other regimes it is proportional to $\sqrt{n}$. Furthermore, in the hard edge regime, the leading coefficient $C_{1}$ will be shown to depend on the parameters $u_{1},\ldots,u_{m}$ in a highly non-trivial non-linear way.

As we show below, the above asymptotic expansions have
several interesting consequences; for example $\mbox{Var}[\mathrm{N}(r_{j})] \asymp n$ in the hard edge regime, while $\mbox{Var}[\mathrm{N}(r_{j})] \asymp \sqrt{n}$ in the three other regimes (actually, a similar statement also holds for the higher order cumulants, as can be seen by comparing Corollary \ref{coro:correlation hard} with Corollary \ref{coro:correlation semihard} and \cite[Corollary 1.5]{ChLe2022}). This indicates that the counting statistics near a hard edge are considerably wilder than near a soft edge, in the bulk or near a semi-hard edge. From a technical point of view, we also found the hard edge regime to be significantly harder to analyze than the three other regimes. For example, our control of the error term in \eqref{shape of asymp in hard edge} is less precise than in \eqref{shape of asymp in semi-hard edge} and \eqref{shape of asymp in bulk}.

In contrast to earlier works on smooth and non-smooth linear statistics on the soft edge and bulk regimes, the leading coefficient $C_{1}$ in the hard edge regime is \textit{not} given by the integral of the test function (in our case $\sum_{j=1}^{m}u_{j} \chi_{(0,r_{j})}(z)$) against the equilibrium measure $\mu_{h}$, and in fact it depends in a non-linear way on the parameters $u_j$. In a sense this behavior becomes less surprising if we recall that we are not considering fixed test functions, but rather increasing sequences corresponding to characteristic functions of expanding discs, and it is known due to Seo \cite{Seo2021} that the 1-point function varies rather dramatically in the hard edge regime. On the other hand, the fact that the relationship becomes non-linear might be less clear on this intuitive level.
See also Remark \ref{remark:C1 hard edge} below for more about this.


The transition from the hard edge regime to the bulk regime is very subtle. The semi-hard edge regime lies in between, i.e., it is genuinely different from the hard edge and the bulk regimes. To the best of our knowledge, it seems that this regime has been unnoticed (or at least unexplored) in the literature so far.\footnote{In a different but somewhat related context, namely in the study of the statistics of the largest modulus of the complex Ginibre ensemble, a new intermediate regime was also recently discovered in \cite{LGMS2018}.} Our results for this regime can be seen as a first step towards understanding the hard-edge-to-bulk transition. However, the fact that the subleading terms in the hard edge and semi-hard edge regimes are of different orders indicates that there is still (at least) one intermediate regime where a critical transition takes place. We will return to this issue in a follow-up work.

As corollaries of our various results on the generating function \eqref{moment generating function intro}, we also provide central limit theorems for the joint fluctuations of $\mathrm{N}(r_{1}),\dots,\mathrm{N}(r_{m})$, and precise asymptotic formulas for all cumulants of these random variables (both at the hard edge and at the semi-hard edge). Our results for the hard edge and semi-hard edge regimes seem to be new, even for $m=1$. Our results about the bulk regime are less novel. Indeed, in this regime the asymptotics of the MGF have been investigated in various settings \cite{CE2020, L et al 2019, FenzlLambert, Charlier 2d jumps, ChLe2022}: see \cite[Proposition 8.1]{CE2020} for second order asymptotics of the one-point MGF of counting statistics of general domains in Ginibre-type ensembles; see \cite{L et al 2019} for second order asymptotics of the one-point MGF of the disk counting statistics of rotation-invariant ensembles with a general potential; see \cite{FenzlLambert} for third order asymptotics for the one-point MGF of disk counting statistics of Ginibre-type ensembles; and see \cite{Charlier 2d jumps, ChLe2022} for fourth order asymptotics for the $m$-point MGF of disk counting statistics in the Mittag-Leffler ensemble \eqref{def of point process}. Both the bulk and the soft edge regimes were investigated in \cite{Charlier 2d jumps, ChLe2022}; however in \cite{Charlier 2d jumps} the radii of the disks were taken fixed, while in \cite{ChLe2022} all radii were assumed to merge at the critical speed $\sim \frac{1}{\sqrt{n}}$ (in this critical regime one observes non-trivial correlations in the disk counting statistics). As it turns out, the bulk statistics of \eqref{def of point process} and \eqref{def of point process hard} are identical up to exponentially small errors (in other words, the points in the bulk almost do not feel the hard wall). Our formulas for the bulk regime \eqref{def of rell bulk} are in fact \textit{identical} to the corresponding formulas in \cite{ChLe2022} (the proof is also almost identical, we only have to handle some additional exponentially small error terms). We have nevertheless decided to include a very short section in this paper on the bulk regime for completeness. We also point out that for $C^2$-smooth test functions $f$ on the plane, the asymptotic normality of fluctuations was worked out quite generally in \cite{AM}, for potentials having a connected droplet. In this case the asymptotic variance of fluctuations is
given by a Dirichlet norm $\frac 1 {4\pi}\int|\nabla f^S(z)|^2\, d^2 z$, where $f^S$ equals $f$ in $S$ and is the bounded harmonic extension of $f|_S$ outside of $S$.

The presentation of our results is organized as follows: Subsection \ref{subsection: hard edge regime} treats the hard edge regime, Subsection \ref{subsection: semi-hard edge regime} the semi-hard edge regime, and Subsection \ref{subsection: bulk regime} the bulk regime.

\subsection{Results for the hard edge regime}\label{subsection: hard edge regime}
Let $r_{1},\dots,r_{m}$ be as in \eqref{def of rell hard}, let $\vec{t}:=(t_{1},\dots,t_{m})$ be such that $t_{1}>\dots>t_{m}\geq 0$, let $\vec{u}:=(u_{1},\dots,u_{m}) \in \mathbb{R}^{m}$, and define
\begin{align}
& f(x;\vec{t},\vec{u}) = - \bigg(\frac{b\rho^{2b}}{x-b\rho^{2b}} + \frac{\alpha}{b} \bigg)\frac{\mathsf{T}_{1}(x;\vec{t},\vec{u})}{1 + \mathsf{T}_{0}(x;\vec{t},\vec{u})}-\frac{x}{2b} \frac{\mathsf{T}_{2}(x;\vec{t},\vec{u})}{1 + \mathsf{T}_{0}(x;\vec{t},\vec{u})}, \label{def of f hard}\\
& \mathsf{T}_{j}(x;\vec{t},\vec{u}) = \sum_{\ell=1}^{m} \omega_{\ell}t_{\ell}^{j}e^{-\frac{t_{\ell}}{b}(x-b\rho^{2b})}, \qquad j \geq 0, \label{def of T hard}
	\\
& \Omega(\vec{u}) = 1 + \mathsf{T}_{0}(b \rho^{2b};\vec{t},\vec{u}) = e^{u_{1}+\dots+u_{m}}, \label{def of Omega intro}
\end{align}
where
\begin{align}\label{def of Omega j intro}
\omega_{\ell} = \omega_{\ell}(\vec{u}) = \begin{cases}
e^{u_{\ell}+\dots+u_{m}}-e^{u_{\ell+1}+\dots+u_{m}}, & \mbox{if } \ell < m, \\
e^{u_{m}}-1, & \mbox{if } \ell=m, \\
1, & \mbox{if } \ell=m+1.
\end{cases}
\end{align}
Recall that the complementary error function is defined by
\begin{align}\label{def of erfc}
\mathrm{erfc} (t) = \frac{2}{\sqrt{\pi}}\int_{t}^{\infty} e^{-x^{2}}dx.
\end{align}
Throughout the paper $\ln(\cdot)$ denotes the principal branch of the logarithm and $D_\delta(z_0) = \{z \in \mathbb{C} : |z - z_0| < \delta\}$ denotes an open disk of radius $\delta$ centered at $z_0 \in \mathbb{C}$.

\begin{theorem}\label{thm:main thm hard}(Merging radii at the hard edge)

\noindent Let $m \in \mathbb{N}_{>0}$,  $b>0$, $\rho \in (0,b^{-\frac{1}{2b}})$, $t_{1}>\dots>t_{m} \geq 0$, and $\alpha > -1$ be fixed parameters, and for $n \in \mathbb{N}_{>0}$, define
\begin{align}\label{rellhardedge}
r_{\ell} = \rho \bigg( 1-\frac{t_{\ell}}{n} \bigg)^{\frac{1}{2b}}, \qquad \ell=1,\dots,m.
\end{align}
For any fixed $x_{1},\dots,x_{m} \in \mathbb{R}$, there exists $\delta > 0$ such that
\begin{align}\label{asymp in main thm hard}
\mathbb{E}\bigg[ \prod_{j=1}^{m} e^{u_{j}\mathrm{N}(r_{j})} \bigg] = \exp \bigg( C_{1} n + C_{2} \ln n + C_{3} +  \frac{C_{4}}{\sqrt{n}} + \bigO\big(n^{-\frac{3}{5}}\big)\bigg), \qquad \mbox{as } n \to + \infty
\end{align}
uniformly for $u_{1} \in D_\delta(x_1),\dots,u_{m} \in D_\delta(x_m)$, where $\{C_{j}=C_{j}(\vec{u})\}_{j=1}^{4}$ are given by
\begin{align*}
C_{1} = &\; b \rho^{2b} \sum_{j=1}^{m}u_{j} + \int_{b\rho^{2b}}^{1} \ln(1+\mathsf{T}_{0}(x;\vec{t},\vec{u}))dx,
	\\
C_{2} = & - \frac{b\rho^{2b}}{2} \frac{\mathsf{T}_{1}(b\rho^{2b};\vec{t},\vec{u})}{\Omega(\vec{u})} = - \frac{b\rho^{2b}}{2}\frac{\sum_{\ell=1}^{m}t_{\ell}\omega_{\ell}}{e^{u_{1}+\dots+u_{m}}},
	\\
C_{3} = & - \frac{1}{2} \sum_{j=1}^{m}u_{j} + \frac{1}{2} \ln\big(1+\mathsf{T}_{0}(1;\vec{t},\vec{u})\big) + \int_{b\rho^{2b}}^{1} \bigg\{ f(x;\vec{t},\vec{u}) + \frac{b \rho^{2b} \mathsf{T}_{1}(b\rho^{2b};\vec{t},\vec{u})}{\Omega(\vec{u}) (x-b\rho^{2b})} \bigg\}dx
	\\
& + b\rho^{2b} \frac{\mathsf{T}_{1}(b\rho^{2b};\vec{t},\vec{u})}{\Omega(\vec{u})} \ln \bigg( \frac{b\rho^{b}}{\sqrt{2\pi}(1-b\rho^{2b})} \bigg),
	\\
C_{4} = &\; \sqrt{2} \,\mathcal{I} \,b\rho^{b}\bigg( \rho^{2b} \frac{ \mathsf{T}_{2}(b\rho^{2b};\vec{t},\vec{u}) }{ \Omega(\vec{u}) } - \frac{ \mathsf{T}_{1}(b\rho^{2b};\vec{t},\vec{u}) }{ \Omega(\vec{u}) } - \rho^{2b} \frac{ \mathsf{T}_{1}(b\rho^{2b};\vec{t},\vec{u})^{2} }{ \Omega(\vec{u})^{2} } \bigg),
\end{align*}
and the real number $\mathcal{I} \in \mathbb{R}$ is given by
\begin{align}
& \mathcal{I} = \int_{-\infty}^{+\infty} \bigg\{ \frac{y\, e^{-y^{2}}}{\sqrt{\pi}\, \mathrm{erfc}(y)} - \chi_{(0,+\infty)}(y) \bigg[ y^{2}+\frac{1}{2} \bigg] \bigg\}dy \approx -0.81367. \label{def of I}
\end{align}
In particular, since $\mathbb{E}\big[ \prod_{j=1}^{m} e^{u_{j}\mathrm{N}(r_{j})} \big]$ depends analytically on $u_{1},\dots,u_{m} \in \mathbb{C}$ and is strictly positive for $u_{1},\dots,u_{m} \in \mathbb{R}$, the asymptotic formula \eqref{asymp in main thm hard} together with Cauchy's formula shows that
\begin{align}\label{der of main result hard}
\partial_{u_{1}}^{k_{1}}\dots \partial_{u_{m}}^{k_{m}} \bigg\{ \ln \mathbb{E}\bigg[ \prod_{j=1}^{m} e^{u_{j}\mathrm{N}(r_{j})} \bigg] - \bigg( C_{1} n + C_{2} \ln n + C_{3} +  \frac{C_{4}}{\sqrt{n}} \bigg) \bigg\} = \bigO\big(n^{-\frac{3}{5}}\big), \quad \mbox{as } n \to + \infty,
\end{align}
for any $k_{1},\dots,k_{m}\in \mathbb{N}$, and $u_{1},\dots,u_{m}\in \mathbb{R}$.
\end{theorem}
\begin{remark}\label{remark:C1 hard edge}
The leading coefficient in the asymptotics of moment generating functions of linear statistics with respect to a fixed, bounded continuous test function $g$ is of course given by the integral of $g$ against the relevant equilibrium measure.
However, in the hard edge regime of Theorem \ref{thm:main thm hard}, we rather use a sequence $g=g_n$ of test-functions, given in terms of characteristic functions of expanding discs  of radii
\eqref{rellhardedge} by
$g_n(z) = \sum_{j=1}^{m}u_{j} \chi_{(0,r_{j})}(z)$.

A direct computation using \eqref{eq mes} shows that, as $n \to + \infty$,
\begin{align*}
\int g_n(x) d\mu_{h}(x) = \begin{cases}
\sum_{j=1}^{m}u_{j} \int_{0}^{r_{j}} 2b^{2}r^{2b-1}dr = b\rho^{2b} \sum_{j=1}^{m}u_{j} +o(1), & \mbox{if } t_{m}>0, \\
\sum_{j=1}^{m}u_{j} \int_{0}^{r_{j}} 2b^{2}r^{2b-1}dr + u_{m} c_{\rho} = b\rho^{2b} \sum_{j=1}^{m}u_{j} + u_{m} c_{\rho} +o(1), & \mbox{if } t_{m}=0,
\end{cases}
\end{align*}
where $c_\rho$ is given by \eqref{crho}.

Since $b\rho^{2b} \sum_{j=1}^{m}u_{j} \neq C_{1} \neq b\rho^{2b} \sum_{j=1}^{m}u_{j} + u_{m} c_{\rho}$, we see that in the hard edge regime, even the leading coefficient $C_{1}$ cannot straightforwardly be obtained from the equilibrium measure, which might be surprising at first sight. 
What is even more surprising is that $C_{1}$ is not even linear in $u_{1},\ldots,u_{m}$ (this contrasts with all previously studied regimes, and also with the semi-hard edge regime).
\end{remark}

 For $\vec{j} \in (\mathbb{N}^{m})_{>0} := \{\vec{j}=(j_{1},\dots,j_{m}) \in \mathbb{N}: j_{1}+\dots+j_{m}\geq 1\}$, the joint cumulant $\kappa_{\vec{j}}=\kappa_{\vec{j}}(r_{1},\dots,r_{m};n,b,\alpha)$ of $\mathrm{N}(r_{1}), \dots, \mathrm{N}(r_{m})$  is defined by
\begin{align}\label{joint cumulant}
\kappa_{\vec{j}}=\kappa_{j_{1},\dots,j_{m}}:=\partial_{\vec{u}}^{\vec{j}} \ln \mathbb{E}[e^{u_{1}\mathrm{N}(r_{1})+\dots + u_{m}\mathrm{N}(r_{m})}] \Big|_{\vec{u}=\vec{0}},
\end{align}
where $\partial_{\vec{u}}^{\vec{j}}:=\partial_{u_{1}}^{j_{1}}\dots \partial_{u_{m}}^{j_{m}}$. In particular,
\begin{align*}
\mathbb{E}[\mathrm{N}(r)] = \kappa_{1}(r), \qquad \mbox{Var}[\mathrm{N}(r)] = \kappa_{2}(r) =\kappa_{(1,1)}(r,r), \qquad \mbox{Cov}[\mathrm{N}(r_{1}),\mathrm{N}(r_{2})] = \kappa_{(1,1)}(r_{1},r_{2}).
\end{align*}
Recall from \eqref{eq mes}--\eqref{crho} that $c_{\rho} = 1-b\rho^{2b} = \int\mu_{\mathrm{sing}}(d^{2}z)$, i.e. $c_{\rho}$ is the density of particles accumulating near the hard-edge as $n\to +\infty$. It turns out that the asymptotics of $\mathbb{E}[\mathrm{N}(r_{\ell})]$ and $\mathrm{Cov}(\mathrm{N}(r_{\ell}),\mathrm{N}(r_{k}))$, which are obtained in Corollary \ref{coro:correlation hard} below, are more elegantly described in terms of $c_{\rho}$, as well as the new parameter
\begin{align}\label{def of sell hard edge}
& s_{\ell} := \frac{t_{\ell}}{b}(1-b\rho^{2b}) =  \frac{c_{\rho}n}{b} \bigg( 1- \bigg( \frac{r_{\ell}}{\rho} \bigg)^{2b} \bigg) = 2 \cdot \frac{c_{\rho}n}{2\pi \rho} \cdot 2\pi (\rho-r_{\ell}) \big( 1+\bigO(n^{-1}) \big).
\end{align}
\begin{corollary}[Hard edge]\label{coro:correlation hard}
Let $m \in \mathbb{N}_{>0}$, $b>0$, $\rho \in (0,b^{-\frac{1}{2b}})$, $\vec{j} \in (\mathbb{N}^{m})_{>0}$, $\alpha > -1$,  and $t_{1}>\dots>t_{m} > 0$ be fixed. Define $s_{1},\ldots,s_{m}$ as in \eqref{def of sell hard edge}. For $n \in \mathbb{N}_{>0}$, define $\{r_\ell\}_{\ell =1}^m$ by \eqref{rellhardedge}.

(a) The joint cumulant $\kappa_{\vec{j}}$ satisfies
\begin{align}\label{asymp cumulant hard edge}
\kappa_{\vec{j}} = \partial_{\vec{u}}^{\vec{j}}C_{1}\big|_{\vec{u}=\vec{0}} \; n + \partial_{\vec{u}}^{\vec{j}}C_{2}\big|_{\vec{u}=\vec{0}}  \;\ln{n} + \partial_{\vec{u}}^{\vec{j}}C_{3}\big|_{\vec{u}=\vec{0}} +  \frac{\partial_{\vec{u}}^{\vec{j}}C_{4}\big|_{\vec{u}=\vec{0}}}{\sqrt{n}} + \bigO\big(n^{-\frac{3}{5}}\big), \qquad  n \to +\infty,
\end{align}
where $C_{1},\dots,C_{4}$ are as in Theorem \ref{thm:main thm hard}. In particular, for any $1 \leq \ell < k \leq m$,
\begin{align*}
& \mathbb{E}[\mathrm{N}(r_{\ell})] = b_1(s_\ell) n + c_1(s_\ell) \ln{n} + d_1(s_\ell) + e_1(s_\ell) n^{-\frac{1}{2}} + \bigO\big(n^{-\frac{3}{5}}\big),
 	\\
& \mathrm{Var}[\mathrm{N}(r_{\ell})] = b_{(1,1)}(s_{\ell},s_{\ell})n + c_{(1,1)}(s_{\ell},s_{\ell})\ln{n} + d_{(1,1)}(s_{\ell},s_{\ell}) + e_{(1,1)}(s_{\ell},s_{\ell})n^{-\frac{1}{2}} + \bigO\big(n^{-\frac{3}{5}}\big),
	 \\
& \mathrm{Cov}(\mathrm{N}(r_{\ell}),\mathrm{N}(r_{k})) = b_{(1,1)}(s_{\ell},s_{k})n + c_{(1,1)}(s_{\ell},s_{k})\ln{n} + d_{(1,1)}(s_{\ell},s_{k}) + e_{(1,1)}(s_{\ell},s_{k})n^{-\frac{1}{2}} + \bigO\big(n^{-\frac{3}{5}}\big) \nonumber
\end{align*}
as $n \to + \infty$, where
\begin{align}\nonumber
 b_1(s_\ell) =&\; 1-c_{\rho} + c_{\rho} \frac{1 - e^{-s_{\ell}}}{s_{\ell}},
\qquad c_1(s_\ell) = - \frac{1-c_{\rho}}{c_{\rho}} \frac{b s_{\ell}}{2},
	\\ \nonumber
d_1(s_\ell) = & -\frac{1 - e^{-s_\ell } }{2} +  \frac{1-c_{\rho}}{c_{\rho}} \frac{b s_{\ell}}{2} \ln\bigg(\frac{b(1-c_{\rho})}{2\pi c_{\rho}^{2}}\bigg)
	\\ \nonumber
&- s_{\ell} \int_{0}^{1}  \frac{e^{-s_{\ell} y}\big( y c_{\rho}(bs_{\ell}y+2\alpha) +(1-c_{\rho})b (2+s_{\ell} y) \big)-2(1-c_{\rho})b }{2c_{\rho}y} dy,	
	\\ \nonumber
 e_1(s_\ell) =&\; \sqrt{2} \, \mathcal{I} b \rho^{-b} \frac{1-c_{\rho}}{c_{\rho}} s_{\ell}\bigg( \frac{1-c_{\rho}}{c_{\rho}}s_{\ell}-1 \bigg),
\end{align}
and, for $l \leq k$,
\begin{align}\label{def of b11 hard edge}
b_{(1,1)}(s_{\ell},s_{k}) = &\; c_{\rho}\frac{1 - e^{-s_{\ell}} }{s_{\ell}} - c_{\rho}\frac{1 - e^{-s_{\ell} - s_k} }{s_{\ell} + s_k},
\qquad c_{(1,1)}(s_{\ell},s_{k}) = \frac{1-c_{\rho}}{c_{\rho}} \frac{b s_k}{2},
	 \\\nonumber
 d_{(1,1)}(s_{\ell},s_{k}) = &\; \frac{e^{-s_\ell}(1 - e^{-s_k})}{2}
 - \frac{1-c_{\rho}}{c_{\rho}} \frac{b s_{k}}{2} \ln\bigg(\frac{b(1-c_{\rho})}{2\pi c_{\rho}^{2}}\bigg)
	\\\nonumber
& - \int_{0}^1
\frac{1}{y}\bigg\{b s_{k} \frac{1-c_{\rho}}{c_{\rho}}+s_\ell e^{-s_\ell y} \bigg(b \frac{1-c_{\rho}}{c_{\rho}} + \alpha y + \frac{b s_{\ell}}{2}y \bigg( y+\frac{1-c_{\rho}}{c_{\rho}} \bigg) \bigg)
	\\\nonumber
& - e^{-(s_\ell+s_k)y} \bigg( \bigg( \frac{1-c_{\rho}}{c_{\rho}}b+\alpha y \bigg)(s_{\ell}+s_{k}) + \frac{b y}{2}\bigg( y + \frac{1-c_{\rho}}{c_{\rho}} \bigg) (s_{\ell}^{2}+s_{k}^{2}) \bigg) \bigg\} dy,
	\\\nonumber
e_{(1,1)}(s_{\ell},s_{k}) = &\; \sqrt{2} \, \mathcal{I} b \rho^{-b} \frac{1-c_{\rho}}{c_{\rho}} s_{k} \bigg( 1-\frac{1-c_{\rho}}{c_{\rho}}(2 s_{\ell} + s_{k}) \bigg).
\end{align}

(b) As $n \to + \infty$, the random variable $(\mathcal{N}_{1},\dots,\mathcal{N}_{m})$, where
\begin{align}
& \mathcal{N}_{\ell} := \frac{\mathrm{N}(r_{\ell})-b_1(s_\ell) n}{\sqrt{b_{(1,1)}(s_\ell,s_\ell) n}}, \qquad \ell=1,\dots,m,
\label{Nj hard edge}
\end{align}
convergences in distribution to a multivariate normal random variable of mean $(0,\dots,0)$ whose covariance matrix $\Sigma$ is defined by
\begin{align*}
\Sigma_{\ell,k} =  \Sigma_{k, \ell} = \frac{b_{(1,1)}(s_{\ell},s_{k})}{\sqrt{b_{(1,1)}(s_{\ell},s_{\ell})b_{(1,1)}(s_{k},s_{k})}}, \qquad 1 \leq \ell \leq k \leq m,
\end{align*}
where $b_{(1,1)}$ is given by \eqref{def of b11 hard edge}.
\end{corollary}

\begin{remark}\label{remark:N=n with prob 1}
Corollary \ref{coro:correlation hard} is stated for $t_{1}>\dots>t_{m} > 0$. It is important for Corollary \ref{coro:correlation hard} (b) that $t_{m}>0$; note however that Corollary \ref{coro:correlation hard} (a) in fact also holds for $t_{1}>\dots>t_{m} \geq 0$. In the case when $t_{m}=0=s_{m}$, one finds $b_{1}(s_{m})=n$ and $c_{1}(s_{m})=d_{1}(s_{m})=e_{1}(s_{m})=0$, which is consistent with the fact that $\mathrm{N}(r_{m})=n$ with probability $1$.

The central limit theorem of Corollary \ref{coro:correlation hard} (b), even though it only uses $b_{1}(s)$ and $b_{(1,1)}(s,s)$, is a non-trivial result because to determine just the leading
term $C_{1}$ in Theorem \ref{thm:main thm hard} one already needs quite subtle asymptotics of the incomplete gamma function.
\end{remark}

\begin{proof}[Proof of Corollary \ref{coro:correlation hard}]
Assertion (a) follows from (\ref{der of main result hard}) and the expressions for the $C_j$ given in Theorem \ref{thm:main thm hard}.
By L\'evy's continuity theorem, assertion (b) will follow if we can show that the characteristic function $\mathbb{E}[e^{i \sum_{\ell = 1}^m v_\ell \mathcal{N}_\ell}]$ converges pointwise to $e^{-\frac{1}{2}\sum_{\ell, k=1}^m v_\ell \Sigma_{\ell,k} v_k}$ for every $v_\ell \in \mathbb{R}^m$ as $n \to +\infty$. Letting $u_\ell = \frac{i v_\ell }{\sqrt{b_{(1,1)}(s_\ell,s_\ell) n}}$, (\ref{Nj hard edge}) and (\ref{asymp in main thm hard}) show that
\begin{align*}
\mathbb{E}[e^{i \sum_{\ell = 1}^m v_\ell \mathcal{N}_\ell}]
& = \mathbb{E}[e^{\sum_{\ell = 1}^m u_\ell \mathrm{N}(r_{\ell})}]
e^{- \sum_{\ell = 1}^m u_\ell b_1(s_\ell) n}
	\\
& = e^{C_{1}(\vec{u}) n + C_{2}(\vec{u}) \ln n + C_{3}(\vec{u}) + \bigO(n^{-\frac{1}{2}})}
e^{- \sum_{\ell = 1}^m u_\ell \partial_{u_\ell} C_1|_{\vec{u}=\vec{0}} n }
\end{align*}
as $n \to +\infty$ for any fixed $v_\ell \in \mathbb{R}^m$. Since $C_j|_{\vec{u}=\vec{0}} = 0$ for $j = 1,2,3$ and $u_\ell = \bigO(n^{-1/2})$, we obtain
\begin{align*}
\mathbb{E}[&e^{i \sum_{\ell = 1}^m v_\ell \mathcal{N}_\ell}]
 =  e^{\frac{1}{2}\sum_{\ell,k= 1}^m u_\ell u_k \partial_{u_\ell}\partial_{u_k} C_1|_{\vec{u}=\vec{0}} n
+ \bigO(|\vec{u}|^3 n + |\vec{u}| \ln{n} + |\vec{u}| + n^{-1/2})}
	\\
& =  e^{\frac{1}{2}\sum_{\ell,k = 1}^m \frac{iv_\ell}{\sqrt{b_{(1,1)}(s_\ell,s_\ell)}} \frac{iv_k}{\sqrt{b_{(1,1)}(s_k,s_k)}} b_{(1,1)}(s_{\min(\ell,k)}, s_{\max(\ell,k)}) + \bigO( \frac{\ln n}{\sqrt{n}})}
\to  e^{-\frac{1}{2}\sum_{\ell, k=1}^m v_\ell \Sigma_{\ell,k} v_k}
\end{align*}
as $n \to +\infty$, which proves (b).
\end{proof}

Let us analyze the leading coefficient $b_{(1,1)}(s,s)$ of $\mathrm{Var}[\mathrm{N}(r)]$, where $r:=\rho \big( 1-\frac{t}{n} \big)^{\frac{1}{2b}}$ and $s:=\frac{t}{b}c_{\rho}$. By \eqref{def of b11 hard edge},
\begin{align}\label{leading coeff var}
b_{(1,1)}(s,s) = c_{\rho}\frac{1 - e^{-s } }{s} - c_{\rho}\frac{1 - e^{-2s} }{2s}.
\end{align}
Note that $b_{(1,1)}(0,0):=\lim_{s \to 0_{+}}b_{(1,1)}(s,s)=0$, which, as mentioned in Remark \ref{remark:N=n with prob 1}, is consistent with the fact that $\mathrm{N}(\rho)=n$ with probability $1$. On the other hand, $b_{(1,1)}(s,s) = \frac{c_{\rho}}{2s}+\bigO(e^{-s})$ as $s \to + \infty$. It is therefore interesting to investigate where the maximum of $b_{(1,1)}(s,s)$ is achieved. It is possible to compute the unique maximum of $s\mapsto b_{(1,1)}(s,s)$ explicitly in terms of the Lambert function $W_{-1}(x)$, which for $-\frac{1}{e} \leq x < 0$ is defined as the unique solution to
\begin{align*}
& W_{-1}(x)e^{W_{-1}(x)} = x, \qquad W_{-1}(x) \leq -1.
\end{align*}
Indeed, taking the derivative of \eqref{leading coeff var} yields
\begin{align*}
\frac{d}{ds} b_{(1,1)}(s,s) = -\frac{c_{\rho}}{2s^{2}} \big( 1-e^{-s} \big) \big( 1-(1+2s)e^{-s} \big), \qquad s>0,
\end{align*}
and a direct inspection shows that $\frac{d}{ds} b_{(1,1)}(s,s)=0$ if and only if $s=s_{\star}$, where
\begin{align*}
s_{\star} = -\big( W_{-1}(\tfrac{-1}{2\sqrt{e}}) + \tfrac{1}{2} \big)  \approx 1.2564 .
\end{align*}
Furthermore,
\begin{align*}
b_{(1,1)}(s_{\star},s_{\star}) = \frac{-2 \, W_{-1}(\tfrac{-1}{2\sqrt{e}})-1}{4 \, W_{-1}(\tfrac{-1}{2\sqrt{e}})^{2}} c_{\rho} \approx 0.20363 c_{\rho}.
\end{align*}
As $\rho$ decreases, the hard wall gets stronger (in the sense that the mass $c_{\rho}$ of $\mu_{\mathrm{sing}}$ increases), and we observe that $b_{(1,1)}(s_{\star},s_{\star})$ increases. The graphs of $b_{1}(s)$ and $b_{(1,1)}(s,s)$ are displayed in Figure \ref{fig: b11} for certain values of $\rho$ and $b$.
\begin{figure}[h]
\begin{center}
\begin{tikzpicture}
\node at (0,0) {\includegraphics[width=6.5cm]{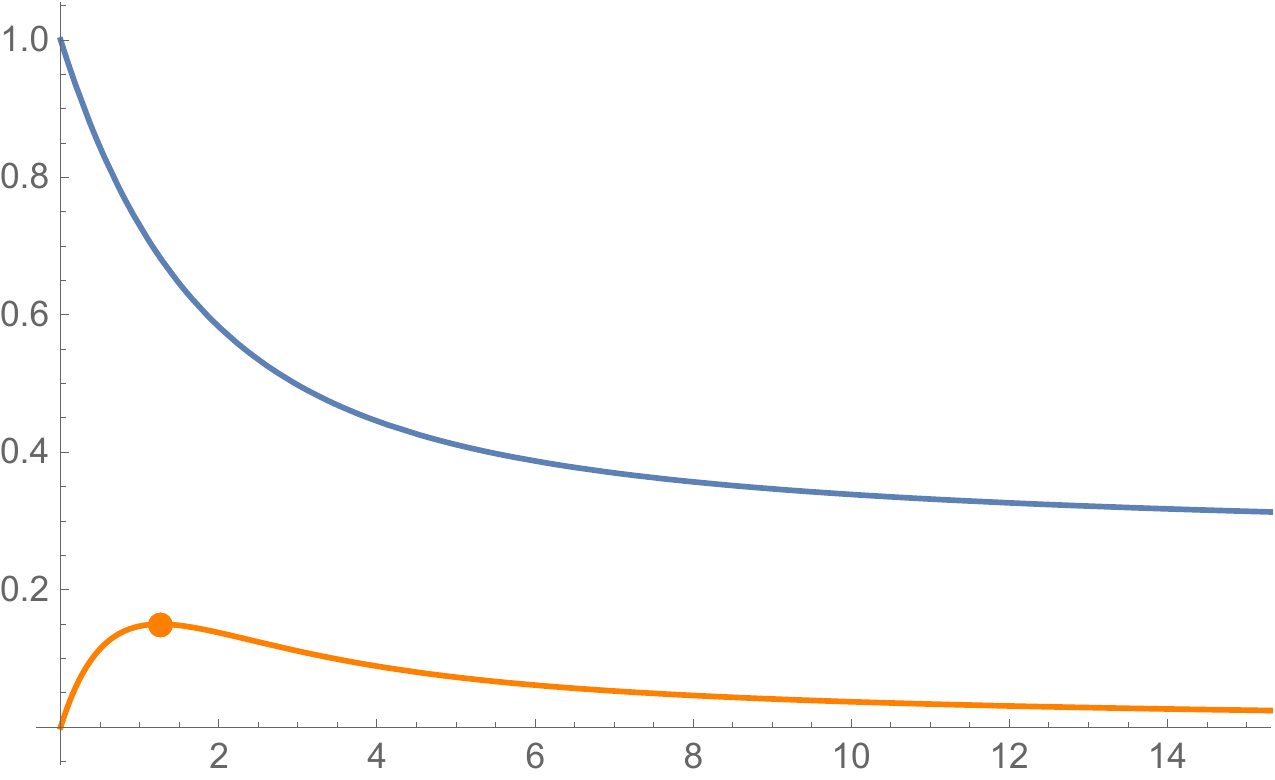}};
\end{tikzpicture}
\end{center}
\vspace{-.5cm}
\caption{\label{fig: b11} The coefficients $s \mapsto b_{1}(s)$ (blue) and $s \mapsto b_{(1,1)}(s,s)$ (orange) for $\rho=0.6b^{-\frac{1}{2b}}$ and $b=\frac{13}{10}$. The orange dot has coordinates $(s_{\star},b_{(1,1)}(s_{\star},s_{\star}))$. }
\end{figure}

\subsection{Results for the semi-hard edge}\label{subsection: semi-hard edge regime}

\begin{theorem}\label{thm:main thm semi-hard}(Merging radii at the semi-hard edge)

\noindent Let $m \in \mathbb{N}_{>0}$, $b>0$, $\rho \in (0,b^{-\frac{1}{2b}})$, $\mathfrak{s}_{1}>\dots>\mathfrak{s}_{m}>0$, and $\alpha > -1$ be fixed parameters, and for $n \in \mathbb{N}_{>0}$, define
\begin{align}\label{rellsemihardedge}
r_{\ell} = \rho \bigg( 1-\frac{\sqrt{2}\, \mathfrak{s}_{\ell}}{\rho^{b}\sqrt{n}} \bigg)^{\frac{1}{2b}}, \qquad \ell=1,\dots,m.
\end{align}
For any fixed $x_{1},\dots,x_{m} \in \mathbb{R}$, there exists $\delta > 0$ such that
\begin{align}\label{asymp in main thm semi-hard}
\mathbb{E}\bigg[ \prod_{j=1}^{m} e^{u_{j}\mathrm{N}(r_{j})} \bigg] = \exp \bigg( C_{1} n + C_{2} \sqrt{n} + C_{3} +  \frac{C_{4}}{\sqrt{n}} + \bigO\bigg(\frac{(\ln n)^{4}}{n}\bigg)\bigg), \qquad \mbox{as } n \to + \infty
\end{align}
uniformly for $u_{1} \in D_\delta(x_1),\dots,u_{m} \in D_\delta(x_m)$, where
\begin{align*}
& C_{1} = b \rho^{2b} \sum_{j=1}^{m}u_{j}, \\
& C_{2} = \sqrt{2} \, b \rho^{b} \int_{-\infty}^{+\infty} \Big( h_{0}(y)-\chi_{(-\infty,0)}(y) \sum_{j=1}^{m}u_{j} \Big)dy, \\
& C_{3} = -\bigg(\frac{1}{2}+\alpha \bigg)\sum_{j=1}^{m}u_{j} + b \int_{-\infty}^{+\infty} \Big( 4y \big( h_{0}(y) - \chi_{(-\infty,0)}(y) \sum_{j=1}^{m}u_{j} \big) + \sqrt{2} \, h_{1}(y) \Big)dy, \\
& C_{4} = b\rho^{-b} \int_{-\infty}^{+\infty} \bigg[ 6\sqrt{2} y^{2} \bigg( h_{0}(y)-\chi_{(-\infty,0)}(y) \sum_{j=1}^{m}u_{j} \bigg) + 4yh_{1}(y)+\sqrt{2} \, h_{2}(y) \bigg]dy,
\end{align*}
where
\begin{align*}
& h_{0}(y) = \ln (g_{0}(y)), \qquad h_{1}(y) = \frac{g_{1}(y)}{g_{0}(y)}, \qquad h_{2}(y) = \frac{g_{2}(y)}{g_{0}(y)} - \frac{1}{2} \bigg( \frac{g_{1}(y)}{g_{0}(y)} \bigg)^{2},
\end{align*}
and
\begin{align*}
g_{0}(y) = &\; 1+\sum_{\ell=1}^{m} \omega_{\ell} \frac{\mathrm{erfc}(y+\mathfrak{s}_{\ell})}{\mathrm{erfc}(y)}, \\
 g_{1}(y) = &\; \sum_{\ell=1}^{m} \frac{\sqrt{2}}{3\sqrt{\pi}} \omega_{\ell} \bigg\{ (5y^{2}-1)\frac{e^{-y^{2}}}{\mathrm{erfc}(y)} \frac{\mathrm{erfc}(y+\mathfrak{s}_{\ell})}{\mathrm{erfc}(y)} - \big( 5y^{2}+\mathfrak{s}_{\ell} y + 2\mathfrak{s}_{\ell}^{2}-1 \big) \frac{e^{-(y+\mathfrak{s}_{\ell})^{2}}}{\mathrm{erfc}(y)} \bigg\}, \\
g_{2}(y) = &\; \sum_{\ell=1}^{m} \omega_{\ell}\bigg\{ \frac{1}{18\sqrt{\pi}} \bigg[ 50 y^{5} + 70 y^{4} \mathfrak{s}_{\ell} + y^{3}(62\mathfrak{s}_{\ell}^{2}-73) + y^{2}\mathfrak{s}_{\ell}(50\mathfrak{s}_{\ell}^{2}-33) \\
& -y (3+18\mathfrak{s}_{\ell}^{2}-16\mathfrak{s}_{\ell}^{4}) + \mathfrak{s}_{\ell}(3-22\mathfrak{s}_{\ell}^{2}+8\mathfrak{s}_{\ell}^{4}) \bigg] \frac{e^{-(y+\mathfrak{s}_{\ell})^{2}}}{\mathrm{erfc}(y)} \\
& + \frac{2(1-5y^{2})(5y^{2}+y \mathfrak{s}_{\ell} -1 + 2\mathfrak{s}_{\ell}^{2})}{9\pi} \frac{e^{-y^{2}}}{\mathrm{erfc}(y)} \frac{e^{-(y+\mathfrak{s}_{\ell})^{2}}}{\mathrm{erfc}(y)} \\
& + \frac{y(3+73y^{2}-50y^{4})}{18 \sqrt{\pi}} \frac{e^{-y^{2}}}{\mathrm{erfc}(y)} \frac{\mathrm{erfc}(y+\mathfrak{s}_{\ell})}{\mathrm{erfc}(y)} + \frac{2(1-5y^{2})^{2}}{9\pi} \bigg( \frac{e^{-y^{2}}}{\mathrm{erfc}(y)} \bigg)^{2} \frac{\mathrm{erfc}(y+\mathfrak{s}_{\ell})}{\mathrm{erfc}(y)} \bigg\}.
\end{align*}
In particular, since $\mathbb{E}\big[ \prod_{j=1}^{m} e^{u_{j}\mathrm{N}(r_{j})} \big]$ depends analytically on $u_{1},\dots,u_{m} \in \mathbb{C}$ and is strictly positive for $u_{1},\dots,u_{m} \in \mathbb{R}$, the asymptotic formula \eqref{asymp in main thm} together with Cauchy's formula shows that
\begin{align}\label{der of main result semi-hard}
\partial_{u_{1}}^{k_{1}}\dots \partial_{u_{m}}^{k_{m}} \bigg\{ \ln \mathbb{E}\bigg[ \prod_{j=1}^{m} e^{u_{j}\mathrm{N}(r_{j})} \bigg] - \bigg( C_{1} n + C_{2} \sqrt{n} + C_{3} +  \frac{C_{4}}{\sqrt{n}} \bigg) \bigg\} = \bigO\bigg(\frac{(\ln n)^{4}}{n}\bigg)\bigg)
\end{align}
as $n \to + \infty$, for any $k_{1},\dots,k_{m}\in \mathbb{N}$ and $u_{1},\dots,u_{m}\in \mathbb{R}$.
\end{theorem}

The proof of the following corollary is similar to that of Corollary \ref{coro:correlation hard} and is omitted.

\begin{corollary}[Semi-hard edge]\label{coro:correlation semihard}
Let $m \in \mathbb{N}_{>0}$, $b>0$, $\rho \in (0,b^{-\frac{1}{2b}})$, $\vec{j} \in (\mathbb{N}^{m})_{>0}$, $\alpha > -1$,  and $\mathfrak{s}_{1}>\dots>\mathfrak{s}_{m}>0$ be fixed. For $n \in \mathbb{N}_{>0}$, define $\{r_\ell\}_{\ell =1}^m$ by \eqref{rellsemihardedge}.

(a) The joint cumulant $\kappa_{\vec{j}}$ satisfies
\begin{align}\label{asymp cumulant semihard edge}
\kappa_{\vec{j}} = \begin{cases}
\ds \partial_{\vec{u}}^{\vec{j}}C_{1}\big|_{\vec{u}=\vec{0}} n + \partial_{\vec{u}}^{\vec{j}}C_{2}\big|_{\vec{u}=\vec{0}} \sqrt{n} + \partial_{\vec{u}}^{\vec{j}}C_{3}\big|_{\vec{u}=\vec{0}} + \partial_{\vec{u}}^{\vec{j}}C_{4}\big|_{\vec{u}=\vec{0}}\frac{1}{\sqrt{n}} + \bigO\bigg(\frac{(\ln n)^{4}}{n}\bigg), & \mbox{if } \vec{j}=1, \\
\ds \hspace{2.07cm} \partial_{\vec{u}}^{\vec{j}}C_{2}\big|_{\vec{u}=\vec{0}} \sqrt{n} + \partial_{\vec{u}}^{\vec{j}}C_{3}\big|_{\vec{u}=\vec{0}} +  \partial_{\vec{u}}^{\vec{j}}C_{4}\big|_{\vec{u}=\vec{0}}\frac{1}{\sqrt{n}} + \bigO\bigg(\frac{(\ln n)^{4}}{n}\bigg), & \mbox{otherwise},
\end{cases}
\end{align}
as $n \to +\infty$, where $C_{1},\dots,C_{4}$ are as in Theorem \ref{thm:main thm semi-hard}. In particular, for any $1 \leq \ell < k \leq m$,
\begin{align*}
& \mathbb{E}[\mathrm{N}(r_{\ell})] = b_1(\mathfrak{s}_\ell) n + c_1(\mathfrak{s}_\ell) \sqrt{n} + d_1(\mathfrak{s}_\ell) + e_1(\mathfrak{s}_\ell) n^{-\frac{1}{2}} + \bigO\big((\ln n)^{4}n^{-1}\big),
 	\\
& \mathrm{Var}[\mathrm{N}(r_{\ell})] = c_{(1,1)}(\mathfrak{s}_{\ell},\mathfrak{s}_{\ell}) \sqrt{n} + d_{(1,1)}(\mathfrak{s}_{\ell},\mathfrak{s}_{\ell}) + e_{(1,1)}(\mathfrak{s}_{\ell},\mathfrak{s}_{\ell}) n^{-\frac{1}{2}} + \bigO\big((\ln n)^{4}n^{-1}\big),
	 \\
& \mathrm{Cov}(\mathrm{N}(r_{\ell}),\mathrm{N}(r_{k})) = c_{(1,1)}(\mathfrak{s}_{\ell},\mathfrak{s}_{k})\sqrt{n} + d_{(1,1)}(\mathfrak{s}_{\ell},\mathfrak{s}_{k}) + e_{(1,1)}(\mathfrak{s}_{\ell},\mathfrak{s}_{k})n^{-\frac{1}{2}} + \bigO\big((\ln n)^{4}n^{-1}\big)  \nonumber
\end{align*}
as $n \to + \infty$, where
\begin{align}\nonumber
b_1(\mathfrak{s}_\ell) = &\; b\rho^{2b},
\qquad c_1(\mathfrak{s}_\ell) = \sqrt{2} \, b \rho^{b} \int_{-\infty}^{+\infty} \Big(\frac{\mathrm{erfc}(y+\mathfrak{s}_{\ell})}{\mathrm{erfc}(y)} -\chi_{(-\infty,0)}(y) \Big)dy,
	\\ \nonumber
d_1(\mathfrak{s}_\ell) = &  -\bigg(\frac{1}{2}+\alpha \bigg) + 2b \int_{-\infty}^{+\infty} \bigg\{2 y \Big(\frac{\mathrm{erfc}(y+\mathfrak{s}_{\ell})}{\mathrm{erfc}(y)} -\chi_{(-\infty,0)}(y) \Big)
	\\ \nonumber
&  + \frac{5y^{2}-1}{3\sqrt{\pi}} \frac{e^{-y^{2}}}{\mathrm{erfc}(y)} \frac{\mathrm{erfc}(y+\mathfrak{s}_{\ell})}{\mathrm{erfc}(y)} + \frac{1-5y^{2}-y \mathfrak{s}_{\ell} - 2 \mathfrak{s}_{\ell}^{2}}{3\sqrt{\pi}} \frac{e^{-(y+\mathfrak{s}_{\ell})^{2}}}{\mathrm{erfc}(y)} \bigg\} dy,
	\\ \nonumber
 e_1(\mathfrak{s}_\ell) =&\; \frac{b \rho^{-b}}{9 \sqrt{2} \pi} \int_{-\infty}^\infty \frac{1}{\erfc(y)^3}\big\{ 108 \pi  y^2 \erfc(y)^2 \erfc(y+\mathfrak{s}_\ell)+\sqrt{\pi } \erfc(y)^2 e^{-(y+\mathfrak{s}_\ell)^2} \big(2 \mathfrak{s}_\ell^3 (25 y^2-11)
  	\\\nonumber
& +2 \mathfrak{s}_\ell^2 y (31 y^2-33)+\mathfrak{s}_\ell (70 y^4-57 y^2+3)+16 \mathfrak{s}_\ell^4 y+8 \mathfrak{s}_\ell^5+y (50 y^4-193 y^2+21)\big)
	\\\nonumber
& +\erfc(y) \big(- e^{-y^2} \sqrt{\pi } y (50 y^4-193 y^2+21) \erfc(y+\mathfrak{s}_\ell)
	\\\nonumber
& -4 e^{-(y+\mathfrak{s}_\ell)^2-y^2}(5 y^2-1) (\mathfrak{s}_\ell y+2 \mathfrak{s}_\ell^2+5 y^2-1)\big)+4 e^{-2 y^2} (1-5 y^2)^2 \erfc(y+\mathfrak{s}_\ell)
	\\\nonumber
& -108 \pi  \chi_{(-\infty,0)}(y) y^2
   \erfc(y)^3 \big\} dy,
\end{align}
and, for $l \leq k$,
\begin{align}\label{def of c11 semihard edge}
& c_{(1,1)}(\mathfrak{s}_{\ell},\mathfrak{s}_{k}) = \sqrt{2} b\rho^{b} \int_{-\infty}^\infty \frac{\erfc(y + \mathfrak{s}_\ell)(\erfc(y) - \erfc(y + \mathfrak{s}_k))}{\erfc(y)^2} dy,
	 \\\nonumber
& d_{(1,1)}(\mathfrak{s}_{\ell},\mathfrak{s}_{k}) = \frac{2b}{3\sqrt{\pi}} \int_{-\infty}^{+\infty} \frac{1}{\erfc(y)^3} \big\{ \erfc(y)^2 \big(6 \sqrt{\pi } y \erfc(y+\mathfrak{s}_\ell)-e^{-(y+\mathfrak{s}_\ell)^2}
   (\mathfrak{s}_\ell y+2 \mathfrak{s}_\ell^2+5 y^2-1)\big)
   	\\\nonumber
&   +\erfc(y) \big( e^{-(y+\mathfrak{s}_\ell)^2}\erfc(y+\mathfrak{s}_k)
   (\mathfrak{s}_\ell y+2 \mathfrak{s}_\ell^2+5 y^2-1) -6\sqrt{\pi}y \, \mathrm{erfc}(y+\mathfrak{s}_{\ell})\mathrm{erfc}(y+\mathfrak{s}_{k})
   	\\\nonumber
&    +(e^{-y^{2}}+e^{-(y+\mathfrak{s}_k)^2})\erfc(y+\mathfrak{s}_\ell) (5y^{2}-1) +e^{-(y+\mathfrak{s}_k)^2} \erfc(y+\mathfrak{s}_\ell) \mathfrak{s}_k (2 \mathfrak{s}_k+y) \big)
 	\\\nonumber
& +2 e^{-y^2} (1-5 y^2) \erfc(y+\mathfrak{s}_\ell)  \erfc(y+\mathfrak{s}_k)\big\} dy,
	\\\nonumber
& e_{(1,1)}(\mathfrak{s}_{\ell},\mathfrak{s}_{k}) = \frac{b \rho^{-b}}{9 \sqrt{2} \pi} \int_{-\infty}^{+\infty} \frac{e^{-(y+\mathfrak{s}_\ell)^2-(y+\mathfrak{s}_k)^2}}{\erfc(y)^4} \Big\{-\erfc(y)^2 \Big(\sqrt{\pi } \erfc(y+\mathfrak{s}_\ell)
      	\\\nonumber
& \times \big(108 \sqrt{\pi } y^2 \erfc(y+\mathfrak{s}_k) e^{2 (\mathfrak{s}_\ell+\mathfrak{s}_k) y+\mathfrak{s}_\ell^2+\mathfrak{s}_k^2+2 y^2}
  +(50 y^4-193 y^2+21) y e^{(y+\mathfrak{s}_\ell)^2} (e^{\mathfrak{s}_k (\mathfrak{s}_k+2 y)}+1)
      	\\\nonumber
&  +\mathfrak{s}_k e^{(y+\mathfrak{s}_\ell)^2} (62 \mathfrak{s}_k y^3+(50 \mathfrak{s}_k^2-57) y^2
   +2 \mathfrak{s}_k (8 \mathfrak{s}_k^2-33) y+8 \mathfrak{s}_k^4-22 \mathfrak{s}_k^2+70
   y^4+3)\big)
      	\\\nonumber
&  +\sqrt{\pi } e^{(y+\mathfrak{s}_k)^2} [2 \mathfrak{s}_\ell^3 (25 y^2-11)+2 \mathfrak{s}_\ell^2 y (31 y^2-33)+\mathfrak{s}_\ell (70 y^4-57 y^2+3)+16 \mathfrak{s}_\ell^4 y+8
   \mathfrak{s}_\ell^5
      	\\\nonumber
&  +y (50 y^4-193 y^2+21)] \erfc(y+\mathfrak{s}_k)
	\\\nonumber
& +4 (\mathfrak{s}_\ell y+2 \mathfrak{s}_\ell^2+5 y^2-1) ((5 y^2-1) e^{\mathfrak{s}_k (\mathfrak{s}_k+2 y)}+\mathfrak{s}_k (2 \mathfrak{s}_k+y)+5
   y^2-1)\Big)
      	\\\nonumber
&  +\sqrt{\pi } \erfc(y)^3 e^{(y+\mathfrak{s}_k)^2} \Big(108 \sqrt{\pi } y^2 e^{(y+\mathfrak{s}_\ell)^2} \erfc(y+\mathfrak{s}_\ell)+2 \mathfrak{s}_\ell^3 (25 y^2-11)
	\\\nonumber
& +2 \mathfrak{s}_\ell^2 y
   (31 y^2-33)+\mathfrak{s}_\ell (70 y^4-57 y^2+3)+16 \mathfrak{s}_\ell^4 y+8 \mathfrak{s}_\ell^5+y (50 y^4-193 y^2+21)\Big)
      	\\\nonumber
&  +2 \erfc(y) \Big(4 (5 y^2-1) e^{\mathfrak{s}_k (\mathfrak{s}_k+2 y)} (\mathfrak{s}_\ell y+2 \mathfrak{s}_\ell^2+5 y^2-1) \erfc(y+\mathfrak{s}_k)
    	\\\nonumber
&+e^{\mathfrak{s}_\ell (\mathfrak{s}_\ell+2 y)} \erfc(y+\mathfrak{s}_\ell) (\sqrt{\pi } y (50 y^4-193 y^2+21) e^{(y+\mathfrak{s}_k)^2} \erfc(y+\mathfrak{s}_k)
    	\\\nonumber
&+2 (1-5 y^2)^2 (e^{\mathfrak{s}_k (\mathfrak{s}_k+2 y)}+2)+4 \mathfrak{s}_k (5 y^2-1) (2 \mathfrak{s}_k+y))\Big)
	\\\nonumber
&-12 (1-5 y^2)^2 e^{2 (\mathfrak{s}_\ell+\mathfrak{s}_k) y+\mathfrak{s}_\ell^2+\mathfrak{s}_k^2} \erfc(y+\mathfrak{s}_\ell) \erfc(y+\mathfrak{s}_k) \Big\}dy.
\end{align}

(b) As $n \to + \infty$, the random variable $(\mathcal{N}_{1},\dots,\mathcal{N}_{m})$, where
\begin{align}
& \mathcal{N}_{\ell} := \frac{\mathrm{N}(r_{\ell})-(b_1(\mathfrak{s}_\ell) n+ c_1(\mathfrak{s}_\ell) \sqrt{n})}{\sqrt{c_{(1,1)}(\mathfrak{s}_\ell, \mathfrak{s}_\ell)} \, n^{1/4}}, \qquad \ell=1,\dots,m,
\label{Nj semihard edge}
\end{align}
convergences in distribution to a multivariate normal random variable of mean $(0,\dots,0)$ whose covariance matrix $\Sigma$ is defined by
\begin{align*}
\Sigma_{\ell,\ell} = 1, \qquad \Sigma_{\ell,k} =  \Sigma_{k, \ell} = \frac{c_{(1,1)}(\mathfrak{s}_{\ell},\mathfrak{s}_{k})}{\sqrt{c_{(1,1)}(\mathfrak{s}_\ell, \mathfrak{s}_\ell)c_{(1,1)}(\mathfrak{s}_k,\mathfrak{s}_k)}}, \qquad 1 \leq \ell < k \leq m,
\end{align*}
where $c_{(1,1)}$ is given by \eqref{def of c11 semihard edge}.
\end{corollary}

\subsection{Results for the bulk}\label{subsection: bulk regime}
It turns out that the points in the bulk only feel the hard wall via exponentially small corrections. Consequently, the formulas for the bulk regime presented in our next theorem are {\it identical} to the corresponding formulas for the case without a hard edge presented in \cite{ChLe2022}. Moreover, the proof is almost identical to the proof of the analogous theorem in \cite{ChLe2022} and is therefore omitted (the only difference between the proofs is that a number of exponentially small error terms stemming from the hard wall appear in the proof of Theorem \ref{thm:main thm bulk}).

\begin{theorem}\label{thm:main thm bulk}(Merging radii in the bulk)

\noindent Let $m \in \mathbb{N}_{>0}$, $b>0$, $r \in (0,b^{-\frac{1}{2b}})$, $\mathfrak{s}_{1}<\dots<\mathfrak{s}_{m}$, and $\alpha > -1$ be fixed parameters, and for $n \in \mathbb{N}_{>0}$, define
\begin{align}\label{r ell bulk in thm}
r_{\ell} = r \bigg( 1+\frac{\sqrt{2}\, \mathfrak{s}_{\ell}}{r^{b}\sqrt{n}} \bigg)^{\frac{1}{2b}}, \qquad \ell=1,\dots,m.
\end{align}
For any fixed $x_{1},\dots,x_{m} \in \mathbb{R}$, there exists $\delta > 0$ such that
\begin{align}\label{asymp in main thm}
\mathbb{E}\bigg[ \prod_{j=1}^{m} e^{u_{j}\mathrm{N}(r_{j})} \bigg] = \exp \bigg( C_{1} n + C_{2} \sqrt{n} + C_{3} +  \frac{C_{4}}{\sqrt{n}} + \bigO\bigg(\frac{(\ln n)^{2}}{n}\bigg)\bigg), \qquad \mbox{as } n \to + \infty
\end{align}
uniformly for $u_{1} \in D_\delta(x_1),\dots,u_{m} \in D_\delta(x_m)$, where
\begin{align*}
& C_{1} = b r^{2b} \sum_{j=1}^{m}u_{j}, \\
& C_{2} = \sqrt{2}\, b r^{b} \int_{0}^{+\infty} \Big( \ln \mathcal{H}_{1}(t; \vec{u},\vec{\mathfrak{s}}) + \ln \mathcal{H}_{2}(t; \vec{u},\vec{\mathfrak{s}}) \Big) dt, \\
& C_{3} = - \bigg( \frac{1}{2}+\alpha \bigg)\sum_{j=1}^{m}u_{j} + 4b \int_{0}^{+\infty} t\Big( \ln \mathcal{H}_{1}(t; \vec{u},\vec{\mathfrak{s}}) - \ln \mathcal{H}_{2}(t; \vec{u},\vec{\mathfrak{s}}) \Big) dt + \sqrt{2}\, b \int_{-\infty}^{+\infty} \mathcal{G}_{1}(t;\vec{u},\vec{\mathfrak{s}}) dt, \\
& C_{4} = \frac{6\sqrt{2}\, b}{r^{b}} \int_{0}^{+\infty} t^{2}\Big( \ln \mathcal{H}_{1}(t; \vec{u},\vec{\mathfrak{s}}) + \ln \mathcal{H}_{2}(t; \vec{u},\vec{\mathfrak{s}}) \Big) dt \\
& \hspace{0.8cm} + \frac{b}{r^{b}} \int_{-\infty}^{+\infty} \bigg( 4t \, \mathcal{G}_{1}(t; \vec{u},\vec{\mathfrak{s}}) - \frac{\mathcal{G}_{1}(t; \vec{u},\vec{\mathfrak{s}})^{2}}{\sqrt{2}} + \mathcal{G}_{2}(t; \vec{u},\vec{\mathfrak{s}}) \bigg)dt,
\end{align*}
where
\begin{align}
 \mathcal{H}_{1}(t; \vec{u},\vec{\mathfrak{s}}):= &\;1 + \sum_{\ell=1}^{m} \frac{e^{u_{\ell}}-1}{2}\exp\bigg[ \sum_{j=\ell+1}^{m}u_{j} \bigg] \mathrm{erfc}(t-\mathfrak{s}_{\ell}), \label{function H1} \\
\mathcal{H}_{2}(t; \vec{u},\vec{\mathfrak{s}}):=&\; 1 + \sum_{\ell=1}^{m} \frac{e^{-u_{\ell}}-1}{2}\exp\bigg[ -\sum_{j=1}^{\ell-1}u_{j} \bigg] \mathrm{erfc}(t+\mathfrak{s}_{\ell}), \label{function H2} \\
 \mathcal{G}_{1}(t; \vec{u},\vec{\mathfrak{s}}) := &\;\frac{1}{\mathcal{H}_{1}(t; \vec{u},\vec{\mathfrak{s}})} \sum_{\ell=1}^{m}(e^{u_{\ell}}-1)\exp\bigg[ \sum_{j=\ell+1}^{m}u_{j} \bigg] \frac{e^{-(t-\mathfrak{s}_{\ell})^{2}}}{\sqrt{2\pi}} \frac{1-2\mathfrak{s}_{\ell}^{2}+t\mathfrak{s}_{\ell}-5t^{2}}{3}, \label{function G1} \\
 \mathcal{G}_{2}(t; \vec{u},\vec{\mathfrak{s}}) := &\; \frac{1}{\mathcal{H}_{1}(t; \vec{u},\vec{\mathfrak{s}})} \sum_{\ell=1}^{m} (e^{u_{\ell}}-1)\exp\bigg[ \sum_{j=\ell+1}^{m}u_{j} \bigg] \frac{e^{-(t-\mathfrak{s}_{\ell})^{2}}}{18\sqrt{2\pi}} \bigg( 50t^{5}-70t^{4}\mathfrak{s}_{\ell}-t^{3}\big( 73-62\mathfrak{s}_{\ell}^{2}\big)  \nonumber \\
& +t^{2}\mathfrak{s}_{\ell}\big(33-50\mathfrak{s}_{\ell}^{2}\big) - t \big( 3+18\mathfrak{s}_{\ell}^{2}-16\mathfrak{s}_{\ell}^{4} \big) - \mathfrak{s}_{\ell} \big( 3-22\mathfrak{s}_{\ell}^{2}+8\mathfrak{s}_{\ell}^{4} \big) \bigg). \label{function G2}
\end{align}

In particular, since $\mathbb{E}\big[ \prod_{j=1}^{m} e^{u_{j}\mathrm{N}(r_{j})} \big]$ depends analytically on $u_{1},\dots,u_{m} \in \mathbb{C}$ and is strictly positive for $u_{1},\dots,u_{m} \in \mathbb{R}$, the asymptotic formula \eqref{asymp in main thm} together with Cauchy's formula shows that
\begin{align}\label{der of main result bulk}
\partial_{u_{1}}^{k_{1}}\dots \partial_{u_{m}}^{k_{m}} \bigg\{ \ln \mathbb{E}\bigg[ \prod_{j=1}^{m} e^{u_{j}\mathrm{N}(r_{j})} \bigg] - \bigg( C_{1} n + C_{2} \sqrt{n} + C_{3} +  \frac{C_{4}}{\sqrt{n}} \bigg) \bigg\} = \bigO\bigg(\frac{(\ln n)^{2}}{n}\bigg), \quad \mbox{as } n \to + \infty,
\end{align}
for any $k_{1},\dots,k_{m}\in \mathbb{N}$, and $u_{1},\dots,u_{m}\in \mathbb{R}$.
\end{theorem}

\begin{remark}
In the above expressions for $C_{2},C_{3},C_{4}$, the functions $\mathcal{H}_{1}$, $\mathcal{H}_{2}$ appear inside logarithms. It was proved in \cite[Lemma 1.1]{ChLe2022} that $\mathcal{H}_1(t; \vec{u},\vec{\mathfrak{s}})>0$ and $\mathcal{H}_2(t; \vec{u},\vec{\mathfrak{s}})>0$ for all $t \in \mathbb{R}$, $\vec{u}=(u_{1},\dots,u_{m}) \in \mathbb{R}^{m}$ and $\mathfrak{s}_{1}<\dots<\mathfrak{s}_{m}$. This ensures that $C_{2},C_{3},C_{4}$ are well-defined and real-valued for $\vec{u}=(u_{1},\dots,u_{m}) \in \mathbb{R}^{m}$, $\mathfrak{s}_{1}<\dots<\mathfrak{s}_{m}$.
\end{remark}

In a similar way as in Subsections \ref{subsection: hard edge regime} and \ref{subsection: semi-hard edge regime}, one could derive from Theorem \ref{thm:main thm bulk} asymptotic formulas for the joint cumulants of $\mathrm{N}(r_{1}),\dots,\mathrm{N}(r_{m})$ in the bulk regime. For example, with $r_{\ell}$ as in \eqref{r ell bulk in thm}, i.e. $r_{\ell} = r \big( 1+\frac{\sqrt{2}\, \mathfrak{s}_{\ell}}{r^{b}\sqrt{n}} \big)^{\frac{1}{2b}}$ with $\mathfrak{s}_{\ell}\in \R$, we have
\begin{align}\label{exp in the bulk}
\mathbb{E}[\mathrm{N}(r_{\ell})] = br^{2b}n + \sqrt{2}\, br^{b}\mathfrak{s}_{\ell}\sqrt{n} + \frac{b-1-2\alpha}{2} + \bigO\bigg( \frac{(\ln n)^{2}}{n} \bigg), \qquad \mbox{as } n \to + \infty.
\end{align}
We do not write down the formulas for the other cumulants as they are identical to the corresponding formulas in \cite[Corollary 1.5]{ChLe2022}.

\medskip
It is interesting to compare \eqref{exp in the bulk} with the corresponding formula for the semi-hard edge regime of Corollary \ref{coro:correlation semihard}. To ease the comparison, it is convenient to replace $\mathfrak{s}_{\ell}$ by $-\mathfrak{s}_{\ell}$ in \eqref{def of rell semi-hard}, i.e. here we take $r_{\ell} = \rho \big( 1+\frac{\sqrt{2}\, \mathfrak{s}_{\ell}}{\rho^{b}\sqrt{n}} \big)^{\frac{1}{2b}}$ with $\mathfrak{s}_{\ell}<0$. Then it follows from Corollary \ref{coro:correlation semihard} that
\begin{align}\label{semi hard edge expecta comp}
\mathbb{E}[\mathrm{N}(r_{\ell})] = b\rho^{2b} n + c_1(-\mathfrak{s}_\ell) \sqrt{n} + d_1(-\mathfrak{s}_\ell) + \bigO\big(n^{-\frac{1}{2}}\big), \qquad \mbox{as } n \to +\infty.
\end{align}
Furthermore, by a long but direct analysis, we obtain as $\mathfrak{s}_{\ell}\to - \infty$ that
\begin{align}\label{c1 d1 large s}
& c_1(-\mathfrak{s}_\ell) = \sqrt{2} b \rho^{b}\mathfrak{s}_{\ell} + \bigO(e^{-c\mathfrak{s}_{\ell}^{2}}), & & d_1(-\mathfrak{s}_\ell) = \frac{b-1-2\alpha}{2} + \bigO(e^{-c\mathfrak{s}_{\ell}^{2}}),
\end{align}
for a small but fixed $c>0$. Recall that the asymptotic formula \eqref{semi hard edge expecta comp} is proved for fixed $s_{\ell}<0$. However, if we formally replace $c_1(-\mathfrak{s}_\ell)$ by $\sqrt{2} b \rho^{b}\mathfrak{s}_{\ell}$ and $d_1(-\mathfrak{s}_\ell)$ by $\frac{b-1-2\alpha}{2}$ in \eqref{semi hard edge expecta comp}, then the terms of order $\sqrt{n}$ and $1$ in \eqref{exp in the bulk} and \eqref{c1 d1 large s} are identical. Thus the above computation suggests that (i) the asymptotic formula \eqref{semi hard edge expecta comp} probably holds as $n\to +\infty$ and simultaneously as $\mathfrak{s}_{\ell}\to -\infty$ at a sufficiently slow speed, and (ii) that the transition between the semi-hard edge regime and the bulk regime does not contain an intermediate regime.

\medskip \textbf{Outline of proof.} Relying on the determinantal structure of \eqref{def of point process hard}, we can rewrite $\mathbb{E}\big[ \prod_{\ell=1}^{m} e^{u_{\ell}\mathrm{N}(r_{\ell})} \big]$ as a ratio of two determinants using e.g. \cite[Lemma 2.1]{WebbWong} or \cite[Lemma 1.9]{Charlier 2d jumps} (see also \cite{ByunKangSeo}),
\begin{align}
\mathbb{E}\bigg[ \prod_{\ell=1}^{m} e^{u_{\ell}\mathrm{N}(r_{\ell})} \bigg] & = \frac{1}{n!\mathcal{Z}_{n}} \int_{\mathbb{C}}\dots \int_{\mathbb{C}} \prod_{1 \leq j < k \leq n} |z_{k} -z_{j}|^{2} \prod_{j=1}^{n}w(z_{j}) d^{2}z_{j} \nonumber \\
& = \frac{1}{\mathcal{Z}_{n}} \det \left( \int_{\mathbb{C}} z^{j} \overline{z}^{k} w(z) d^{2}z \right)_{j,k=0}^{n-1} = \frac{1}{\mathcal{Z}_{n}}(2\pi)^{n}\prod_{j=0}^{n-1}\int_{0}^{\rho}u^{2j+1}w(u)du, \label{simplified determinant}
\end{align}
where
\begin{align}\label{def of w and omega}
w(z):=|z|^{2\alpha} e^{-n |z|^{2b}} \omega(|z|), \qquad \omega(x) := \prod_{\ell=1}^{m} \begin{cases}
e^{u_{\ell}}, & \mbox{if } x < r_{\ell}, \\
1, & \mbox{if } x \geq r_{\ell}.
\end{cases}
\end{align}
For $x < \rho$, let us write
\begin{align}\label{def of omegaell}
\omega (x) = \sum_{\ell=1}^{m+1}\omega_{\ell} \mathbf{1}_{[0,r_{\ell})}(x), \qquad \omega_{\ell} := \begin{cases}
e^{u_{\ell}+\dots+u_{m}}-e^{u_{\ell+1}+\dots+u_{m}}, & \mbox{if } \ell < m, \\
e^{u_{m}}-1, & \mbox{if } \ell=m, \\
1, & \mbox{if } \ell=m+1,
\end{cases}
\end{align}
where $r_{m+1}:=\rho$. Note also that $\Omega := e^{u_{1}+\dots+u_{m}} = \sum_{j=1}^{m+1}\omega_{j}$. By \eqref{def of w and omega}--\eqref{def of omegaell},
\begin{align*}
\int_{0}^{\rho}u^{2j+1}w(u)du & = \int_{0}^{\rho} u^{2j+1}u^{2\alpha}e^{-nu^{2b}}du + \sum_{\ell=1}^{m} \omega_{\ell} \int_{0}^{r_{\ell}} u^{2j+1}u^{2\alpha}e^{-nu^{2b}}du
	\\
& = \int_{0}^{n \rho^{2b}} \Big( \frac{y}{n} \Big)^{\frac{j+1+\alpha}{b}} \frac{e^{-y}}{2by}dy + \sum_{\ell=1}^{m} \omega_{\ell} \int_{0}^{n r_{\ell}^{2b}} \Big( \frac{y}{n} \Big)^{\frac{j+1+\alpha}{b}} \frac{e^{-y}}{2by}dy \\
& = \frac{n^{-\frac{j+1+\alpha}{b}}}{2b} \bigg( \gamma(\tfrac{j+1+\alpha}{b},n \rho^{2b}) + \sum_{\ell=1}^{m} \omega_{\ell} \; \gamma(\tfrac{j+1+\alpha}{b},n r_{\ell}^{2b}) \bigg),
\end{align*}
where $\gamma(a,z)$ is the incomplete gamma function
\begin{align*}
\gamma(a,z) = \int_{0}^{z}t^{a-1}e^{-t}dt.
\end{align*}
Hence,
\begin{align*}
(2\pi)^{n}\prod_{j=0}^{n-1}\int_{0}^{\rho}u^{2j+1}w(u)du = n^{-\frac{n^{2}}{2b}}n^{-\frac{1+2\alpha}{2b}n} \frac{\pi^{n}}{b^{n}} \prod_{j=1}^{n} \bigg( \gamma(\tfrac{j+\alpha}{b},n \rho^{2b}) + \sum_{\ell=1}^{m} \omega_{\ell} \; \gamma(\tfrac{j+\alpha}{b},n r_{\ell}^{2b}) \bigg).
\end{align*}
An expression for $\mathcal{Z}_{n}$ in terms of $\gamma$ can be found by setting $\omega_{1}=\dots=\omega_{m} =0$ above:
\begin{align*}
\mathcal{Z}_{n} = n^{-\frac{n^{2}}{2b}}n^{-\frac{1+2\alpha}{2b}n} \frac{\pi^{n}}{b^{n}} \prod_{j=1}^{n}  \gamma(\tfrac{j+\alpha}{b},n \rho^{2b}),
\end{align*}
and therefore, by \eqref{simplified determinant},
\begin{align}\label{main exact formula}
\ln \mathcal{E}_{n} = \sum_{j=1}^{n} \ln \bigg(1+\sum_{\ell=1}^{m} \omega_{\ell} \frac{\gamma(\tfrac{j+\alpha}{b},nr_{\ell}^{2b})}{\gamma(\tfrac{j+\alpha}{b},n \rho^{2b})} \bigg),
\end{align}
where $\mathcal{E}_{n}:=\mathbb{E}\big[ \prod_{\ell=1}^{m} e^{u_{\ell}\mathrm{N}(r_{\ell})} \big]$. The above formula is the starting point of the proofs of Theorems \ref{thm:main thm hard}, \ref{thm:main thm semi-hard} and \ref{thm:main thm bulk}. We infer from \eqref{main exact formula} that, to obtain the large $n$ asymptotics of $\mathcal{E}_{n}$, we need the asymptotics of $\gamma(a,z)$ as $a,z$ tend to $+\infty$ at various relative speeds. The uniform asymptotics of $\gamma$ are actually well-known, and we recall them in Appendix \ref{section:uniform asymp gamma}.


\medskip The approach considered here shows similarities with \cite{Charlier 2d jumps, Charlier 2d gap, ChLe2022, ByunCharlier}. The large $n$ behavior of $\gamma(\tfrac{j+\alpha}{b},n \rho^{2b})$ depends crucially on whether $\frac{j+\alpha}{b} \ll n\rho^{2b}$, $\frac{j+\alpha}{b}\approx n\rho^{2b}$ or $\frac{j+\alpha}{b}\gg n\rho^{2b}$. Hence, for the proofs of both Theorem \ref{thm:main thm hard} and Theorem \ref{thm:main thm semi-hard}, we will split the sum in \eqref{main exact formula} into four parts,
\begin{align*}
\ln \mathcal{E}_{n} = S_{0} + S_{1} + S_{2} + S_{3},
\end{align*}
where $S_{0},\ldots,S_{3}$ are defined in \eqref{def of S0 and S1 hard}--\eqref{def of S2 and S3 hard}. The sum $S_{0}$ involves a large but fixed number of $j$'s; the sum $S_{1}$ corresponds to those $j$'s that are ``large" and for which $\frac{j+\alpha}{b} \ll n\rho^{2b}$; and the sum $S_{3}$ involves the $j$'s for which $\frac{j+\alpha}{b} \gg n\rho^{2b}$. For both theorems, the most delicate sum is $S_{2}$: this sum involves the $j$-terms in \eqref{main exact formula} for which $\frac{j+\alpha}{b}\approx n\rho^{2b}$, and therefore critical transitions occur in the asymptotic behavior of the functions $\{\gamma(\tfrac{j+\alpha}{b},nr_{\ell}^{2b})\}_{\ell=1}^{m}$ and $\gamma(\tfrac{j+\alpha}{b},n \rho^{2b})$ when performing the sum $S_{2}$.

\medskip For the two novel regimes considered in this work, namely the hard edge regime \eqref{def of rell hard} and the semi-hard edge regime \eqref{def of rell semi-hard}, the proofs require precise Riemann sum approximations for functions with singularities (the singularities are more difficult to handle in the hard edge regime). In comparison, the bulk regime of Theorem \ref{thm:main thm bulk} (whose proof is omitted here as it is essentially identical to \cite{ChLe2022}) is simpler as the corresponding Riemann sum approximations involve more well-behaved functions.

\medskip \textbf{Related works.} By \eqref{simplified determinant}--\eqref{def of w and omega}, we have $\mathcal{E}_{n} = D_{n}/\mathcal{Z}_{n}$ where $D_{n}$ is an $n \times n$ determinant with a rotation-invariant weight supported on $\mathbb{C}$ and with $m$ merging discontinuities: for Theorem \ref{thm:main thm hard}, the discontinuities are merging near the hard edge at speed $1/n$; for Theorem \ref{thm:main thm semi-hard}, the discontinuities are merging near the hard edge at speed $1/\sqrt{n}$; and for Theorem \ref{thm:main thm bulk}, the discontinuities are merging in the bulk at speed $1/\sqrt{n}$.

The problem of determining asymptotics of structured determinants with discontinuities has a long history. When the weight is supported on the unit circle or on the real line, this problem was studied by many authors, including Lenard, Fisher, Hartwig, Widom, Basor, B\"ottcher, Silbermann, Ehrhardt, Deift, Its and Krasovsky, see e.g. \cite{BasMor, DIKreview, Charlier} for some historical background, \cite{CD2019, CharlierBessel, LCX2022, DXZ2022 bis, DZ2022} for structured determinants with discontinuities near a  hard edge, and \cite{CK2015, Fahs} for merging discontinuities in the bulk.

A central theme in normal random matrix theories concerns the asymptotic distribution of linear statistics $\sum_1^n f(z_j)$ where $f$ is a given test-function on the plane.
The analytical situation depends crucially on whether or not $f$ belongs to the Sobolev class $W^{1,2}$, since this is believed to be the right condition under which we obtain a well-defined limiting normal distribution (say, after subtracting the expectation). This is
rigorously verified in the Ginibre case in \cite{RV} and if the test-function is $C^2$-smooth for more general ensembles in \cite{AM}. However, the class $W^{1,2}$ excludes certain natural test-functions,
or the logarithm $l_z(w)=\ln|z-w|$ (or close relatives like Green's functions) which is used in connection with the Gaussian free field, and characteristic functions $\chi_E(z)$ which define counting statistics.



The works \cite{CE2020, L et al 2019, FenzlLambert, Charlier 2d jumps, ChLe2022} were already mentioned earlier in the introduction and deal with determinants with discontinuities in dimension two.
Determinants corresponding to the logarithmic test-function $l_z$, for some special ensembles, have attracted considerable attention in recent years \cite{WebbWong, DeanoSimm, ByunCharlier, BDH2022}, see also e.g. \cite{AKS2018,BBLM2015, BGM17, BEG18, LeeYang3}.

\section{Proof of Theorem \ref{thm:main thm hard}}\label{section:proof edge}
In this section, the $r_{\ell}$'s are as in \eqref{def of rell hard}. Our proof strategy follows \cite{Charlier 2d jumps, Charlier 2d gap, ChLe2022, ByunCharlier}.

\medskip Let us define
\begin{align}\label{jplusjminusdef}
& j_{-}:=\lceil \tfrac{bn\rho^{2b}}{1+\epsilon} - \alpha \rceil, \qquad j_{+} := \lfloor  \tfrac{bn\rho^{2b}}{1-\epsilon} - \alpha \rfloor,
\end{align}
where $\epsilon > 0$ is independent of $n$. We assume that $\epsilon$ is sufficiently small such that
\begin{align}\label{brho2b1epsilon}
\frac{b\rho^{2b}}{1-\epsilon} < 1,
\end{align}
so that, recalling the formula \eqref{main exact formula} for $\ln \mathcal{E}_{n}$, we can write
\begin{align}\label{log Dn as a sum of sums hard}
\ln \mathcal{E}_{n} = S_{0} + S_{1} + S_{2} + S_{3},
\end{align}
where
\begin{align}
& S_{0} = \sum_{j=1}^{M'} \ln \bigg( 1+\sum_{\ell=1}^{m} \omega_{\ell} \frac{\gamma(\tfrac{j+\alpha}{b},nr_{\ell}^{2b})}{\gamma(\tfrac{j+\alpha}{b},n \rho^{2b})} \bigg), & & S_{1} = \sum_{j=M'+1}^{j_{-}-1} \ln \bigg( 1+\sum_{\ell=1}^{m} \omega_{\ell} \frac{\gamma(\tfrac{j+\alpha}{b},nr_{\ell}^{2b})}{\gamma(\tfrac{j+\alpha}{b},n \rho^{2b})} \bigg), \label{def of S0 and S1 hard} \\
& S_{2} = \sum_{j=j_{-}}^{j_{+}} \ln \bigg( 1+\sum_{\ell=1}^{m} \omega_{\ell} \frac{\gamma(\tfrac{j+\alpha}{b},nr_{\ell}^{2b})}{\gamma(\tfrac{j+\alpha}{b},n \rho^{2b})} \bigg), & & S_{3}=\sum_{j=j_{+}+1}^{n} \ln \bigg( 1+\sum_{\ell=1}^{m} \omega_{\ell} \frac{\gamma(\tfrac{j+\alpha}{b},nr_{\ell}^{2b})}{\gamma(\tfrac{j+\alpha}{b},n \rho^{2b})} \bigg). \label{def of S2 and S3 hard}
\end{align}
In the above, $M'>0$ is an integer independent of $n$. For $j=1,\dots,n$ and $k =1,\dots,m$, we also define $a_{j}:=\frac{j+\alpha}{b}$, and
\begin{subequations}\label{def etajl hard}
\begin{align}
& \lambda_{j,k} := \frac{bnr_{k}^{2b}}{j+\alpha}, & & \eta_{j,k} := (\lambda_{j,k}-1)\sqrt{\frac{2 (\lambda_{j,k}-1-\ln \lambda_{j,k})}{(\lambda_{j,k}-1)^{2}}}, \\
& \lambda_{j} := \frac{bn\rho^{2b}}{j+\alpha}, & & \eta_{j} := (\lambda_{j}-1)\sqrt{\frac{2 (\lambda_{j}-1-\ln \lambda_{j})}{(\lambda_{j}-1)^{2}}}.
\end{align}
\end{subequations}
With this notation, the summand in \eqref{def of S0 and S1 hard}--\eqref{def of S2 and S3 hard} can be rewritten as
\begin{align*}
\ln \bigg( 1+\sum_{\ell=1}^{m} \omega_{\ell} \frac{\gamma(a_{j},a_{j}\lambda_{j,\ell})}{\gamma(a_{j},a_{j}\lambda_{j})} \bigg).
\end{align*}
The notation $\eta_{j}$ and $\eta_{j,k}$ in \eqref{def of S0 and S1 hard}--\eqref{def of S2 and S3 hard} is introduced in the same spirit as the notation $\eta$ of Lemma \ref{lemma: uniform}.
Recall also that $\Omega := e^{u_{1}+\dots+u_{m}} = \sum_{j=1}^{m+1}\omega_{j}$.

\begin{lemma}\label{lemma: S0 hard}
For any $x_{1},\dots,x_{m} \in \mathbb{R}$, there exists $\delta > 0$ such that
\begin{align}\label{asymp of S0 hard}
S_{0} = M' \ln \Omega + \bigO(e^{-cn}), \qquad \mbox{as } n \to + \infty,
\end{align}
uniformly for $u_{1} \in D_\delta(x_1),\dots,u_{m} \in D_\delta(x_m)$.
\end{lemma}
\begin{proof}
We infer from \eqref{def of S0 and S1 hard} and Lemma \ref{lemma:various regime of gamma} that
\begin{align*}
S_{0} & = \sum_{j=1}^{M'} \ln \bigg( \sum_{\ell=1}^{m+1} \omega_{\ell} \big[1 + \bigO(e^{-cn}) \big] \bigg) = \sum_{j=1}^{M'} \ln \Omega + \bigO(e^{-cn}), \quad \mbox{as } n \to +\infty.
\end{align*}
In the above, the error terms before the second equality are independent of $u_{1},\dots,u_{m}$, so the claim follows.
\end{proof}

\begin{lemma}\label{lemma: S2km1 hard}
The constant $M'$ can be chosen sufficiently large such that the following holds. For any $x_{1},\dots,x_{m} \in \mathbb{R}$, there exists $\delta > 0$ such that
\begin{align*}
& S_{1} = (j_{-}-M'-1) \ln \Omega + \bigO(e^{-cn}),
\end{align*}
as $n \to +\infty$ uniformly for $u_{1} \in D_\delta(x_1),\dots,u_{m} \in D_\delta(x_m)$.
\end{lemma}
\begin{proof}
According to (\ref{def of S0 and S1 hard}) and (\ref{def etajl hard}), we have
$$S_{1} = \sum_{j=M'+1}^{j_{-}-1} \ln \bigg( 1+\sum_{\ell=1}^{m} \omega_{\ell}
\frac{\gamma(a_j, a_j \lambda_{j, \ell}) }{ \gamma(a_j, a_j \lambda_j )} \bigg).$$
There is a $\delta>0$ such that $\lambda_j > 1 + \delta$ and $\lambda_{j,\ell} = \lambda_j (1 - t_\ell/n) > 1 + \delta$ for all $j \in \{M'+1, \dots, j_{-}-1\}$ and $\ell \in \{1,\dots,m\}$. Hence, by Lemma \ref{lemma: uniform} $(i)$ we can choose $M'$  such that
$$S_{1} = \sum_{j=M'+1}^{j_{-}-1} \ln \bigg( 1+\sum_{\ell=1}^{m} \omega_{\ell}
\frac{1 + \bigO(e^{-\frac{a_j \eta_{j, \ell}^2}{2}}) }{1 + \bigO(e^{-\frac{a_j \eta_j^2}{2}})} \bigg),$$
where the error terms are uniform with respect to $j$ and $\ell$. The functions $j \mapsto a_j \eta_j^2$ and $j \mapsto a_j \eta_{j,\ell}^2$ are decreasing, because
$$\partial_j (a_j \eta_j^2) = -\frac{2}{b}\ln{\lambda_j} < 0, \qquad
\partial_j (a_j \eta_{j,\ell}^2) = -\frac{2}{b}\ln{\lambda_{j,\ell}} < 0.$$
Moreover, we have $a_{j_-} \eta_{j_-}^2 > 2c n$ and hence $a_{j_-} \eta_{j_-,\ell}^2 = a_{j_-} \eta_{j_-}^2 + \bigO(1) > cn$ for all sufficiently large $n$ for some $c > 0$. It follows that
$$S_{1} = \sum_{j=M'+1}^{j_{-}-1} \ln \bigg( 1+\sum_{\ell=1}^{m} \omega_{\ell}
\frac{1 + \bigO(e^{-cn})}{1 + \bigO(e^{-cn})} \bigg)
= \sum_{j=M'+1}^{j_{-}-1} \ln\bigg( 1+ \sum_{\ell=1}^{m} \omega_{\ell} \bigg) + \bigO(e^{-cn}),$$
from which the desired conclusion follows.
\end{proof}

To obtain the large $n$ asymptotics of $S_{3}$, we will rely on the following lemma.
\begin{lemma}\label{lemma:Riemann sum NEW}[Adapted from \cite[Lemma 3.4]{Charlier 2d gap}]
Let $A=A(n),a_{0}=a_{0}(n)$, $B=B(n),b_{0}=b_{0}(n)$ be bounded functions of $n \in \{1,2,\dots\}$, such that
\begin{align*}
& a_{n} := An + a_{0} \qquad \mbox{ and } \qquad b_{n} := Bn + b_{0}
\end{align*}
are integers. Assume also that $B-A$ is positive and remains bounded away from $0$. Let $f$ be a function independent of $n$, which is $C^{2}([\min\{\frac{a_{n}}{n},A\},\max\{\frac{b_{n}}{n},B\}])$ for all $n\in \{1,2,\dots\}$. Then as $n \to + \infty$, we have
\begin{multline}
 \sum_{j=a_{n}}^{b_{n}}f(\tfrac{j}{n}) = n \int_{A}^{B}f(x)dx + \frac{(1-2a_{0})f(A)+(1+2b_{0})f(B)}{2}   \\
+ \bigO \bigg( \frac{\mathfrak{m}_{A,n}(f')+\mathfrak{m}_{B,n}(f')}{n} + \sum_{j=a_{n}}^{b_{n}-1} \frac{\mathfrak{m}_{j,n}(f'')}{n^{2}} \bigg), \label{sum f asymp gap NEW}
\end{multline}
where, for a given function $g$ continuous on $\big[\min\{\frac{a_{n}}{n},A\},\max\{\frac{b_{n}}{n},B\}\big]$,
\begin{align*}
\mathfrak{m}_{A,n}(g) := \max_{x \in [\min\{\frac{a_{n}}{n},A\},\max\{\frac{a_{n}}{n},A\}]}|g(x)|, \quad \mathfrak{m}_{B,n}(g) := \max_{x \in [\min\{\frac{b_{n}}{n},B\},\max\{\frac{b_{n}}{n},B\}]}|g(x)|,
\end{align*}
and for $j \in \{a_{n},\dots,b_{n}-1\}$, $\mathfrak{m}_{j,n}(g) := \max_{x \in [\frac{j}{n},\frac{j+1}{n}]}|g(x)|$.
\end{lemma}
Following the approach of \cite{Charlier 2d jumps, Charlier 2d gap}, we define
\begin{align} \label{def of theta n eps hard}
\theta_{+}^{(n,\epsilon)} = \bigg( \frac{b n \rho^{2b}}{1-\epsilon}-\alpha \bigg)-\bigg\lfloor \frac{b n \rho^{2b}}{1-\epsilon}-\alpha \bigg\rfloor, \qquad \theta_{-}^{(n,\epsilon)} = \bigg\lceil \frac{b n \rho^{2b}}{1+\epsilon}-\alpha \bigg\rceil-\bigg( \frac{b n \rho^{2b}}{1+\epsilon}-\alpha \bigg).
\end{align}
\begin{lemma}\label{lemma:S3 asymp hard}
For any $x_{1},\dots,x_{m} \in \mathbb{R}$, there exists $\delta > 0$ such that
\begin{align*}
& S_{3} = n \int_{\frac{b \rho^{2b}}{1-\epsilon}}^{1} f_{1}(x)dx + \int_{\frac{b\rho^{2b}}{1-\epsilon}}^{1} f(x)dx + (\alpha+\theta_{+}^{(n,\epsilon)}-\tfrac{1}{2})f_{1}(\tfrac{b\rho^{2b}}{1-\epsilon})+\tfrac{1}{2}f_{1}(1) + \bigO(n^{-1}),
\end{align*}
as $n \to +\infty$ uniformly for $u_{1} \in D_\delta(x_1),\dots,u_{m} \in D_\delta(x_m)$, where $f_{1}(x) := \ln \big( 1+\mathsf{T}_{0}(x) \big)$ and $f$ and $\mathsf{T}_{j}$ are defined in \eqref{def of f hard} and \eqref{def of T hard}.
\end{lemma}
\begin{proof}
For $j \geq j_{+}+1$ and $k \in \{1,\dots,m\}$, $1-\lambda_{j,k}$ and $1-\lambda_{j}$ are positive and bounded away from $0$. Hence, using Lemma \ref{lemma: asymp of gamma for lambda bounded away from 1} (ii), we obtain
\begin{align*}
S_{3} & = \sum_{j=j_{+}+1}^{n} \ln \bigg\{ 1+ \frac{\sum_{\ell=1}^{m} \omega_{\ell} \frac{e^{-\frac{a_j}{2}\eta_{j,\ell}^2}}{\sqrt{2\pi}} \big\{\sum_{k=0}^{1} \frac{S(\varphi_k(\lambda_{j,\ell}))}{a_j^{k+1/2}} + \bigO\big(\frac{1}{a_j^{5/2}}\big) +  \bigO\big(\frac{1}{(a_j \eta_{j,\ell}^2)^{5/2}}\big)\big\}}{\frac{e^{-\frac{a_j}{2}\eta_j^2}}{\sqrt{2\pi}} \big\{\sum_{k=0}^{1} \frac{S(\varphi_k(\lambda_j))}{a_j^{k+1/2}} + \bigO\big(\frac{1}{a_j^{5/2}}\big) +  \bigO\big(\frac{1}{(a_j \eta_j^2)^{5/2}}\big)\big\}  } \bigg\} \\
& = \sum_{j=j_{+}+1}^{n} \ln \bigg\{ 1+\sum_{\ell=1}^{m} \omega_{\ell} \frac{e^{-\frac{a_{j}\eta_{j,\ell}^{2}}{2}}( \frac{-1}{\lambda_{j,\ell}-1}\frac{1}{\sqrt{a_{j}}}+\frac{1+10\lambda_{j,\ell}+\lambda_{j,\ell}^{2}}{12(\lambda_{j,\ell}-1)^{3}} \frac{1}{a_{j}^{3/2}} + \bigO(n^{-5/2}) )}{e^{-\frac{a_{j}\eta_{j}^{2}}{2}}( \frac{-1}{\lambda_{j}-1}\frac{1}{\sqrt{a_{j}}}+\frac{1+10\lambda_{j}+\lambda_{j}^{2}}{12(\lambda_{j}-1)^{3}} \frac{1}{a_{j}^{3/2}} + \bigO(n^{-5/2}) )} \bigg\} \\
& = \sum_{j=j_{+}+1}^{n} \bigg(f_{1}(j/n)+\frac{1}{n}f(j/n) + \bigO(n^{-2})\bigg),
\end{align*}
where the above $\bigO$-terms are uniform for $j \in \{j_{+}+1,\dots,n\}$. The claim then follows after a computation using Lemma \ref{lemma:Riemann sum NEW} (with $A=\frac{b\rho^{2b}}{1-\epsilon}$, $a_{0}=1-\alpha-\theta_{+}^{(n,\epsilon)}$, $B=1$ and $b_{0}=0$).
\end{proof}
We now focus on $S_{2}$. Let $M:=n^{\frac{1}{10}}$. We split $S_{2}$ in three pieces as follows
\begin{align}\label{asymp prelim of S2kpvp hard}
& S_{2}=S_{2}^{(1)}+S_{2}^{(2)}+S_{2}^{(3)}, \qquad S_{2}^{(v)} := \sum_{j:\lambda_{j}\in I_{v}}  \ln \bigg( 1+\sum_{\ell=1}^{m} \omega_{\ell} \frac{\gamma(a_{j},a_{j}\lambda_{j,\ell})}{\gamma(a_{j},a_{j}\lambda_{j})} \bigg), \quad v=1,2,3,
\end{align}
where
\begin{align}\label{I1I2I3def}
I_{1} = [1-\epsilon,1-\tfrac{M}{\sqrt{n}}), \qquad I_{2} = [1-\tfrac{M}{\sqrt{n}},1+\tfrac{M}{\sqrt{n}}], \qquad I_{3} = (1+\tfrac{M}{\sqrt{n}},1+\epsilon].
\end{align}
From \eqref{asymp prelim of S2kpvp hard}, we see that the large $n$ asymptotics of $\{S_{2}^{(v)}\}_{v=1,2,3}$ involve the asymptotics of $\gamma(a,z)$ when $a \to + \infty$, $z \to +\infty$ with $\lambda=\frac{z}{a} \in [1-\epsilon,1+\epsilon]$. These sums can also be rewritten using
\begin{align}\label{sums lambda j hard}
& \sum_{j:\lambda_{j}\in I_{3}} = \sum_{j=j_{-}}^{g_{-}-1}, \qquad \sum_{j:\lambda_{j}\in I_{2}} = \sum_{j= g_{-}}^{g_{+}}, \qquad \sum_{j:\lambda_{j}\in I_{1}} = \sum_{j= g_{+}+1}^{j_{+}},
\end{align}
where $g_{-} := \lceil \frac{bn\rho^{2b}}{1+\frac{M}{\sqrt{n}}}-\alpha \rceil$, $g_{+} := \lfloor \frac{bn\rho^{2b}}{1-\frac{M}{\sqrt{n}}}-\alpha \rfloor$. Let us also define
\begin{align*}
& \theta_{-}^{(n,M)} := g_{-} - \bigg( \frac{bn \rho^{2b}}{1+\frac{M}{\sqrt{n}}} - \alpha \bigg) = \bigg\lceil \frac{bn \rho^{2b}}{1+\frac{M}{\sqrt{n}}} - \alpha \bigg\rceil - \bigg( \frac{bn \rho^{2b}}{1+\frac{M}{\sqrt{n}}} - \alpha \bigg), \\
& \theta_{+}^{(n,M)} := \bigg( \frac{bn \rho^{2b}}{1-\frac{M}{\sqrt{n}}} - \alpha \bigg) - g_{+} = \bigg( \frac{bn \rho^{2b}}{1-\frac{M}{\sqrt{n}}} - \alpha \bigg) - \bigg\lfloor \frac{bn \rho^{2b}}{1-\frac{M}{\sqrt{n}}} - \alpha \bigg\rfloor.
\end{align*}
Clearly, $\theta_{-}^{(n,M)},\theta_{+}^{(n,M)} \in [0,1)$. Note that the individual sums $S_{2}^{(1)},S_{2}^{(2)},S_{2}^{(3)}$ depend on $M$, although $S_{2}=S_{2}^{(1)}+S_{2}^{(2)}+S_{2}^{(3)}$ is independent of $M$. Below, we will first obtain large $n$ asymptotics of $S_{2}^{(1)},S_{2}^{(2)},S_{2}^{(3)}$. After adding the asymptotic formulas of $S_{2}^{(1)},S_{2}^{(2)},S_{2}^{(3)}$, we will find that all $M$-dependent terms cancel, as they must. For this reason, below we will not replace $M$ by $n^{1/10}$ until the last step of the proof. The reason why we choose $M=n^{1/10}$ is technical. In the various asymptotic formulas below, there will be different types of error terms, such as $\bigO(\frac{M^{4}}{\sqrt{n}})$, $\bigO(\frac{\sqrt{n}}{M^{11}})$, etc, and in the last step of the proof we will find that $M=n^{1/10}$ is the choice that produces the best control over the total error.

\begin{lemma}\label{lemma:S2kp3p hard}
For any $x_{1},\dots,x_{m} \in \mathbb{R}$, there exists $\delta > 0$ such that
\begin{align*}
S_{2}^{(3)} = & \; \Big( b\rho^{2b}n - j_{-} - bM\rho^{2b}\sqrt{n} + bM^{2}\rho^{2b} -\alpha+\theta_{-}^{(n,M)} - bM^{3}\rho^{2b}n^{-\frac{1}{2}} \Big) \ln  \Omega + \bigO(M^{4}n^{-1}),
\end{align*}
as $n \to +\infty$ uniformly for $u_{1} \in D_\delta(x_1),\dots,u_{m} \in D_\delta(x_m)$.
\end{lemma}
\begin{proof}
Recall that $a_{j}, \lambda_{j}, \lambda_{j,k}, \eta_{j}, \eta_{j,k}$ are defined in \eqref{def etajl hard}. By \eqref{asymp prelim of S2kpvp hard}, we have
\begin{align*}
S_{2}^{(3)} & = \sum_{j:\lambda_{j}\in I_{3}} \ln \bigg( 1+ \sum_{\ell=1}^{m} \omega_{\ell} \frac{\gamma\big(a_j, a_j \lambda_{j, \ell}\big) }{ \gamma\big(a_j, a_j \lambda_j \big)}\bigg).
\end{align*}
If $\lambda_{j}\in I_{3}$, then $\lambda_j > 1 + \frac{M}{\sqrt{n}}$ and $\lambda_{j,\ell} = \lambda_j (1 - t_\ell/n) > 1 + \frac{M}{\sqrt{n}} + \bigO(n^{-1})$. So there exists a constant $c > 0$ such that
\begin{align*}
&\eta_j \geq c \frac{M}{\sqrt{n}}, \qquad   -\eta_{j} \sqrt{a_{j}/2} \leq - c M, \qquad \eta_{j,\ell} \geq c \frac{M}{\sqrt{n}}, \qquad -\eta_{j,\ell} \sqrt{a_{j}/2} \leq - c M,
\end{align*}
for all sufficiently large $n$, $\ell \in \{1,\dots,m\}$ and $j\in \{j:\lambda_{j}\in I_{3}\}$.
Hence, by part $(i)$ of Lemma \ref{lemma: asymp of gamma for lambda bounded away from 1},
\begin{align*}
S_{2}^{(3)} & = \sum_{j:\lambda_{j}\in I_{3}} \ln \bigg( 1+\sum_{\ell=1}^{m} \omega_{\ell} \frac{ 1 + \bigO(e^{-\frac{a_{j}\eta_{j,\ell}^{2}}{2}}) }{ 1 + \bigO(e^{-\frac{a_{j}\eta_{j}^{2}}{2}}) } \bigg) = \sum_{j=j_{-}}^{g_{-}-1} \ln \Omega +\bigO(e^{-c^2M^2})
	\\
& = (g_{-}-j_{-})  \ln  \Omega +\bigO(e^{-c^2M^2})
\end{align*}
as $n \to +\infty$. Since
\begin{align*}
 g_{-}-j_{-} & = \bigg( \frac{bn \rho^{2b}}{1+\frac{M}{\sqrt{n}}} - \alpha \bigg)+\theta_{-}^{(n,M)} -j_{-} \nonumber \\
& = b\rho^{2b}n - j_{-} - bM\rho^{2b}\sqrt{n} + bM^{2}\rho^{2b}-\alpha+\theta_{-}^{(n,M)} - bM^{3}\rho^{2b}n^{-\frac{1}{2}} + \bigO(M^{4}n^{-1})
\end{align*}
as $n \to + \infty$, the desired conclusion follows.
\end{proof}

\begin{lemma}\label{lemma:S2p1p hard}
For any $x_{1},\dots,x_{m} \in \mathbb{R}$, there exists $\delta > 0$ such that
\begin{align*}
& S_{2}^{(1)} = D_{1}^{(\epsilon)} n + D_{2}^{(M)} \sqrt{n} + D_{3} \ln n + D_{4}^{(n,\epsilon,M)} + \frac{D_{5}^{(n,M)}}{\sqrt{n}} + \bigO\bigg( \frac{M^{4}}{n} + \frac{1}{\sqrt{n} M} + \frac{1}{M^{6}} + \frac{\sqrt{n}}{M^{11}} \bigg),
\end{align*}
as $n \to +\infty$ uniformly for $u_{1} \in D_\delta(x_1),\dots,u_{m} \in D_\delta(x_m)$, where
\begin{align*}
& D_{1}^{(\epsilon)} = \int_{b\rho^{2b}}^{\frac{b\rho^{2b}}{1-\epsilon}}f_{1}(x)dx, \qquad D_{2}^{(M)} = -b\rho^{2b} f_{1}(b\rho^{2b})M, \qquad D_{3} = - \frac{b\rho^{2b} \mathsf{T}_{1}(b\rho^{2b})}{2(1+\mathsf{T}_{0}(b\rho^{2b}))}, \\
& D_{4}^{(n,\epsilon,M)} = -b\rho^{2b} M^{2} \Big( f_{1}(b\rho^{2b}) + \frac{b\rho^{2b}}{2}f_{1}'(b\rho^{2b}) \Big) - \frac{b\rho^{2b} \mathsf{T}_{1}(b\rho^{2b})}{1+\mathsf{T}_{0}(b\rho^{2b})} \ln \bigg( \frac{\epsilon}{M(1-\epsilon)} \bigg) \\
& + \int_{b\rho^{2b}}^{\frac{b\rho^{2b}}{1-\epsilon}} \bigg\{ f(x) + \frac{b\rho^{2b}\mathsf{T}_{1}(b\rho^{2b})}{(1+\mathsf{T}_{0}(b\rho^{2b}))(x-b\rho^{2b})} \bigg\}dx + \bigg( \alpha - \frac{1}{2} + \theta_{+}^{(n,M)} \bigg)f_{1}(b \rho^{2b}) \\
& + \bigg( \frac{1}{2} - \alpha - \theta_{+}^{(n,\epsilon)} \bigg) f_{1}\bigg( \frac{b\rho^{2b}}{1-\epsilon} \bigg) + \frac{b \mathsf{T}_{1}(b\rho^{2b})}{M^{2}(1+\mathsf{T}_{0}(b\rho^{2b}))} + \frac{-5b \mathsf{T}_{1}(b\rho^{2b})}{2\rho^{2b}M^{4}(1+\mathsf{T}_{0}(b\rho^{2b}))}, \\
& D_{5}^{(n,M)} = -M^{3}b\rho^{2b} \bigg( f_{1}(b\rho^{2b}) + b\rho^{2b} f_{1}'(b\rho^{2b}) + \frac{(b\rho^{2b})^{2}}{6}f_{1}''(b\rho^{2b}) \bigg) + M b\rho^{2b} f_{1}'(b\rho^{2b}) \bigg( \alpha-\frac{1}{2}+\theta_{+}^{(n,M)} \bigg) \\
& + M \bigg( \frac{(b+\alpha)\rho^{2b}\mathsf{T}_{1}(b\rho^{2b})}{1+\mathsf{T}_{0}(b\rho^{2b})} - \frac{b\rho^{4b} \mathsf{T}_{2}(b\rho^{2b})}{2(1+\mathsf{T}_{0}(b\rho^{2b}))} + \frac{b\rho^{4b} \mathsf{T}_{1}(b\rho^{2b})^{2}}{(1+\mathsf{T}_{0}(b\rho^{2b}))^{2}} \bigg),
\end{align*}
where $f_{1}$ and $f$ are as in the statement of Lemma \ref{lemma:S3 asymp hard}.
\end{lemma}
\begin{proof}
We have
\begin{align}\label{S21Xj}
S_{2}^{(1)}  = \sum_{j= g_{+}+1}^{j_{+}} \ln ( 1+ X_j), \quad \text{where} \quad X_j := \frac{\sum_{\ell=1}^{m} \omega_{\ell} \gamma\big(a_j, a_j \lambda_{j, \ell}\big) }{ \gamma\big(a_j, a_j \lambda_j \big)}.
\end{align}
Since $\lambda_j \in [1-\epsilon, 1 - \frac{M}{\sqrt{n}})$ for $g_+ + 1 \leq j \leq j_+$ and $\lambda_{j,\ell} = \lambda_j(1-\frac{t_{\ell}}{n})$, we can apply part $(ii)$ of Lemma \ref{lemma: asymp of gamma for lambda bounded away from 1} to find, for each $N \geq 0$,
\begin{align}\label{Xjexpression}
X_j = \frac{\sum_{\ell=1}^{m} \omega_{\ell} \frac{e^{-\frac{a_j}{2}\eta_{j,\ell}^2}}{\sqrt{2\pi}} \big\{\sum_{k=0}^{N-1} \frac{S(\varphi_k(\lambda_{j,\ell}))}{a_j^{k+1/2}} + \bigO\big(\frac{1}{a_j^{N+1/2}}\big) +  \bigO\big(\frac{1}{(a_j \eta_{j,\ell}^2)^{N+1/2}}\big)\big\}}{\frac{e^{-\frac{a_j}{2}\eta_j^2}}{\sqrt{2\pi}} \big\{\sum_{k=0}^{N-1} \frac{S(\varphi_k(\lambda_j))}{a_j^{k+1/2}} + \bigO\big(\frac{1}{a_j^{N+1/2}}\big) +  \bigO\big(\frac{1}{(a_j \eta_j^2)^{N+1/2}}\big)\big\}  }.
\end{align}
Let $x := j/n$. For all sufficiently large $n$ we have $\eta_j \asymp \lambda_{j}-1$,\footnote{More precisely, this means that $\eta_j$ and $\lambda_{j}-1$ are of the same order in the sense that there exist constants $c_1, c_2 > 0$ such that $c_1 \leq \eta_j/(\lambda_{j}-1) \leq c_2$ for all sufficiently large $n$ and all $g_+ + 1 \leq j \leq j_+$.} $\eta_{j,\ell} \asymp \lambda_{j,\ell}-1 \asymp \lambda_{j}-1$, and
\begin{align*}
& x \in \bigg[\frac{b\rho^{2b}}{1-\frac{M}{\sqrt{n}}} + \bigO(n^{-1}), \frac{b\rho^{2b}}{1-\epsilon} + \bigO(n^{-1})\bigg], \qquad a_j = \frac{xn}{b} + \bigO(1),
\end{align*}
uniformly for $g_+ + 1 \leq j \leq j_+$.
Thus, multiplying both the numerator and denominator on the right-hand side of (\ref{Xjexpression}) by $-a_j^{1/2}(\lambda_{j} -1)$ and using that $S(\varphi_0(\lambda)) = -\frac{1}{\lambda-1}$, we find
\begin{align}\label{XjYjl}
X_j &= \sum_{\ell=1}^{m} \omega_{\ell} e^{-\frac{a_j}{2}(\eta_{j,\ell}^2 - \eta_j^2)} Y_{j,\ell}, \quad Y_{j,\ell} := \frac{\frac{\lambda_{j} -1}{\lambda_{j,\ell} -1} - (\lambda_{j} -1)\sum_{k=1}^{N-1} \frac{S(\varphi_k(\lambda_{j,\ell}))}{a_j^{k}} +  \bigO\big(\frac{1}{(n (\lambda_{j} -1)^2)^{N}}\big)}{1 - (\lambda_{j} -1) \sum_{k=1}^{N-1} \frac{S(\varphi_k(\lambda_j))}{a_j^{k}} +  \bigO\big(\frac{1}{(n (\lambda_{j} -1)^2)^{N}}\big)}.
\end{align}
Using that $a_j = \frac{xn + \alpha}{b}$, we can expand the exponential as $n \to +\infty$:
\begin{align}\label{expexpansion}
e^{-\frac{a_j}{2}(\eta_{j,\ell}^2 - \eta_j^2)}
= e^{a_j \ln(1 - \frac{t_{\ell}}{n}) + a_j \frac{b\rho^{2b} t_{\ell}}{nx + \alpha}}
= e^{-\frac{t_{\ell}}{b}(x - b \rho^{2b})}\bigg(1 - \frac{t_{\ell}^2 x + 2t_{\ell} \alpha}{2 b n} + \bigO\Big(\frac{1}{n^2}\Big)\bigg)
\end{align}
uniformly for $g_+ + 1 \leq j \leq j_+$. On the other hand, as $n \to +\infty$,
$$\lambda_{j,\ell} = \frac{b \rho^{2b}}{x}\bigg(1 - \frac{\alpha + x t_{\ell}}{x n} + \frac{\alpha(\alpha + x t_{\ell})}{x^2 n^2} + \bigO\Big(\frac{1}{n^3}\Big)\bigg), \qquad \lambda_j = \frac{b \rho^{2b}}{x}\bigg(1 - \frac{\alpha}{x n} + \frac{\alpha^2}{x^2 n^2} + \bigO\Big(\frac{1}{n^3}\Big)\bigg),$$
uniformly for $g_+ + 1 \leq j \leq j_+$. Substituting these expansions into the expression for $Y_{j,\ell}$ in (\ref{XjYjl}) with $N = 6$, a calculation gives
\begin{align}\nonumber
Y_{j,\ell} = &\; 1 -\frac{b \rho^{2 b} t_{\ell}}{n \left(x-b \rho^{2 b}\right)} + \frac{2 b^3 \rho^{4 b} t_{\ell}}{n^2 \left(x-b \rho^{2 b}\right)^3} + \bigO\bigg(\frac{1}{n^{2}(x-b\rho^{2b})^{2}}\bigg)
-\frac{10 b^5 \rho^{6 b} t_{\ell}}{n^3 \left(x-b \rho^{2 b}\right)^5}
	\\ \label{Yjlexpansion}
& + \bigO\bigg(\frac{1}{n^{3}(x-b\rho^{2b})^{4}}\bigg) + \bigO\bigg(\frac{1}{n^{4}(x-b\rho^{2b})^{7}} \bigg) + \bigO\bigg(\frac{1}{(n (x-b\rho^{2b})^2)^{6}}\bigg)
\end{align}
uniformly for $g_+ + 1 \leq j \leq j_+$. The asymptotic formulas (\ref{expexpansion}) and (\ref{Yjlexpansion}) imply that
\begin{align}\nonumber
X_j = &\; \mathsf{T}_0(x) -\frac{b \mathsf{T}_1(x) \rho^{2 b}}{n \left(x-b \rho^{2 b}\right)} -\frac{x\mathsf{T}_2(x) }{2 b n}-\frac{\alpha  \mathsf{T}_1(x)}{b n}
+ \frac{2 b^3 \mathsf{T}_1(x) \rho^{4 b}}{n^2 \left(x-b \rho^{2 b}\right)^3}
-\frac{10 b^5 \mathsf{T}_1(x) \rho^{6 b}}{n^3 \left(x-b \rho^{2 b}\right)^5}
	\\ \label{Xjexpansion}
& + \bigO\bigg(\frac{1}{n^{2}(x-b\rho^{2b})^{2}} + \frac{1}{n^{3}(x-b\rho^{2b})^{4}} + \frac{1}{n^{4}(x-b\rho^{2b})^{7}} + \frac{1}{n^{6}(x-b\rho^{2b})^{12}} \bigg).
\end{align}
If $A, B > 1$, then
\begin{align*}
\sum_{j= g_{+}+1}^{j_{+}} \bigO\bigg(\frac{1}{n^A (x - b\rho^{2b})^B}\bigg)
& = \bigO\bigg(\int_{g_+}^{j_+} \frac{1}{n^A (j/n - b\rho^{2b})^B} dj\bigg)
= \bigO\bigg(\int_{g_+/n}^{j_+/n} \frac{1}{n^{A-1} (x - b\rho^{2b})^B} dx\bigg)
	\\
&
= \bigO\bigg(\frac{1}{n^{A-1} (M/\sqrt{n})^{B-1}}\bigg)
= \bigO\bigg(\frac{1}{n^{A-(B+1)/2} M^{B-1}}\bigg),
\end{align*}
so substitution of (\ref{Xjexpansion}) into (\ref{S21Xj}) yields
\begin{align}\nonumber
 S_{2}^{(1)} = & \sum_{j= g_{+}+1}^{j_{+}} \bigg(  f_{1}(x)+\frac{1}{n}f(x) + \frac{1}{n^{2}} \frac{2b^{3}\rho^{4b} \mathsf{T}_{1}(x)}{(1+\mathsf{T}_{0}(x))(x-b\rho^{2b})^{3}} + \frac{1}{n^{3}} \frac{-10b^{5}\rho^{6b} \mathsf{T}_{1}(x)}{(1+\mathsf{T}_{0}(x))(x-b\rho^{2b})^{5}} \bigg)
 	\\ \label{S21sumjlargen}
& + \bigO\bigg( \frac{1}{M\sqrt{n}} + \frac{1}{M^{3}\sqrt{n}} + \frac{1}{M^{6}} + \frac{\sqrt{n}}{M^{11}} \bigg).
\end{align}
Employing Lemma \ref{lemma:Riemann sum NEW} with $A=\frac{b\rho^{2b}}{1-\frac{M}{\sqrt{n}}}$, $a_{0}=1-\alpha-\theta_{+}^{(n,M)}$, $B=\frac{b\rho^{2b}}{1-\epsilon}$ and $b_{0} = -\alpha-\theta_{+}^{(n,\epsilon)}$, and using that $f^{(k)}(A) = \bigO(n^{(k+1)/2} M^{-(k+1)})$ for $k \geq 0$, we get
\begin{align}\nonumber
& \sum_{j= g_{+}+1}^{j_{+}} f_{1}(x) = n \int_{\frac{b\rho^{2b}}{1-\frac{M}{\sqrt{n}}}}^{\frac{b\rho^{2b}}{1-\epsilon}}f_{1}(x)dx + \big(\alpha-\tfrac{1}{2}+\theta_{+}^{(n,M)}\big)f_{1}(\tfrac{b\rho^{2b}}{1-\frac{M}{\sqrt{n}}})+\big(\tfrac{1}{2}-\alpha-\theta_{+}^{(n,\epsilon)}\big)f_{1}(\tfrac{b\rho^{2b}}{1-\epsilon})+\bigO(n^{-1}), \\ \nonumber
& \frac{1}{n}\sum_{j= g_{+}+1}^{j_{+}} f(x) = \int_{\frac{b\rho^{2b}}{1-\frac{M}{\sqrt{n}}}}^{\frac{b\rho^{2b}}{1-\epsilon}}f(x)dx + \bigO\bigg(\frac{1}{M\sqrt{n}}\bigg), \\ \nonumber
& \frac{1}{n^{2}}\sum_{j= g_{+}+1}^{j_{+}}  \frac{2b^{3}\rho^{4b} \mathsf{T}_{1}(x)}{(1+\mathsf{T}_{0}(x))(x-b\rho^{2b})^{3}} = \frac{1}{n} \int_{\frac{b\rho^{2b}}{1-\frac{M}{\sqrt{n}}}}^{\frac{b\rho^{2b}}{1-\epsilon}} \frac{2b^{3}\rho^{4b} \mathsf{T}_{1}(x)dx}{(1+\mathsf{T}_{0}(x))(x-b\rho^{2b})^{3}}  + \bigO\bigg(\frac{1}{M^{3}\sqrt{n}}\bigg), \\ \label{f1fsumslargen}
& \frac{1}{n^{3}}\sum_{j= g_{+}+1}^{j_{+}} \frac{-10b^{5}\rho^{6b} \mathsf{T}_{1}(x)}{(1+\mathsf{T}_{0}(x))(x-b\rho^{2b})^{5}} = \frac{1}{n^{2}} \int_{\frac{b\rho^{2b}}{1-\frac{M}{\sqrt{n}}}}^{\frac{b\rho^{2b}}{1-\epsilon}} \frac{-10b^{5}\rho^{6b} \mathsf{T}_{1}(x)dx}{(1+\mathsf{T}_{0}(x))(x-b\rho^{2b})^{5}}  + \bigO\bigg(\frac{1}{M^{5}\sqrt{n}}\bigg).
\end{align}
The large $n$ behavior of the integrals in (\ref{f1fsumslargen}) can be determined as follows. Let us write
\begin{align}\label{nintebrho2bf1}
 n \int_{\frac{b\rho^{2b}}{1-\frac{M}{\sqrt{n}}}}^{\frac{b\rho^{2b}}{1-\epsilon}}f_{1}(x)dx
= &\; n \int_{b\rho^{2b}}^{\frac{b\rho^{2b}}{1-\epsilon}}f_{1}(x)dx - n \int_{b\rho^{2b}}^{\frac{b\rho^{2b}}{1-\frac{M}{\sqrt{n}}}} f_{1}(x)dx.
\end{align}
Using the integration by parts formula
\begin{align*}
\int_A^B f_1(x)dx
= \bigg((x-A)f_{1}(x) - \frac{(x-A)^2}{2!}f_{1}'(x) + \frac{(x-A)^3}{3!} f_{1}''(x)\bigg)\bigg|_A^B - \int_A^B \frac{(x-A)^3}{3!} f_1'''(x)dx
\end{align*}
with $A =b\rho^{2b} $ and $B = \frac{b\rho^{2b}}{1-\frac{M}{\sqrt{n}}}$ in the second integral in (\ref{nintebrho2bf1}), and then expanding as $n \to +\infty$, we obtain
\begin{align*}
 n \int_{\frac{b\rho^{2b}}{1-\frac{M}{\sqrt{n}}}}^{\frac{b\rho^{2b}}{1-\epsilon}}f_{1}(x)dx
= &\; n \int_{b\rho^{2b}}^{\frac{b\rho^{2b}}{1-\epsilon}}f_{1}(x)dx - b\rho^{2b}f_{1}(b \rho^{2b})M\sqrt{n} - M^{2} b \rho^{2b}\Big( f_{1}(b\rho^{2b})+ \frac{b \rho^{2b}}{2} f_{1}'(b \rho^{2b}) \Big) \\
& - \frac{M^{3}}{\sqrt{n}} b\rho^{2b} \Big( f_{1}(b\rho^{2b}) + b\rho^{2b}f_{1}'(b\rho^{2b}) + \frac{(b\rho^{2b})^{2}}{6}f_{1}''(b\rho^{2b})  \Big) + \bigO\bigg( \frac{M^{4}}{n} \bigg),
\end{align*}
where we have used that
$$n\int_A^B \frac{(x-A)^3}{3!} f_1'''(x)dx = \bigO(n(B-A)^4) = \bigO(M^4/n).$$
Similar calculations using that $\mathsf{T}_j^{(k)}(x) = (-1/b)^k \mathsf{T}_{j+k}(x)$ for $j,k \geq 0$ give
\begin{align*}
& \int_{\frac{b\rho^{2b}}{1-\frac{M}{\sqrt{n}}}}^{\frac{b\rho^{2b}}{1-\epsilon}}f(x)dx = \int_{b\rho^{2b}}^{\frac{b\rho^{2b}}{1-\epsilon}} \bigg\{ f(x) + \frac{b\rho^{2b} \mathsf{T}_{1}(b\rho^{2b})}{(1+\mathsf{T}_{0}(b\rho^{2b}))(x-b\rho^{2b})} \bigg\}dx - \frac{b\rho^{2b} \mathsf{T}_{1}(b\rho^{2b})}{2(1+\mathsf{T}_{0}(b\rho^{2b}))} \ln n \\
&\hspace{2.6cm} - \frac{b\rho^{2b} \mathsf{T}_{1}(b\rho^{2b})}{1+\mathsf{T}_{0}(b\rho^{2b})} \ln \frac{\epsilon}{M(1-\epsilon)} + \frac{M}{\sqrt{n}} \bigg\{ \frac{(b+\alpha) \rho^{2b} \mathsf{T}_{1}(b\rho^{2b})}{1+\mathsf{T}_{0}(b\rho^{2b})} - \frac{b \rho^{4b} \mathsf{T}_{2}(b\rho^{2b})}{2(1+\mathsf{T}_{0}(b\rho^{2b}))} \\
&\hspace{2.6cm} + \frac{b\rho^{4b} \mathsf{T}_{1}(b\rho^{2b})^{2}}{(1+\mathsf{T}_{0}(b\rho^{2b}))^{2}} \bigg\} + \bigO\bigg( \frac{M^{2}}{n} \bigg).
\end{align*}
Furthermore,
\begin{align*}
\frac{1}{n} \int_{\frac{b\rho^{2b}}{1-\frac{M}{\sqrt{n}}}}^{\frac{b\rho^{2b}}{1-\epsilon}} \frac{2b^{3}\rho^{4b} \mathsf{T}_{1}(x)}{(1+\mathsf{T}_{0}(x))(x-b\rho^{2b})^{3}} dx
& = \frac{1}{n} \int_{\frac{b\rho^{2b}}{1-\frac{M}{\sqrt{n}}}}^{\frac{b\rho^{2b}}{1-\epsilon}} \bigg(\frac{2b^{3}\rho^{4b} \mathsf{T}_{1}(b \rho^{2b})}{(1+\mathsf{T}_{0}(b \rho^{2b}))(x-b\rho^{2b})^{3}} + \bigO\Big(\frac{1}{(x-b\rho^{2b})^2}\Big)\bigg)dx
	\\
& = \frac{b \mathsf{T}_{1}(b \rho^{2b})}{M^{2}(1+\mathsf{T}_{0}(b \rho^{2b}))} + \bigO\bigg(\frac{1}{M \sqrt{n}}\bigg),
\end{align*}
and a similar calculation yields
\begin{align*}
& \frac{1}{n^{2}} \int_{\frac{b\rho^{2b}}{1-\frac{M}{\sqrt{n}}}}^{\frac{b\rho^{2b}}{1-\epsilon}} \frac{-10b^{5}\rho^{6b} \mathsf{T}_{1}(x)}{(1+\mathsf{T}_{0}(x))(x-b\rho^{2b})^{5}} dx = \frac{-5 b \mathsf{T}_{1}(b \rho^{2b})}{2\rho^{2b}M^{4}(1+\mathsf{T}_{0}(b \rho^{2b}))} + \bigO\bigg(\frac{1}{M^{3} \sqrt{n}}\bigg).
\end{align*}
Substituting the above expansions into (\ref{f1fsumslargen}), the claim follows from (\ref{S21sumjlargen}).
\end{proof}

For $k \in \{1,\dots,m\}$ and $j \in \{j: \lambda_{j} \in I_{2}\}=\{g_{-},\dots,g_{+}\}$, we define $M_{j,k} := \sqrt{n}(\lambda_{j,k}-1)$ and $M_{j} := \sqrt{n}(\lambda_{j}-1)$. For the large $n$ asymptotics of $\smash{S_{2}^{(2)}}$ we will need the following lemma.
\begin{lemma}(Taken from \cite[Lemma 3.11]{Charlier 2d gap})\label{lemma:Riemann sum}
Let $h \in C^{3}(\mathbb{R})$. As $n \to + \infty$, we have
\begin{align}
& \sum_{j=g_{-}}^{g_{+}}h(M_{j}) = b\rho^{2b} \int_{-M}^{M} h(t) dt \; \sqrt{n} - 2 b \rho^{2b} \int_{-M}^{M} th(t) dt + \bigg( \frac{1}{2}-\theta_{-}^{(n,M)} \bigg)h(M)+ \bigg( \frac{1}{2}-\theta_{+}^{(n,M)} \bigg)h(-M) \nonumber \\
& + \frac{1}{\sqrt{n}}\bigg[ 3b\rho^{2b} \int_{-M}^{M}t^{2}h(t)dt + \bigg( \frac{1}{12}+\frac{\theta_{-}^{(n,M)}(\theta_{-}^{(n,M)}-1)}{2} \bigg)\frac{h'(M)}{b\rho^{2b}} - \bigg( \frac{1}{12}+\frac{\theta_{+}^{(n,M)}(\theta_{+}^{(n,M)}-1)}{2} \bigg)\frac{h'(-M)}{b\rho^{2b}} \bigg] \nonumber \\
& + \bigO\Bigg(  \frac{1}{n^{3/2}} \sum_{j=g_{-}+1}^{g_{+}} \bigg( (1+|M_{j}|^{3}) \tilde{\mathfrak{m}}_{j,n}(h) + (1+M_{j}^{2})\tilde{\mathfrak{m}}_{j,n}(h') + (1+|M_{j}|) \tilde{\mathfrak{m}}_{j,n}(h'') + \tilde{\mathfrak{m}}_{j,n}(h''') \bigg)   \Bigg), \label{sum f asymp 2}
\end{align}
where, for $\tilde{h} \in C(\mathbb{R})$ and $j \in \{g_{-}+1,\dots,g_{+}\}$, we define $\tilde{\mathfrak{m}}_{j,n}(\tilde{h}) := \max_{x \in [M_{j},M_{j-1}]}|\tilde{h}(x)|$.
\end{lemma}

\begin{lemma}\label{lemma:S2kp2p hard}
For any $x_{1},\dots,x_{m} \in \mathbb{R}$, there exists $\delta > 0$ such that
\begin{align*}
&  S_{2}^{(2)} = E_{2}^{(M)} \sqrt{n} + E_{4}^{(M)} + \frac{E_{5}^{(M)}}{\sqrt{n}} + \bigO\bigg( \frac{M^{4}}{n} + \frac{M^{14}}{n^{2}} \bigg), \\
& E_{2}^{(M)} = 2b\rho^{2b} M \ln(1+\mathsf{T}_{0}(b\rho^{2b})), \\
& E_{4}^{(M)} = \ln(1+\mathsf{T}_{0}(b\rho^{2b})) \big( 1-\theta_{-}^{(n,M)}-\theta_{+}^{(n,M)} \big) + b \rho^{2b} \int_{-M}^{M} h_{1}(t)dt, \\
& E_{5}^{(M)} = 2b\rho^{2b}M^{3} \ln(1+\mathsf{T}_{0}(b\rho^{2b})) + \bigg( \frac{1}{2}-\theta_{-}^{(n,M)} \bigg) h_{1}(M) + \bigg( \frac{1}{2}-\theta_{+}^{(n,M)} \bigg) h_{1}(-M) \\
& \hspace{1.15cm} + b \rho^{2b} \int_{-M}^{M} \big( h_{2}(t)-2th_{1}(t) \big)dt,
\end{align*}
as $n \to +\infty$ uniformly for $u_{1} \in D_\delta(x_1),\dots,u_{m} \in D_\delta(x_m)$, where $h_{1}$, $h_{2}$ are given by
\begin{align}
& h_{1}(x) = - \frac{2\rho^{b} \mathsf{T}_{1}(b\rho^{2b})}{1+\mathsf{T}_{0}(b\rho^{2b})} \frac{e^{-\frac{1}{2}x^{2}\rho^{2b}}}{\sqrt{2\pi} \, \mathrm{erfc}(-\frac{x \rho^{b}}{\sqrt{2}})}, \label{def of h1 HARD} \\
& h_{2}(x) = -\frac{h_{1}(x)^{2}}{2} + \frac{1}{1+\mathsf{T}_{0}(b\rho^{2b})} \frac{e^{-\frac{1}{2}x^{2}\rho^{2b}}}{\sqrt{2\pi} \, \mathrm{erfc}(-\frac{x \rho^{b}}{\sqrt{2}})} \bigg\{ \Big( \rho^{b} x - \frac{5}{3}\rho^{3b}x^{3} \Big) \mathsf{T}_{1}(b\rho^{2b}) \nonumber \\
& \hspace{1.25cm} -\rho^{3b}x \mathsf{T}_{2}(b\rho^{2b}) + \frac{4-10\rho^{2b}x^{2}}{3} \mathsf{T}_{1}(b\rho^{2b}) \frac{e^{-\frac{1}{2}x^{2}\rho^{2b}}}{\sqrt{2\pi}\, \mathrm{erfc}(-\frac{x\rho^{b}}{\sqrt{2}})} \bigg\}. \nonumber
\end{align}
\end{lemma}
\begin{proof}
Using \eqref{asymp prelim of S2kpvp hard} and Lemma \ref{lemma: uniform}, we obtain
\begin{align}\label{lol1 hard}
& S_{2}^{(2)} = \sum_{j:\lambda_{j}\in I_{2}} \ln \bigg( 1+\sum_{\ell=1}^{m} \omega_{\ell} \frac{ \frac{1}{2}\mathrm{erfc}\Big(-\eta_{j,\ell} \sqrt{a_{j}/2}\Big) - R_{a_{j}}(\eta_{j,\ell}) }{ \frac{1}{2}\mathrm{erfc}\big(-\eta_{j} \sqrt{a_{j}/2}\big) - R_{a_{j}}(\eta_{j}) } \bigg).
\end{align}
For $j \in \{j:\lambda_{j}\in I_{2}\}$, we have $1-\frac{M}{\sqrt{n}} \leq \lambda_{j} = \frac{bn\rho^{2b}}{j+\alpha} \leq 1+\frac{M}{\sqrt{n}}$, $-M \leq M_{j} \leq M$, and
\begin{align*}
M_{j,k} = M_{j} - \frac{t_{k}}{\sqrt{n}} - \frac{t_{k}M_{j}}{n}, \qquad k=1,\dots,m.
\end{align*}
Furthermore, as $n \to + \infty$ we have
\begin{align}\nonumber
\eta_{j,\ell} = &\; \frac{M_{j}}{\sqrt{n}} - \frac{M_{j}^{2}+3t_{\ell}}{3n} + \frac{7M_{j}^{3}-12 t_{\ell} M_{j}}{36 n^{3/2}}
-\frac{73 M_j^4-45 M_j^2 t_{\ell}+180 t_{\ell}^2}{540 n^2}
	\\
& + \frac{1331 M_j^5-552 M_j^3 t_{\ell}-1080 M_j t_{\ell}^2}{12960
   n^{5/2}}
+ \bigO\bigg(\frac{1+M_{j}^{6}}{n^3}\bigg) \label{asymp etaj and etajsqrtajover2 1 hard}
	\\ \nonumber
-\eta_{j,\ell} \sqrt{a_{j}/2} = & - \frac{M_{j}\rho^{b}}{\sqrt{2}} + \frac{(5M_{j}^{2}+6t_{\ell})\rho^{b}}{6\sqrt{2} \sqrt{n}} - \frac{\rho^{b} M_{j}(53M_{j}^{2}+12t_{\ell})}{72\sqrt{2} n}
+ \frac{\rho^b \left(270 M_j^2 t_{\ell}+1447 M_j^4+720
   t_{\ell}^2\right)}{2160 \sqrt{2} n^{3/2}}
   	\\
 & -\frac{M_j \rho^b \left(5352 M_j^2 t_{\ell}+32183 M_j^4+4320
   t_{\ell}^2\right)}{51840 \sqrt{2} n^2}
+ \bigO\bigg( \frac{1+M_{j}^{6}}{n^{5/2}} \bigg) \label{asymp etaj and etajsqrtajover2 2 hard}
\end{align}
uniformly for $j\in \{j:\lambda_{j}\in I_{2}\}$. Hence, by \eqref{asymp of Ra}, as $n \to + \infty$ we have
\begin{align}\nonumber
R_{a_{j}}(\eta_{j,\ell}) = &\; \frac{e^{-\frac{M_{j}^{2}\rho^{2b}}{2}}}{\sqrt{2\pi}} \bigg\{ \frac{-1}{3\rho^{b}\sqrt{n}} - \frac{M_{j}(3+10M_{j}^{2}\rho^{2b}+12t_{\ell}\rho^{2b})}{36\rho^{b}n} 	
	\\ \nonumber
& + \frac{45 \rho^{4 b} (6 M_j^2 t_{\ell}+7 M_j^4+4
   t_{\ell}^2)+2 \rho^{2 b} (22 M_j^2-45 t_{\ell})-5 \rho^{6
   b} (5 M_j^3+6 M_j t_{\ell})^2-2}{1080 \rho^{3 b} n^{3/2}}
   	\\\nonumber
& + \frac{M_j \rho^{-3 b}}{38880 n^2}\Big(-6 \rho^{4 b} (1806 M_j^2 t_{\ell}+1967
   M_j^4+1350 t_{\ell}^2)
   	\\\nonumber
&   +45 \rho^{6 b} (5 M_j^2+6
   t_{\ell}) (42 M_j^2 t_{\ell}+47 M_j^4+24
   t_{\ell}^2)-36 \rho^{2 b} (29 M_j^2+45 t_{\ell})
   	\\
&   -10 \label{Rajetajlexpansion}
   M_j^2 \rho^{8 b} (5 M_j^2+6 t_{\ell})^3-243\Big)
+ \bigO((1+M_{j}^{12})n^{-\frac{5}{2}}) \bigg\}
\end{align}
and
\begin{align}\nonumber
\frac{1}{2}\mathrm{erfc}\Big(-\eta_{j,\ell} &\sqrt{a_{j}/2}\Big) = \frac{1}{2}\mathrm{erfc}\Big(-\frac{\rho^{b}M_{j}}{\sqrt{2}}\Big) -\frac{e^{-\frac{M_{j}^{2}\rho^{2b}}{2}}\rho^{b}(5 M_{j}^{2} - 6 t_{\ell})}{6\sqrt{2\pi}\sqrt{n}}
	\\ \nonumber
& + \frac{e^{-\frac{M_{j}^{2}\rho^{2b}}{2}}M_{j} \rho^{b}}{72\sqrt{2\pi} \, n} \Big( 53M_{j}^{2} + 12 t_{\ell} - 25 M_{j}^{4} \rho^{2b} - 60 M_{j}^{2} t_{\ell} \rho^{2b} - 36 t_{\ell}^{2}\rho^{2b} \Big)
	\\ \label{erfcetajlexpansion}
& + \frac{e^{-\frac{M_{j}^{2}\rho^{2b}}{2}}P_8(M_j, t_\ell)}{n^{3/2}}
+ \frac{e^{-\frac{M_{j}^{2}\rho^{2b}}{2}}P_{11}(M_j, t_\ell)}{n^{2}}
+ \bigO\Big(e^{-\frac{M_{j}^{2}\rho^{2b}}{2}}\frac{1+M_{j}^{14}}{n^{5/2}} \Big),
\end{align}
uniformly for $j\in \{j:\lambda_{j}\in I_{2}\}$, where $P_8(M_j, t_\ell)$ and $P_{11}(M_j, t_\ell)$ are polynomials in $M_j$ of order $8$ and $11$, respectively. If $t_\ell = 0$, then $\lambda_{j,\ell} = \lambda_j$ and $\eta_{j,\ell} = \eta_j$; hence analogous expansions of $R_{a_{j}}(\eta_j)$ and $\frac{1}{2}\mathrm{erfc}(-\eta_{j} \sqrt{a_{j}/2})$ can be obtained by setting $t_\ell = 0$ in (\ref{Rajetajlexpansion}) and (\ref{erfcetajlexpansion}).
Substituting the above asymptotics into \eqref{lol1 hard}, we obtain
\begin{align}
1+\sum_{\ell=1}^{m} \omega_{\ell} \frac{ \frac{1}{2}\mathrm{erfc}\Big(-\eta_{j,\ell} \sqrt{a_{j}/2}\Big) - R_{a_{j}}(\eta_{j,\ell}) }{ \frac{1}{2}\mathrm{erfc}\big(-\eta_{j} \sqrt{a_{j}/2}\big) - R_{a_{j}}(\eta_{j}) }   = &\; g_{1}(M_{j}) + \frac{g_{2}(M_{j})}{\sqrt{n}} + \frac{g_{3}(M_{j})}{n} \nonumber \\
& + \frac{g_{4}(M_{j})}{n^{3/2}} + \frac{g_{5}(M_{j})}{n^{2}} + \bigO\Big(\frac{1+|M_{j}|^{13}}{n^{5/2}}\Big), \label{asymp of S2kp2p in proof hard}
\end{align}
as $n \to + \infty$, where
\begin{align*}
& g_{1}(x) = 1+\mathsf{T}_{0}(b\rho^{2b}), \qquad g_{2}(x) = - \frac{e^{-\frac{1}{2}x^{2}\rho^{2b}}2\rho^{b} \mathsf{T}_{1}(b\rho^{2b})}{\sqrt{2\pi} \mathrm{erfc}(-\frac{x\rho^{b}}{\sqrt{2}})}, \\
& g_{3}(x) = \frac{e^{-\frac{1}{2}x^{2}\rho^{2b}}}{3\sqrt{2\pi} \, \mathrm{erfc}(-\frac{x\rho^{b}}{\sqrt{2}})} \bigg\{ \frac{e^{-\frac{1}{2}x^{2}\rho^{2b}} \mathsf{T}_{1}(b\rho^{2b})}{\sqrt{2\pi} \, \mathrm{erfc}(-\frac{x\rho^{b}}{\sqrt{2}})} (4-10x^{2}\rho^{2b}) + \mathsf{T}_{1}(b\rho^{2b}) \big( 3x\rho^{b}-5x^{3}\rho^{3b} \big) \\
& \hspace{4.1cm} -3 \rho^{3b}x \mathsf{T}_{2}(b\rho^{2b}) \bigg\}.
\end{align*}
The functions $g_{4}$ and $g_{5}$ can also be computed explicitly, but we do not write them down.
The functions $g_j(x)$, $j = 2, \dots, 5$, have exponential decay as $x \to +\infty$.
Also, since
\begin{align}\label{trivial asymp}
\frac{e^{-\frac{1}{2}x^{2}\rho^{2b}}}{\sqrt{2\pi} \, \mathrm{erfc}(-\frac{x\rho^{b}}{\sqrt{2}})} = -\frac{\rho^{b} x}{2} + \bigO(x^{-1}), \qquad \mbox{as } x \to -\infty,
\end{align}
$g_{2}(x) = \bigO(x)$ as $x \to -\infty$. It appears at first sight that $g_{3}(x) = \bigO(x^{4})$ as $x \to -\infty$. However, a direct computation using \eqref{trivial asymp} shows that some cancellations occur and in fact $g_{3}(x) = \bigO(x^{2})$ as $x \to -\infty$. Similarly, the exact expressions for $g_{4}$ and $g_{5}$ suggest at first sight that $g_{4}(x) = \bigO(x^{7})$ and $g_{5}(x) = \bigO(x^{10})$ as $x \to -\infty$, but here too, cancellations occur and in fact we have $g_{4}(x) = \bigO(x^{3})$ and $g_{5}(x) = \bigO(x^{4})$ as $x \to -\infty$. Thus, after a computation using \eqref{asymp of S2kp2p in proof hard}, we obtain
\begin{align*}
& S_{2}^{(2)} = \sum_{j=g_{-}}^{g_{+}} \bigg\{ \ln(1+\mathsf{T}_{0}(b\rho^{2b})) + \frac{h_{1}(M_{j})}{\sqrt{n}} + \frac{h_{2}(M_{j})}{n} + \bigO\bigg( \frac{1+|M_{j}|^{3}}{n^{3/2}} + \frac{1+|M_{j}|^{13}}{n^{5/2}} \bigg) \bigg\}.
\end{align*}
as $n \to + \infty$, where $h_1 = g_2/g_1$ and $h_2 = -h_1^2/2 + g_3/g_1$. Note that
\begin{align*}
\sum_{j=g_{-}}^{g_{+}}\bigO\bigg( \frac{1+|M_{j}|^{3}}{n^{3/2}} + \frac{1+|M_{j}|^{13}}{n^{5/2}} \bigg) = \bigO\bigg( \frac{M^{4}}{n} + \frac{M^{14}}{n^{2}} \bigg), \qquad \mbox{as } n \to + \infty.
\end{align*}
Using Lemma \ref{lemma:Riemann sum}, we find the claim.
\end{proof}
Let us define
\begin{align}
& \mathcal{I}_1 = \int_{-\infty}^{+\infty} \bigg\{ \frac{e^{-y^{2}}}{\sqrt{\pi}\, \mathrm{erfc}(y)} - \chi_{(0,+\infty)}(y) \bigg[ y + \frac{y}{2(1+y^{2})} \bigg] \bigg\}dy, \label{def of I1}
	\\
& \mathcal{I}_{2} = \int_{-\infty}^{+\infty} \bigg\{ \frac{y^{3}e^{-y^{2}}}{\sqrt{\pi} \, \mathrm{erfc}(y)} - \chi_{(0,+\infty)}(y) \bigg[ y^{4}+\frac{y^{2}}{2}-\frac{1}{2} \bigg] \bigg\}dy, \label{def of I3} \\
& \mathcal{I}_{3} = \int_{-\infty}^{+\infty} \bigg\{ \bigg( \frac{e^{-y^{2}}}{\sqrt{\pi} \, \mathrm{erfc}(y)} \bigg)^{2} - \chi_{(0,+\infty)}(y) \bigg[ y^{2}+1 \bigg] \bigg\} dy, \label{def of I4}
	\\
& \mathcal{I}_{4} = \int_{-\infty}^{+\infty} \bigg\{ \bigg( \frac{y \, e^{-y^{2}}}{\sqrt{\pi} \, \mathrm{erfc}(y)} \bigg)^{2} - \chi_{(0,+\infty)}(y)\bigg[ y^{4}+y^{2}-\frac{3}{4} \bigg] \bigg\} dy, \label{def of I5}
\end{align}
and recall that $\mathcal{I}$ is defined in \eqref{def of I}.

\begin{lemma}\label{lemma: asymp of S2k final hard}
The constant $M'$ can be chosen sufficiently large such that the following holds. For any $x_{1},\dots,x_{m} \in \mathbb{R}$, there exists $\delta > 0$ such that
\begin{align*}
& S_{2} =  - j_{-} \ln \Omega + C_{1}^{(\epsilon)}n + C_{2} \ln n + C_{3}^{(n,\epsilon)} + \frac{\widehat{C}_{4}}{\sqrt{n}}  + \bigO\bigg(\frac{\sqrt{n}}{M^{11}} +  \frac{1}{M^{6}} + \frac{1}{\sqrt{n} M} + \frac{M^{4}}{n} + \frac{M^{14}}{n^{2}} \bigg),
\end{align*}
as $n \to +\infty$ uniformly for $u_{1} \in D_\delta(x_1),\dots,u_{m} \in D_\delta(x_m)$, where $C_{2}$ is as in the statement of Theorem \ref{thm:main thm hard} and
\begin{align*}
 C_{1}^{(\epsilon)} =&\; b \rho^{2b} \ln \Omega + \int_{b\rho^{2b}}^{\frac{b\rho^{2b}}{1-\epsilon}} f_{1}(x)dx, \\
 C_{3}^{(n,\epsilon)} = &\; \frac{1}{2}\ln \Omega + \int_{b\rho^{2b}}^{\frac{b\rho^{2b}}{1-\epsilon}} \bigg\{ f(x) + \frac{b\rho^{2b} \mathsf{T}_{1}(b\rho^{2b})}{\Omega (x-b\rho^{2b})} \bigg\}dx + \bigg( \frac{1}{2}-\alpha - \theta_{+}^{(n,\epsilon)}  \bigg) f_{1}\bigg(\frac{b\rho^{2b}}{1-\epsilon}\bigg) \\
&  - \frac{2b\rho^{2b}}{\Omega}\mathsf{T}_{1}(b\rho^{2b}) \mathcal{I}_{1} + \frac{b\rho^{2b}}{2\Omega}\mathsf{T}_{1}(b \rho^{2b}) \big( \ln 2 - 2b \ln(\rho) \big) - \frac{\mathsf{T}_{1}(b\rho^{2b})}{\Omega}b\rho^{2b} \ln \bigg( \frac{\epsilon}{1-\epsilon} \bigg), \\
\widehat{C}_{4} = & \; \sqrt{2}b\rho^{b}\frac{ \rho^{2b} \mathsf{T}_{2}(b\rho^{2b}) -5 \mathsf{T}_{1}(b\rho^{2b}) }{ \Omega} \mathcal{I}  + \frac{10\sqrt{2} b \rho^{b}}{3} \frac{\mathsf{T}_{1}(b\rho^{2b})}{\Omega} \mathcal{I}_{2}
	\\
& + \sqrt{2} b \rho^{2b} \frac{\mathsf{T}_{1}(b\rho^{2b})}{\Omega} \bigg( \frac{2}{3\rho^{b}} -  \rho^{b} \frac{\mathsf{T}_{1}(b\rho^{2b})}{\Omega} \bigg) \mathcal{I}_{3}
 - \frac{10 \sqrt{2} b \rho^{b}}{3} \frac{\mathsf{T}_{1}(b\rho^{2b})}{\Omega} \mathcal{I}_{4},
\end{align*}
and $f_{1}$ and $f$ are as in the statement of Lemma \ref{lemma:S3 asymp hard}.
\end{lemma}
\begin{proof}
By combining Lemmas \ref{lemma:S2kp3p hard}, \ref{lemma:S2p1p hard} and \ref{lemma:S2kp2p hard}, we have
\begin{align*}
S_{2} = & - j_{-} \ln \Omega + C_{1}^{(\epsilon)}n + \widetilde{C}_{2}\sqrt{n} + C_{2} \ln n + C_{3}^{(n,\epsilon,M)} + \frac{C_{4}^{(M)}}{\sqrt{n}}
	\\
& + \bigO\bigg(\frac{\sqrt{n}}{M^{11}} + \frac{1}{M^{6}} + \frac{1}{\sqrt{n} M} + \frac{M^{4}}{n} + \frac{M^{14}}{n^{2}} \bigg),
\end{align*}
as $n \to +\infty$ uniformly for $u_{1} \in D_\delta(x_1),\dots,u_{m} \in D_\delta(x_m)$, where $C_{1}^{(\epsilon)}$ is as in the statement, and
\begin{align*}
& \widetilde{C}_{2} = - bM \rho^{2b} \ln \Omega + D_{2}^{(M)}+E_{2}^{(M)}, \\
& C_{3}^{(n,\epsilon,M)} = \big( bM^{2}\rho^{2b} - \alpha + \theta_{-}^{(n,M)} \big) \ln \Omega+ D_4^{(n,\epsilon,M)} + E_4^{(M)}, \\
& C_{4}^{(n,M)} = -bM^{3}\rho^{2b} \ln \Omega + D_{5}^{(n,M)} + E_{5}^{(M)}.
\end{align*}
Using that $f_{1}(b \rho^{2b}) = \ln(1+\mathsf{T}_{0}(b\rho^{2b})) = \ln \Omega$, we readily verify that $\widetilde{C}_{2}=0$. Furthermore, by rearranging the terms and using $f_{1}'(b\rho^{2b}) = \frac{\frac{-1}{b}\mathsf{T}_{1}(b\rho^{2b})}{1+\mathsf{T}_{0}(b\rho^{2b})}$, we obtain
\begin{align*}
C_{3}^{(n,\epsilon,M)} & = \frac{1}{2} \ln \Omega + \widetilde{C}_{3}^{(\epsilon,M)} + \int_{b\rho^{2b}}^{\frac{b\rho^{2b}}{1-\epsilon}} \bigg\{ f(x) + \frac{b\rho^{2b}\mathsf{T}_{1}(b\rho^{2b})}{(1+\mathsf{T}_{0}(b\rho^{2b}))(x-b\rho^{2b})} \bigg\}dx \\
& + \bigg( \frac{1}{2}-\alpha-\theta_{+}^{(n,\epsilon)} \bigg) f_{1}\bigg( \frac{b\rho^{2b}}{1-\epsilon} \bigg),
\end{align*}
where
\begin{align*}
& \widetilde{C}_{3}^{(\epsilon,M)} := b\rho^{2b}\int_{-M}^{M}h_{1}(t)dt + \frac{\mathsf{T}_{1}(b\rho^{2b})}{1+\mathsf{T}_{0}(b\rho^{2b})}\bigg( M^{2}\frac{b\rho^{4b}}{2} - b\rho^{2b} \ln \bigg( \frac{\epsilon}{M(1-\epsilon)} \bigg) + \frac{b}{M^{2}} + \frac{-5b}{2\rho^{2b}M^{4}} \bigg).
\end{align*}
Using the definition \eqref{def of h1 HARD} of $h_{1}$ and a change of variables, we rewrite $\widetilde{C}_{3}^{(\epsilon,M)}$ as
\begin{align*}
& \widetilde{C}_{3}^{(\epsilon,M)} = -2b \rho^{2b} \frac{\mathsf{T}_{1}(b\rho^{2b})}{1+\mathsf{T}_{0}(b\rho^{2b})} \int_{-\frac{M\rho^{b}}{\sqrt{2}}}^{\frac{M\rho^{b}}{\sqrt{2}}} \bigg\{ \frac{e^{-y^{2}}}{\sqrt{\pi} \, \mathrm{erfc}(y)} - \chi_{(0,+\infty)}(y) \bigg[ y + \frac{y}{2(1+y^{2})} + \frac{3y}{4(1+y^{6})} \bigg] \bigg\}dy \\
& + \frac{\mathsf{T}_{1}(b\rho^{2b})}{1+\mathsf{T}_{0}(b\rho^{2b})} \bigg\{ -2b \rho^{2b} \int_{0}^{\frac{M\rho^{b}}{\sqrt{2}}} \bigg( y + \frac{y}{2(1+y^{2})} + \frac{3y}{4(1+y^{6})} \bigg)dy + M^{2}\frac{b\rho^{4b}}{2} + b\rho^{2b} \ln M \\
& + \frac{b}{M^{2}} + \frac{-5b}{2\rho^{2b}M^{4}} \bigg\} - \frac{\mathsf{T}_{1}(b\rho^{2b})}{1+\mathsf{T}_{0}(b\rho^{2b})} b \rho^{2b} \ln \frac{\epsilon}{1-\epsilon}.
\end{align*}
The reason for the above rewriting stems from the following asymptotics:
\begin{align*}
\frac{e^{-y^{2}}}{\sqrt{\pi} \, \mathrm{erfc}(y)} - \bigg[ y + \frac{y}{2(1+y^{2})} + \frac{3y}{4(1+y^{6})} \bigg] = \bigO(y^{-7}), \qquad \mbox{as } y \to + \infty,
\end{align*}
which implies
\begin{align*}
& \int_{-\frac{M\rho^{b}}{\sqrt{2}}}^{\frac{M\rho^{b}}{\sqrt{2}}} \bigg\{ \frac{e^{-y^{2}}}{\sqrt{\pi} \, \mathrm{erfc}(y)} - \chi_{(0,+\infty)}(y) \bigg[ y + \frac{y}{2(1+y^{2})} + \frac{3y}{4(1+y^{6})} \bigg] \bigg\}dy
	\\
& = \int_{-\infty}^{\infty} \bigg\{ \frac{e^{-y^{2}}}{\sqrt{\pi} \, \mathrm{erfc}(y)} - \chi_{(0,+\infty)}(y) \bigg[ y + \frac{y}{2(1+y^{2})} + \frac{3y}{4(1+y^{6})} \bigg] \bigg\}dy + \bigO(M^{-6})
	\\
& = \int_{-\infty}^{\infty} \bigg\{ \frac{e^{-y^{2}}}{\sqrt{\pi} \, \mathrm{erfc}(y)} - \chi_{(0,+\infty)}(y) \bigg[ y + \frac{y}{2(1+y^{2})} \bigg] \bigg\}dy - \frac{\pi}{4\sqrt{3}} + \bigO(M^{-6}), \qquad \mbox{as } n \to + \infty.
\end{align*}
Furthermore, using a primitive and then expanding yields
\begin{align*}
& -2b \rho^{2b} \int_{0}^{\frac{M\rho^{b}}{\sqrt{2}}} \bigg( y + \frac{y}{2(1+y^{2})} + \frac{3y}{4(1+y^{6})} \bigg)dy + M^{2}\frac{b\rho^{4b}}{2} + b\rho^{2b} \ln M + \frac{b}{M^{2}} + \frac{-5b}{2\rho^{2b}M^{4}} \\
& = - \frac{b\rho^{2b}}{6} \bigg( \sqrt{3} \, \pi - 3 \ln 2 + 6b \ln \rho \bigg) + \bigO(M^{-6}), \qquad \mbox{as } n \to + \infty.
\end{align*}
It follows from the above and some further simplifications that
\begin{align*}
C_{3}^{(n,\epsilon,M)} = C_{3}^{(n,\epsilon)} + \bigO(M^{-6}), \qquad \mbox{as } n \to + \infty,
\end{align*}
where $C_{3}^{(n,\epsilon)}$ is as in the statement. Similar (but longer) computation, using among other things that
\begin{align*}
f_{1}''(b\rho^{2b}) = - \bigg( \frac{\frac{-1}{b}\mathsf{T}_{1}(b\rho^{2b})}{\Omega} \bigg)^{2} + \frac{(-\frac{1}{b})^{2}\mathsf{T}_{2}(b\rho^{2b})}{\Omega},
\end{align*}
show that $C_{4}^{(n,M)}$ can be rewritten as
\begin{align}\label{lol2}
C_{4}^{(n,M)} = Q_{1}^{(n,M)} + Q_{2}^{(n,M)} + Q_{3}^{(M)} + Q_{4}^{(M)} + Q_{5}^{(M)} + Q_{6}^{(M)},
\end{align}
where
\begin{align*}
& Q_{1}^{(n,M)} = - \frac{2\rho^{b}\mathsf{T}_{1}(b\rho^{2b})}{\Omega} \bigg( \frac{1}{2}-\theta_{-}^{(n,M)} \bigg) \frac{e^{-\frac{M^{2}\rho^{2b}}{2}}}{\sqrt{2\pi} \, \mathrm{erfc}(-\frac{M\rho^{b}}{\sqrt{2}})}, \\
& Q_{2}^{(n,M)} = - \frac{2\rho^{b}\mathsf{T}_{1}(b\rho^{2b})}{\Omega} \bigg( \frac{1}{2}-\theta_{+}^{(n,M)} \bigg) \bigg( \frac{e^{-\frac{M^{2}\rho^{2b}}{2}}}{\sqrt{2\pi} \, \mathrm{erfc}(\frac{M\rho^{b}}{\sqrt{2}})} - \frac{M\rho^{b}}{2} \bigg), \\
& Q_{3}^{(M)} = \frac{\sqrt{2}\, b \rho^{b}}{\Omega} \big( -5\mathsf{T}_{1}(b\rho^{2b}) + \rho^{2b} \mathsf{T}_{2}(b\rho^{2b}) \big) \int_{-\frac{M\rho^{b}}{\sqrt{2}}}^{\frac{M\rho^{b}}{\sqrt{2}}} \bigg\{ \frac{ye^{-y^{2}}}{\sqrt{\pi}\, \mathrm{erfc}(y)} - \chi_{(0,+\infty)}(y) \bigg[ y^{2}+\frac{1}{2} \bigg] \bigg\}dy, \\
& Q_{4}^{(M)} = \frac{10\sqrt{2} \, b \rho^{b}}{3\Omega} \mathsf{T}_{1}(b\rho^{2b}) \int_{-\frac{M\rho^{b}}{\sqrt{2}}}^{\frac{M\rho^{b}}{\sqrt{2}}} \bigg\{ \frac{y^{3}e^{-y^{2}}}{\sqrt{\pi}\, \mathrm{erfc}(y)} - \chi_{(0,+\infty)}(y) \bigg[ y^{4} + \frac{y^{2}}{2} - \frac{1}{2} \bigg] \bigg\}dy, \\
& Q_{5}^{(M)} = \sqrt{2} \, b \rho^{b} \frac{\mathsf{T}_{1}(b\rho^{2b})}{\Omega}\bigg( \frac{2}{3} - \rho^{2b} \frac{\mathsf{T}_{1}(b\rho^{2b})}{\Omega} \bigg) \int_{-\frac{M\rho^{b}}{\sqrt{2}}}^{\frac{M\rho^{b}}{\sqrt{2}}} \bigg\{ \bigg( \frac{ e^{-y^{2}}}{\sqrt{\pi}\, \mathrm{erfc}(y)} \bigg)^{2} - \chi_{(0,+\infty)}(y) \bigg[ y^{2} + 1 \bigg] \bigg\}dy, \\
& Q_{6}^{(M)} = - \frac{10\sqrt{2}\, b  \rho^{b}}{3}\frac{\mathsf{T}_{1}(b\rho^{2b})}{\Omega} \int_{-\frac{M\rho^{b}}{\sqrt{2}}}^{\frac{M\rho^{b}}{\sqrt{2}}} \bigg\{ \bigg( \frac{ ye^{-y^{2}}}{\sqrt{\pi}\, \mathrm{erfc}(y)} \bigg)^{2} - \chi_{(0,+\infty)}(y) \bigg[ y^{4} + y^{2} - \frac{3}{4} \bigg] \bigg\}dy.
\end{align*}
Furthermore, using the asymptotics of $\mathrm{erfc}(y)$ as $y \to \pm \infty$, we infer that
\begin{align*}
& Q_{1}^{(n,M)} = \bigO(e^{-\frac{M^{2}\rho^{2b}}{2}}), \qquad Q_{2}^{(n,M)} = \bigO(M^{-1}), \\
& Q_{3}^{(M)} = \frac{\sqrt{2}\, b \rho^{b}}{\Omega} \big( \rho^{2b} \mathsf{T}_{2}(b\rho^{2b})-5\mathsf{T}_{1}(b\rho^{2b}) \big) \int_{-\infty}^{\infty} \bigg\{ \frac{ye^{-y^{2}}}{\sqrt{\pi}\, \mathrm{erfc}(y)} - \chi_{(0,+\infty)}(y) \bigg[ y^{2}+\frac{1}{2} \bigg] \bigg\}dy + \bigO(M^{-1}), \\
& Q_{4}^{(M)} = \frac{10\sqrt{2} \, b \rho^{b}}{3\Omega} \mathsf{T}_{1}(b\rho^{2b}) \int_{-\infty}^{\infty} \bigg\{ \frac{y^{3}e^{-y^{2}}}{\sqrt{\pi}\, \mathrm{erfc}(y)} - \chi_{(0,+\infty)}(y) \bigg[ y^{4} + \frac{y^{2}}{2} - \frac{1}{2} \bigg] \bigg\}dy +\bigO(M^{-1}), \\
& Q_{5}^{(M)} = \sqrt{2} \, b \rho^{b} \frac{\mathsf{T}_{1}(b\rho^{2b})}{\Omega}\bigg( \frac{2}{3} - \rho^{2b} \frac{\mathsf{T}_{1}(b\rho^{2b})}{\Omega} \bigg) \int_{-\infty}^{\infty} \bigg\{ \bigg( \frac{ e^{-y^{2}}}{\sqrt{\pi}\, \mathrm{erfc}(y)} \bigg)^{2} - \chi_{(0,+\infty)}(y) \bigg[ y^{2} + 1 \bigg] \bigg\}dy +\bigO(M^{-1}), \\
& Q_{6}^{(M)} = - \frac{10\sqrt{2}\, b  \rho^{b}}{3}\frac{\mathsf{T}_{1}(b\rho^{2b})}{\Omega} \int_{-\infty}^{\infty} \bigg\{ \bigg( \frac{ ye^{-y^{2}}}{\sqrt{\pi}\, \mathrm{erfc}(y)} \bigg)^{2} - \chi_{(0,+\infty)}(y) \bigg[ y^{4} + y^{2} - \frac{3}{4} \bigg] \bigg\}dy +\bigO(M^{-1}),
\end{align*}
as $n \to + \infty$. Substituting the above asymptotics in \eqref{lol2} yields
\begin{align}
& C_{4}^{(n, M)} = \widehat{C}_{4} + \bigO(M^{-1}),
\end{align}
and the claim follows.
\end{proof}

Recall that $\mathcal{I}_{1}, \mathcal{I}_{2}, \mathcal{I}_{3}, \mathcal{I}_{4}$ are defined in \eqref{def of I1}--\eqref{def of I5}, and that $\mathcal{I}$ is defined in \eqref{def of I}.
\begin{lemma}\label{lemma:some simplification of integrals}
The following relations hold:
\begin{align}\label{lol3}
& \mathcal{I}_{1} = \frac{\ln (2\sqrt{\pi})}{2}, & &  \mathcal{I}_{3} = \mathcal{I}, & & \mathcal{I}_{4} = \mathcal{I}_{2}-\mathcal{I}.
\end{align}
In particular, $\widehat{C}_{4}=C_{4}$, where $C_{4}$ is as in the statement of Theorem \ref{thm:main thm hard}.
\end{lemma}
\begin{proof}
The first identity in \eqref{lol3} follows from a direct calculation using the primitive
\begin{align*}
\int \frac{e^{-y^{2}}}{\sqrt{\pi} \, \mathrm{erfc}(y)} dy = -\frac{1}{2}\ln \big(\mathrm{erfc}(y)\big) + \mbox{const}.
\end{align*}
Integration by parts gives
\begin{align*}
& \int \bigg( \frac{e^{-y^{2}}}{\sqrt{\pi} \, \mathrm{erfc}(y)} \bigg)^{2} dy = \frac{e^{-y^{2}}}{2\sqrt{\pi} \, \mathrm{erfc}(y)} + \int \frac{y \, e^{-y^{2}}}{\sqrt{\pi} \, \mathrm{erfc}(y)} dy + \mbox{const}, \\
& \int \bigg( \frac{y \, e^{-y^{2}}}{\sqrt{\pi} \, \mathrm{erfc}(y)} \bigg)^{2} dy = \frac{y^{2}e^{-y^{2}}}{2\sqrt{\pi} \, \mathrm{erfc}(y)} + \int \frac{ (y^{3}-y)\, e^{-y^{2}}}{\sqrt{\pi} \, \mathrm{erfc}(y)} dy + \mbox{const}.
\end{align*}
Hence, for any $N>0$,
\begin{multline*}
\int_{-N}^{N} \bigg\{ \bigg( \frac{e^{-y^{2}}}{\sqrt{\pi} \, \mathrm{erfc}(y)} \bigg)^{2} - \chi_{(0,+\infty)}(y) \bigg[ y^{2}+1 \bigg] \bigg\} dy = \bigg( \frac{e^{-N^{2}}}{2\sqrt{\pi}\mathrm{erfc}(N)} - \frac{N}{2} \bigg) - \frac{e^{-N^{2}}}{2\sqrt{\pi}\mathrm{erfc}(-N)} \\
+ \int_{-N}^{N} \bigg\{ \frac{y \, e^{-y^{2}}}{\sqrt{\pi}\, \mathrm{erfc}(y)} - \chi_{(0,+\infty)}(y) \bigg[ y^{2}+\frac{1}{2} \bigg] \bigg\} dy,
\end{multline*}
and
\begin{multline*}
\int_{-N}^{N} \bigg\{ \bigg( \frac{y \, e^{-y^{2}}}{\sqrt{\pi} \, \mathrm{erfc}(y)} \bigg)^{2} - \chi_{(0,+\infty)}(y) \bigg[ y^{4} + y^{2} - \frac{3}{4} \bigg] \bigg\} dy = \bigg( \frac{N^{2}e^{-N^{2}}}{2\sqrt{\pi}\mathrm{erfc}(N)} - \frac{N^{3}}{2} - \frac{N}{4} \bigg) - \frac{N^{2}e^{-N^{2}}}{2\sqrt{\pi}\mathrm{erfc}(-N)} \\
+ \int_{-N}^{N} \bigg\{ \frac{y^{3} \, e^{-y^{2}}}{\sqrt{\pi} \, \mathrm{erfc}(y)} - \chi_{(0,+\infty)}(y) \bigg[ y^{4} + \frac{y^{2}}{2} - \frac{1}{2} \bigg] \bigg\} dy - \int_{-N}^{N} \bigg\{ \frac{y \, e^{-y^{2}}}{\sqrt{\pi} \, \mathrm{erfc}(y)} - \chi_{(0,+\infty)}(y) \bigg[ y^{2} + \frac{1}{2} \bigg] \bigg\} dy.
\end{multline*}
The second and third identities in \eqref{lol3} are obtained by letting $N \to + \infty$ in the above two formulas. We then find $\widehat{C}_{4} = C_{4}$ after a direct computation.
\end{proof}

\begin{proof}[End of the proof of Theorem \ref{thm:main thm hard}]
Let $M' > 0$ be sufficiently large such that Lemmas \ref{lemma: S2km1 hard} and \ref{lemma: asymp of S2k final hard} hold. Using \eqref{log Dn as a sum of sums hard} and Lemmas \ref{lemma: S0 hard}, \ref{lemma: S2km1 hard}, \ref{lemma:S3 asymp hard} and \ref{lemma: asymp of S2k final hard}, we conclude that for any $x_{1},\dots,x_{m} \in \mathbb{R}$, there exists $\delta > 0$ such that
\begin{align*}
& \ln \mathcal{E}_{n} = S_{0}+S_{1}+S_{2}+S_{3} \\
& = M' \ln \Omega + (j_{-}-M'-1) \ln \Omega - j_{-} \ln  \Omega + C_{1}^{(\epsilon)}n + n \int_{\frac{b \rho^{2b}}{1-\epsilon}}^{1} f_{1}(x)dx + C_{2} \ln n + C_{3}^{(n,\epsilon)} + \frac{C_{4}}{\sqrt{n}} \\
&  + \int_{\frac{b\rho^{2b}}{1-\epsilon}}^{1} f(x)dx + (\alpha+\theta_{+}^{(n,\epsilon)}-\tfrac{1}{2})f_{1}(\tfrac{b\rho^{2b}}{1-\epsilon})+\tfrac{1}{2}f_{1}(1) + \bigO\bigg( \frac{\sqrt{n}}{M^{11}} + \frac{1}{M^{6}} + \frac{1}{\sqrt{n} M} + \frac{M^{4}}{n} + \frac{M^{14}}{n^{2}} \bigg),
\end{align*}
as $n \to +\infty$ uniformly for $u_{1} \in D_\delta(x_1),\dots,u_{m} \in D_\delta(x_m)$. Since $M=n^{1/10}$, the above error term is $\bigO(n^{-3/5})$. Furthermore, using Lemma \ref{lemma:some simplification of integrals}, a computation shows that
\begin{align*}
& C_{1}^{(\epsilon)} + \int_{\frac{b \rho^{2b}}{1-\epsilon}}^{1} f_{1}(x)dx = C_{1}, \\
& -\ln \Omega + C_{3}^{(n,\epsilon)} + \int_{\frac{b\rho^{2b}}{1-\epsilon}}^{1} f(x)dx + (\alpha+\theta_{+}^{(n,\epsilon)}-\tfrac{1}{2})f_{1}(\tfrac{b\rho^{2b}}{1-\epsilon}))+\tfrac{1}{2}f_{1}(1) = C_{3},
\end{align*}
where $C_{1}$ and $C_{3}$ are as in the statement of Theorem \ref{thm:main thm hard}. This concludes the proof of Theorem \ref{thm:main thm hard}.
\end{proof}

\section{Proof of Theorem \ref{thm:main thm semi-hard}}\label{section:proof}
 As in the proof of Theorem \ref{thm:main thm hard}, our starting point is formula \eqref{log Dn as a sum of sums hard}, where $M'>0$ is an integer independent of $n$, $j_\pm$ are defined in (\ref{jplusjminusdef}), and $\epsilon > 0$ is such that (\ref{brho2b1epsilon}) holds. The variables $a_{j}, \lambda_j, \lambda_{j,k}, \eta_j, \eta_{j,k}$ are given by (\ref{def etajl hard}), where $r_k$ is now defined by (\ref{def of rell semi-hard}) (in contrast to Section \ref{section:proof edge} where $r_k$ was given by (\ref{def of rell hard})). The following two lemmas are analogous to Lemmas \ref{lemma: S0 hard} and \ref{lemma: S2km1 hard} and are proved in the same way.

\begin{lemma}\label{lemma: S0 semi-hard}
For any $x_{1},\dots,x_{m} \in \mathbb{R}$, there exists $\delta > 0$ such that
\begin{align}\label{asymp of S0 semi-hard}
S_{0} = M' \ln \Omega + \bigO(e^{-cn}), \qquad \mbox{as } n \to + \infty,
\end{align}
uniformly for $u_{1} \in D_\delta(x_1),\dots,u_{m} \in D_\delta(x_m)$.
\end{lemma}

\begin{lemma}\label{lemma: S2km1 semi-hard}
The constant $M'$ can be chosen sufficiently large such that the following holds. For any $x_{1},\dots,x_{m} \in \mathbb{R}$, there exists $\delta > 0$ such that
\begin{align*}
& S_{1} = (j_{-}-M'-1) \ln \Omega + \bigO(e^{-cn}),
\end{align*}
as $n \to +\infty$ uniformly for $u_{1} \in D_\delta(x_1),\dots,u_{m} \in D_\delta(x_m)$.
\end{lemma}

\begin{lemma}\label{lemma:S3 asymp semi-hard}
For any $x_{1},\dots,x_{m} \in \mathbb{R}$, there exists $\delta > 0$ such that
\begin{align*}
& S_{3} = \bigO(e^{-c\sqrt{n}}),
\end{align*}
as $n \to +\infty$ uniformly for $u_{1} \in D_\delta(x_1),\dots,u_{m} \in D_\delta(x_m)$.
\end{lemma}
\begin{proof}
For $j \geq j_{+}+1$ and $k \in \{1,\dots,m\}$, $1-\lambda_{j}$ and $1-\lambda_{j,k}$ are positive and remain bounded away from $0$. Hence, using Lemma \ref{lemma: asymp of gamma for lambda bounded away from 1} $(ii)$, we obtain
\begin{align*}
S_{3} & = \sum_{j=j_{+}+1}^{n} \ln \bigg\{ 1+\sum_{\ell=1}^{m} \omega_{\ell}\frac{e^{-\frac{a_{j}\eta_{j,\ell}^{2}}{2}}( \frac{-1}{\lambda_{j,\ell}-1}\frac{1}{\sqrt{a_{j}}}
+ \bigO(n^{-\frac{3}{2}}) )}{e^{-\frac{a_{j}\eta_{j}^{2}}{2}}( \frac{-1}{\lambda_{j}-1}\frac{1}{\sqrt{a_{j}}}+ \bigO(n^{-\frac{3}{2}}) )} \bigg\}
= \sum_{j=j_{+}+1}^{n} \ln \bigg\{ 1+\sum_{\ell=1}^{m} \omega_{\ell} \bigO(e^{\frac{a_j}{2}(\eta_j^2 - \eta_{j,\ell}^2)}) \bigg\},
\end{align*}
where the $\bigO$-terms are uniform for $j \in \{j_{+}+1,\dots,n\}$ and independent of $u_{1},\dots,u_{m}$. Using that $r_k$ is given by (\ref{def of rell semi-hard}), we find, as $n \to +\infty$,
\begin{align}\label{ajetasemihard}
\frac{a_j}{2}(\eta_j^2 - \eta_{j,\ell}^2)
= -\frac{\sqrt{2} \mathfrak{s}_{\ell}(j/n-b\rho^{2b})\sqrt{n}}{b\rho^{b}}
+ \bigO(1)
\end{align}
and hence
\begin{align*}
S_3 = \sum_{j=j_{+}+1}^{n} \ln \bigg( 1 + \sum_{\ell=1}^{m} \omega_{\ell} \bigO(e^{-\frac{\sqrt{2} \mathfrak{s}_{\ell}(j/n-b\rho^{2b})\sqrt{n}}{b\rho^{b}}}) \bigg),
\end{align*}
where the $\bigO$-terms are uniform for $j \in \{j_{+}+1,\dots,n\}$ and independent of $u_{1},\dots,u_{m}$.
 Since $\mathfrak{s}_{\ell}>0$ for all $\ell \in \{1,\dots,m\}$ and since $j/n-b\rho^{2b}$ is positive and bounded away from $0$ as $n \to + \infty$ with $j \in \{j_{+}+1,\dots,n\}$, the claim follows.
\end{proof}

We now focus on $S_{2}$. As in Section \ref{section:proof edge}, we decompose $S_{2}$ into three pieces, $S_{2}=S_{2}^{(1)}+S_{2}^{(2)}+S_{2}^{(3)}$, where the $S_{2}^{(v)}$ are given by (\ref{asymp prelim of S2kpvp hard}). However, in contrast to Section \ref{section:proof edge}, we let the intervals $I_v$ be given by (\ref{I1I2I3def}) with $M:=M' \ln n$. Using this $M$, we define $g_{\pm}$ and
$\theta_{-}^{(n,M)},\theta_{+}^{(n,M)} \in [0,1)$ as in Section \ref{section:proof edge}. The following lemma is analogous to Lemma \ref{lemma:S2kp3p hard}
and is proved in the same way.

\begin{lemma}\label{lemma:S2kp3p semi-hard}
The constant $M'$ can be chosen sufficiently large such that the following holds. For any $x_{1},\dots,x_{m} \in \mathbb{R}$, there exists $\delta > 0$ such that
\begin{align*}
S_{2}^{(3)} = & \; \Big( b\rho^{2b}n - j_{-} - bM\rho^{2b}\sqrt{n} + bM^{2}\rho^{2b} -\alpha+\theta_{-}^{(n,M)} - bM^{3}\rho^{2b}n^{-\frac{1}{2}} \Big) \ln  \Omega + \bigO(M^{4}n^{-1}),
\end{align*}
as $n \to +\infty$ uniformly for $u_{1} \in D_\delta(x_1),\dots,u_{m} \in D_\delta(x_m)$.
\end{lemma}

In the case of the hard edge, we found that $S_2^{(1)}$ made important contributions to the asymptotic formula for large $n$ (see Lemma \ref{lemma:S2p1p hard}). However, in the semi-hard regime, $S_{2}^{(1)}$ is small as the next lemma shows.

\begin{lemma}\label{lemma:S2p1p semi-hard}
$M'$ can be chosen sufficiently large such that the following holds. For any $x_{1},\dots,x_{m} \in \mathbb{R}$, there exists $\delta > 0$ such that
\begin{align*}
& S_{2}^{(1)} = \bigO\big( n^{-100} \big),
\end{align*}
as $n \to +\infty$ uniformly for $u_{1} \in D_\delta(x_1),\dots,u_{m} \in D_\delta(x_m)$.
\end{lemma}
\begin{proof}
Since $\lambda_j \in [1-\epsilon, 1 - \frac{M}{\sqrt{n}})$ for $g_+ + 1 \leq j \leq j_+$ and $\lambda_{j,\ell} = \lambda_j(1-\frac{\sqrt{2} \mathfrak{s}_{\ell}}{\rho^b \sqrt{n}})$, we have $\eta_j, \eta_{j,\ell} \leq -c M/\sqrt{n}$ for some $c > 0$, and so Lemma \ref{lemma: asymp of gamma for lambda bounded away from 1} $(ii)$ yields
\begin{align*}
S_{2}^{(1)} & = \sum_{j= g_{+}+1}^{j_{+}} \ln \bigg( 1+ \frac{\sum_{\ell=1}^{m} \omega_{\ell} \gamma\big(a_j, a_j \lambda_{j, \ell}\big) }{ \gamma\big(a_j, a_j \lambda_j \big)} \bigg)
	\\
&= \sum_{j= g_{+}+1}^{j_{+}} \ln \bigg( 1+\sum_{\ell=1}^{m} \omega_{\ell}\frac{e^{-\frac{a_{j}\eta_{j,\ell}^{2}}{2}}( \frac{-1}{\lambda_{j,\ell}-1}\frac{1}{\sqrt{a_{j}}}
+ \bigO((a_j M^2/n)^{-\frac{3}{2}}) )}{e^{-\frac{a_{j}\eta_{j}^{2}}{2}}( \frac{-1}{\lambda_{j}-1}\frac{1}{\sqrt{a_{j}}}+ \bigO((a_j M^2/n)^{-\frac{3}{2}}) )} \bigg)
	\\
& = \sum_{j= g_{+}+1}^{j_{+}} \ln \bigg( 1+\sum_{\ell=1}^{m} \omega_{\ell} \bigO(e^{\frac{a_j}{2}(\eta_j^2 - \eta_{j,\ell}^2)}) \bigg)
 = \sum_{j= g_{+}+1}^{j_{+}} \ln \bigg( 1+\sum_{\ell=1}^{m} \omega_{\ell} \bigO(e^{-\frac{\sqrt{2} \mathfrak{s}_{\ell}(j/n-b\rho^{2b})\sqrt{n}}{b\rho^{b}}}) \bigg),
\end{align*}
where we have used \eqref{ajetasemihard} in the last step. Since $M=M'\ln n$ and $\mathfrak{s}_{\ell}>0$, the claim follows from the fact that $j/n-b\rho^{2b} \geq b\rho^{2b}\frac{M+\bigO(1)}{\sqrt{n}}$ as $n \to + \infty$ for $j \in \{g_{+}+1,\dots,j_{+}\}$.
\end{proof}

For $k \in \{1,\dots,m\}$ and $j \in \{j: \lambda_{j} \in I_{2}\}=\{g_{-},\dots,g_{+}\}$, we define $M_{j,k} := \sqrt{n}(\lambda_{j,k}-1)$ and $M_{j} := \sqrt{n}(\lambda_{j}-1)$.

\begin{lemma}\label{lemma:S2kp2p semi-hard}
For any $x_{1},\dots,x_{m} \in \mathbb{R}$, there exists $\delta > 0$ such that
\begin{align*}
S_{2}^{(2)} = &\; E_{2}^{(M)} \sqrt{n} + E_{3}^{(M)} + \frac{E_{4}^{(M)}}{\sqrt{n}} + \bigO\bigg( \frac{M^{4}}{n} \bigg), \\
E_{2}^{(M)} = &\;  \sqrt{2}b\rho^{b} \int_{-\frac{M\rho^{b}}{\sqrt{2}}}^{\frac{M\rho^{b}}{\sqrt{2}}}h_{0}(y)dy, \\
 E_{3}^{(M)} = &\; b \int_{-\frac{M\rho^{b}}{\sqrt{2}}}^{\frac{M\rho^{b}}{\sqrt{2}}} \big( 4yh_{0}(y)+\sqrt{2} h_{1}(y) \big)dy+ \bigg( \frac{1}{2}-\theta_{-}^{(n,M)}\bigg) h_{0}\Big( -\frac{M\rho^{b}}{\sqrt{2}} \Big) + \bigg( \frac{1}{2}-\theta_{+}^{(n,M)} \bigg)h_{0}\Big( \frac{M\rho^{b}}{\sqrt{2}} \Big), \\
 E_{4}^{(M)} = &\; b\rho^{-b} \int_{-\frac{M\rho^{b}}{\sqrt{2}}}^{\frac{M\rho^{b}}{\sqrt{2}}} \big( 6\sqrt{2} y^{2}h_{0}(y) + 4yh_{1}(y)+\sqrt{2}h_{2}(y) \big)dy - \bigg( \frac{1}{12} + \frac{\theta_{-}^{(n,M)}(\theta_{-}^{(n,M)}-1)}{2} \bigg) \frac{h_{0}'(-\frac{M\rho^{b}}{\sqrt{2}})}{\sqrt{2} b \rho^{b}} \\
& + \bigg( \frac{1}{12} + \frac{\theta_{+}^{(n,M)}(\theta_{+}^{(n,M)}-1)}{2} \bigg) \frac{h_{0}'(\frac{M\rho^{b}}{\sqrt{2}})}{\sqrt{2} b \rho^{b}} + \bigg( \frac{1}{2}-\theta_{-}^{(n,M)} \bigg)\rho^{-b}h_{1}\Big( -\frac{M\rho^{b}}{\sqrt{2}} \Big)
	\\
& + \bigg( \frac{1}{2}-\theta_{+}^{(n,M)} \bigg)\rho^{-b}h_{1}\Big( \frac{M\rho^{b}}{\sqrt{2}} \Big)
\end{align*}
as $n \to +\infty$ uniformly for $u_{1} \in D_\delta(x_1),\dots,u_{m} \in D_\delta(x_m)$, where $h_{0}$, $h_{1}$, $h_{2}$ are as in the statement of Theorem \ref{thm:main thm semi-hard}.
\end{lemma}
\begin{proof}
Using \eqref{asymp prelim of S2kpvp hard} and Lemma \ref{lemma: uniform}, we obtain
\begin{align}\label{lol1 semi-hard}
& S_{2}^{(2)} = \sum_{j:\lambda_{j}\in I_{2}} \ln \bigg( 1+\sum_{\ell=1}^{m} \omega_{\ell} \frac{ \frac{1}{2}\mathrm{erfc}\Big(-\eta_{j,\ell} \sqrt{a_{j}/2}\Big) - R_{a_{j}}(\eta_{j,\ell}) }{ \frac{1}{2}\mathrm{erfc}\big(-\eta_{j} \sqrt{a_{j}/2}\big) - R_{a_{j}}(\eta_{j}) } \bigg).
\end{align}
For $j \in \{j:\lambda_{j}\in I_{2}\}$, we have $1-\frac{M}{\sqrt{n}} \leq \lambda_{j} = \frac{bn\rho^{2b}}{j+\alpha} \leq 1+\frac{M}{\sqrt{n}}$, $-M \leq M_{j} \leq M$, and
\begin{align*}
M_{j,k} = M_{j} - \frac{\sqrt{2}\,\mathfrak{s}_{k}}{\rho^{b}} - \frac{\sqrt{2}\, \mathfrak{s}_{k} M_{j}}{\rho^{b}\sqrt{n}}, \qquad k=1,\dots,m.
\end{align*}
Furthermore, as $n \to + \infty$ we have
\begin{align}
\eta_{j,\ell} & = \frac{M_{j}-\sqrt{2} \, \mathfrak{s}_{\ell} \rho^{-b}}{\sqrt{n}} - \frac{M_{j}^{2} + \sqrt{2}M_{j}\mathfrak{s}_{\ell} \rho^{-b} + 2\mathfrak{s}_{\ell}^{2}\rho^{-2b}}{3n} \nonumber \\
& + \frac{7M_{j}^{3} + 3\sqrt{2} M_{j}^{2} \mathfrak{s}_{\ell} \rho^{-b} -6M_{j} \mathfrak{s}_{\ell}^{2}\rho^{-2b} - 14\sqrt{2}\mathfrak{s}_{\ell}^{3}\rho^{-3b}}{36 n^{3/2}} + \bigO\bigg(\frac{1+M_{j}^{4}}{n^{2}}\bigg),
	\label{asymp etaj and etajsqrtajover2 1 semi-hard} \\
-\eta_{j,\ell} \sqrt{a_{j}/2} & = - \frac{M_{j}\rho^{b}}{\sqrt{2}} + \mathfrak{s}_{\ell} + \frac{5\sqrt{2}\, M_{j}^{2}\rho^{b}-2M_{j}\mathfrak{s}_{\ell}+4\sqrt{2} \mathfrak{s}_{\ell}^{2}\rho^{-b}}{12 \sqrt{n}} \nonumber \\
& - \frac{53\sqrt{2}M_{j}^{3}\rho^{b}-18M_{j}^{2}\mathfrak{s}_{\ell}+12\sqrt{2}M_{j}\mathfrak{s}_{\ell}^{2}\rho^{-b} - 56 \mathfrak{s}_{\ell}^{3}\rho^{-2b}}{144 n} + \bigO\bigg( \frac{1+M_{j}^{4}}{n^{3/2}} \bigg) \label{asymp etaj and etajsqrtajover2 2 semi-hard}
\end{align}
uniformly for $j\in \{j:\lambda_{j}\in I_{2}\}$. Hence, after a long computation using \eqref{asymp of Ra}, we obtain
\begin{align}
1+\sum_{\ell=1}^{m} \omega_{\ell} \frac{ \frac{1}{2}\mathrm{erfc}\Big(-\eta_{j,\ell} \sqrt{a_{j}/2}\Big) - R_{a_{j}}(\eta_{j,\ell}) }{ \frac{1}{2}\mathrm{erfc}\big(-\eta_{j} \sqrt{a_{j}/2}\big) - R_{a_{j}}(\eta_{j}) }  & = g_{0}(-\tfrac{\rho^{b}M_{j}}{\sqrt{2}}) + \frac{g_{1}(-\frac{\rho^{b}M_{j}}{\sqrt{2}})}{\rho^{b}\sqrt{n}} + \frac{g_{2}(-\tfrac{\rho^{b}M_{j}}{\sqrt{2}})}{\rho^{2b}n} + \bigO\Big(\frac{e^{-c|M_{j}|}}{n^{3/2}}\Big), \label{asymp of S2kp2p in proof semi-hard}
\end{align}
as $n \to + \infty$, where $g_{0}$, $g_{1}$ and $g_{2}$ are as in the statement of Theorem \ref{thm:main thm semi-hard}. For the above error term, we have used that $s_{\ell}>0$, $\ell \in \{1,\dots,m\}$. Thus
\begin{align*}
S_{2}^{(2)} & = \sum_{j=g_{-}}^{g_{+}} \ln \bigg( 1+\sum_{\ell=1}^{m} \omega_{\ell} \frac{ \frac{1}{2}\mathrm{erfc}\Big(-\eta_{j,\ell} \sqrt{a_{j}/2}\Big) - R_{a_{j}}(\eta_{j,\ell}) }{ \frac{1}{2}\mathrm{erfc}\big(-\eta_{j} \sqrt{a_{j}/2}\big) - R_{a_{j}}(\eta_{j}) } \bigg) \\
&  = \sum_{j=g_{-}}^{g_{+}} \bigg\{ h_{0}(-\tfrac{\rho^{b}M_{j}}{\sqrt{2}}) + \frac{h_{1}(-\frac{\rho^{b}M_{j}}{\sqrt{2}})}{\rho^{b}\sqrt{n}} + \frac{h_{2}(-\tfrac{\rho^{b}M_{j}}{\sqrt{2}})}{\rho^{2b}n} + \bigO\Big(\frac{e^{-c|M_{j}|}}{n^{3/2}}\Big) \bigg\}, \qquad \mbox{as } n \to + \infty.
\end{align*}
After a computation using Lemma \ref{lemma:Riemann sum}, a change of variables and the fact that $g_{1}(y),g_{2}(y)=\bigO(e^{-c|y|})$ as $y \to \pm \infty$, we find the claim.
\end{proof}
\begin{lemma}\label{lemma: asymp of S2k final semi-hard}
The constant $M'$ can be chosen sufficiently large such that the following holds. For any $x_{1},\dots,x_{m} \in \mathbb{R}$, there exists $\delta > 0$ such that
\begin{align*}
& S_{2} =  - j_{-} \ln  \Omega + C_{1}n + C_{2} \sqrt{n} + C_{3} + \ln \Omega + \frac{C_{4}}{\sqrt{n}}  + \bigO \bigg( \frac{M^{4}}{n} \bigg),
\end{align*}
as $n \to +\infty$ uniformly for $u_{1} \in D_\delta(x_1),\dots,u_{m} \in D_\delta(x_m)$, where $C_{1},\dots,C_{4}$ are as in the statement of Theorem \ref{thm:main thm semi-hard}.
\end{lemma}
\begin{proof}
By combining Lemmas \ref{lemma:S2kp3p semi-hard}, \ref{lemma:S2p1p semi-hard} and \ref{lemma:S2kp2p semi-hard}, we obtain
\begin{align*}
& S_{2} = - j_{-} \ln  \Omega + C_{1} n + C_{2}^{(M)}\sqrt{n} + C_{3}^{(M)} + \frac{C_{4}^{(M)}}{\sqrt{n}}  + \bigO\bigg( \frac{M^{4}}{n} \bigg),
\end{align*}
as $n \to +\infty$ uniformly for $u_{1} \in D_\delta(x_1),\dots,u_{m} \in D_\delta(x_m)$, where $C_{1}$ is as in the statement, and
\begin{align*}
& C_{2}^{(M)} = - bM \rho^{2b} \ln \Omega +E_{2}^{(M)}, \\
& C_{3}^{(M)} = \big( bM^{2}\rho^{2b} - \alpha + \theta_{-}^{(n,M)} \big) \ln \Omega+E_{3}^{(M)}, \\
& C_{4}^{(M)} = -bM^{3}\rho^{2b} \ln \Omega + E_{4}^{(M)}.
\end{align*}
A direct analysis shows that $M'$ can be chosen sufficiently large such that
\begin{align*}
& C_{2}^{(M)} = C_{2} + \bigO(n^{-100}), \qquad C_{3}^{(M)} = C_{3} + \ln \Omega + \bigO(n^{-100}), \qquad C_{4}^{(M)} = C_{4} + \bigO(n^{-100}),
\end{align*}
and the claim follows.
\end{proof}
\begin{proof}[End of the proof of Theorem \ref{thm:main thm semi-hard}]
Let $M' > 0$ be sufficiently large such that Lemmas \ref{lemma: S2km1 semi-hard} and \ref{lemma: asymp of S2k final semi-hard} hold. Using \eqref{log Dn as a sum of sums hard} and Lemmas \ref{lemma: S0 semi-hard}, \ref{lemma: S2km1 semi-hard}, \ref{lemma:S3 asymp semi-hard} and \ref{lemma: asymp of S2k final semi-hard}, we conclude that for any $x_{1},\dots,x_{m} \in \mathbb{R}$, there exists $\delta > 0$ such that
\begin{align*}
& \ln \mathcal{E}_{n} = S_{0}+S_{1}+S_{2}+S_{3} \\
& = M' \ln \Omega + (j_{-}-M'-1) \ln \Omega - j_{-} \ln  \Omega + C_{1}n + C_{2} \sqrt{n} + C_{3} + \ln \Omega + \frac{C_{4}}{\sqrt{n}}  + \bigO(M^{4}n^{-1}) \\
& = C_{1} n + C_{2} \sqrt{n} + C_{3} + \frac{C_{4}}{\sqrt{n}}  + \bigO(M^{4}n^{-1}),
\end{align*}
as $n \to +\infty$ uniformly for $u_{1} \in D_\delta(x_1),\dots,u_{m} \in D_\delta(x_m)$. This concludes the proof of Theorem \ref{thm:main thm semi-hard}.
\end{proof}

\appendix
\section{Uniform asymptotics of the incomplete gamma function}\label{section:uniform asymp gamma}
\begin{lemma}\label{lemma:various regime of gamma}(From \cite[formula 8.11.2]{NIST}).
Let $a>0$ be fixed. As $z \to +\infty$,
\begin{align*}
\gamma(a,z) = \Gamma(a) + \bigO(e^{-\frac{z}{2}}).
\end{align*}
\end{lemma}

\begin{lemma}\label{lemma: uniform}(From \cite[Section 11.2.4]{Temme}).
We have
\begin{align*}
& \frac{\gamma(a,z)}{\Gamma(a)} = \frac{1}{2}\mathrm{erfc}(-\eta \sqrt{a/2}) - R_{a}(\eta), \qquad R_{a}(\eta) = \frac{e^{-\frac{1}{2}a \eta^{2}}}{2\pi i}\int_{-\infty}^{\infty}e^{-\frac{1}{2}a u^{2}}g(u)du,
\end{align*}
where $\mathrm{erfc}$ is defined in \eqref{def of erfc},
\begin{align}\label{lol8}
& \lambda = \frac{z}{a}, \qquad \eta = (\lambda-1)\sqrt{\frac{2 (\lambda-1-\ln \lambda)}{(\lambda-1)^{2}}}, \qquad g(u) := \frac{dt}{du} \frac{1}{\lambda -t} + \frac{1}{u +i\eta},
\end{align}
with $t$ and $u$ being related by the bijection $t \mapsto u$ from $\mathcal{L} := \{\frac{\theta}{\sin \theta} e^{i\theta} : - \pi < \theta < \pi\}$ to $\R$ given by
$$u = -i(t-1) \sqrt{\frac{2(t-1-\ln t)}{(t-1)^2}}, \qquad  t \in \mathcal{L},$$
and the principal branch is used for the roots. Furthermore, as $a \to + \infty$, uniformly for $z \in [0,\infty)$,
\begin{align}\label{asymp of Ra}
& R_{a}(\eta) \sim \frac{e^{-\frac{1}{2}a \eta^{2}}}{\sqrt{2\pi a}}\sum_{j=0}^{\infty} \frac{c_{j}(\eta)}{a^{j}},
\end{align}
where all coefficients $c_{j}(\eta)$ are bounded functions of $\eta \in \mathbb{R}$ (i.e. bounded for $\lambda \in (0,+\infty)$). The first two coefficients are given by (see \cite[p. 312]{Temme})
\begin{align*}
c_{0}(\eta) = \frac{1}{\lambda-1}-\frac{1}{\eta}, \qquad c_{1}(\eta) = \frac{1}{\eta^{3}}-\frac{1}{(\lambda-1)^{3}}-\frac{1}{(\lambda-1)^{2}}-\frac{1}{12(\lambda-1)}.
\end{align*}
More generally, we have
\begin{align}\label{recursive def of the cj}
c_{j}(\eta) = \frac{1}{\eta} \frac{d}{d\eta}c_{j-1}(\eta) + \frac{\gamma_{j}}{\lambda-1}, \; j \geq 1,
\end{align}
where the $\gamma_{j}$ are the Stirling coefficients
\begin{align}\label{stirlinggammaj}
\gamma_{j} = \frac{(-1)^{j}}{2^{j} \, j!} \bigg[ \frac{d^{2j}}{dx^{2j}} \bigg( \frac{1}{2}\frac{x^{2}}{x-\ln(1+x)} \bigg)^{j+\frac{1}{2}} \bigg]_{x=0}.
\end{align}
In particular, the following hold:
\item[(i)] Let $z=\lambda a$ and let $\delta >0$ be fixed. As $a \to +\infty$, uniformly for $\lambda \geq 1+\delta$,
\begin{align*}
\gamma(a,z) = \Gamma(a)\big(1 + \bigO(e^{-\frac{a \eta^{2}}{2}})\big).
\end{align*}
\item[(ii)] Let $z=\lambda a$. As $a \to +\infty$, uniformly for $\lambda$ in compact subsets of $(0,1)$,
\begin{align*}
\gamma(a,z) = \Gamma(a)\bigO(e^{-\frac{a \eta^{2}}{2}}).
\end{align*}
\end{lemma}

The following lemma establishes a non-recursive formula for the coefficients $c_{j}$, which is new to our knowledge.

\begin{lemma}\label{lemma:coeff non-recursive}
For $j \geq 0$, the coefficients $c_j(\eta)$ in (\ref{asymp of Ra}) can be expressed as
\begin{align}\label{cjvarphij}
c_j(\eta) = \varphi_j(\lambda) - S(\varphi_j(\lambda)), \quad \text{where} \quad \varphi_j(\lambda) := \frac{(-1)^{j+1} (2j-1)!!}{\eta^{2j+1}}
\end{align}
and where $S(\varphi_j(\lambda))$ denotes the singular part of $\varphi_j(\lambda)$ at $\lambda = 1$, i.e., $S(\varphi_j(\lambda))$ is the sum of the singular terms in the Laurent expansion of $\varphi_j(\lambda)$ at $\lambda = 1$.
\end{lemma}
\begin{proof}
The formula (\ref{cjvarphij}) holds for $j = 0$. Suppose it holds for $j = k-1 \geq 0$.  Then (\ref{recursive def of the cj}) yields
$$c_k(\eta) = \frac{1}{\eta} \frac{d}{d\eta}\varphi_{k-1}(\lambda) - \frac{1}{\eta} \frac{d}{d\eta}S(\varphi_{k-1}(\lambda)) + \frac{\gamma_{k}}{\lambda-1}.$$
We have $\partial_\eta \varphi_{k-1}(\lambda) = \eta \varphi_k(\lambda)$. Hence, using also that $\partial_\eta$ commutes with $S$,
$$c_k(\eta) = \varphi_{k}(\lambda) - \frac{1}{\eta} S(\eta \varphi_{k}(\lambda)) + \frac{\gamma_{k}}{\lambda-1}.$$
On the other hand, $\varphi_k$ has a pole of order $2k+1$ at $\lambda = 1$, so in light of the identity $(2k)! = (2k-1)!! 2^k k!$ and (\ref{stirlinggammaj}), we obtain
$$\underset{\lambda = 1}{\res} \varphi_k(\lambda)
= \frac{1}{(2k)!} \lim_{\lambda \to 1} \frac{d^{2k}}{d\lambda^{2k}}((\lambda-1)^{2k+1}\varphi_k(\lambda))
= \frac{(-1)^{k+1}}{2^k k!} \lim_{\lambda \to 1} \frac{d^{2k}}{d\lambda^{2k}} \bigg(\frac{(\lambda-1)^2}{2(\lambda - 1 - \ln{\lambda})}\bigg)^{k + \frac{1}{2}}
= -\gamma_k.$$
It follows that (\ref{cjvarphij}) holds also for $j = k$, completing the proof.
\end{proof}

Note that $S(\varphi_j(\lambda))$ is a polynomial of order $2j+1$ in $(\lambda -1)^{-1}$ without constant term. The first $S(\varphi_j(\lambda))$ are given by
  \begin{align*}
& S(\varphi_0(\lambda)) = -\frac{1}{\lambda-1}, \qquad
S(\varphi_1(\lambda)) = \frac{1}{(\lambda-1)^{3}}+\frac{1}{(\lambda-1)^{2}}+\frac{1}{12(\lambda-1)},
	\\
& S(\varphi_2(\lambda)) = -\frac{3}{(\lambda -1)^5}-\frac{5}{(\lambda -1)^4} - \frac{25}{12 (\lambda
   -1)^3} - \frac{1}{12 (\lambda -1)^2} - \frac{1}{288 (\lambda -1)}.
\end{align*}

The following lemma follows from a result of Tricomi \cite{Tricomi}, see also \cite{AC2021}. However, in contrast to \cite{Tricomi, AC2021}, the coefficients appearing in Lemma \ref{lemma: asymp of gamma for lambda bounded away from 1} below are written in a non-recursive way. Here we give a short proof relying on Lemmas \ref{lemma: uniform} and \ref{lemma:coeff non-recursive}.
\begin{lemma}\label{lemma: asymp of gamma for lambda bounded away from 1}
Let $N \geq 0$ be an integer and let $\eta$ and $S(\varphi_j(\lambda))$ be as in \eqref{cjvarphij}.

\item[(i)] As $a \to +\infty$, uniformly for $\lambda \geq 1+ \frac{1}{\sqrt{a}}$,
\begin{align*}
\frac{\gamma(a,\lambda a)}{\Gamma(a)} & = 1 + \frac{e^{-\frac{a}{2}\eta^2}}{\sqrt{2\pi}} \bigg\{\sum_{j=0}^{N-1} \frac{S(\varphi_j(\lambda))}{a^{j+\frac{1}{2}}} + \bigO\bigg(\frac{1}{a^{N+\frac{1}{2}}}\bigg) +  \bigO\bigg(\frac{1}{(a \eta^2)^{N+\frac{1}{2}}}\bigg)\bigg\}.
\end{align*}

\item[(ii)] As $a \to +\infty$, uniformly for $\lambda \in [\epsilon, 1-\frac{1}{\sqrt{a}}]$ for any fixed $\epsilon > 0$,
\begin{align*}
\frac{\gamma(a,\lambda a)}{\Gamma(a)} = \frac{e^{-\frac{a}{2}\eta^2}}{\sqrt{2\pi}} \bigg\{\sum_{j=0}^{N-1} \frac{S(\varphi_j(\lambda))}{a^{j+\frac{1}{2}}} + \bigO\bigg(\frac{1}{a^{N+\frac{1}{2}}}\bigg) +  \bigO\bigg(\frac{1}{(a \eta^2)^{N+\frac{1}{2}}}\bigg)\bigg\}.
\end{align*}
\end{lemma}
\begin{proof}
$(i)$ The assumption $\lambda \geq 1+ \frac{1}{\sqrt{a}}$ implies that $-\eta \sqrt{a} \leq -c$ for some $c > 0$.
In view of the identity $\mathrm{erfc}(-x) = 2 - \mathrm{erfc}(x)$ and the expansion
\begin{align}\label{erfclargex}
\mathrm{erfc}(x) \sim \frac{e^{-x^2}}{\sqrt{\pi}} \sum_{j=0}^\infty \frac{(-1)^j (1/2)_j}{x^{2j+1}}, \qquad x \to +\infty,
\end{align}
where $(1/2)_j=\prod_{k=0}^{j-1}(\frac{1}{2}+k)$ is the rising factorial,
Lemma \ref{lemma: uniform} implies that, for any $N \geq 0$,
\begin{align*}
\frac{\gamma(a,\lambda a)}{\Gamma(a)}
= 1 - \frac{e^{-\frac{a}{2}\eta^2}}{2\sqrt{\pi}} \sum_{j=0}^{N-1} \frac{(-1)^j (1/2)_j}{(\eta \sqrt{a/2})^{2j+1}} +  \bigO\bigg(\frac{1}{(\eta \sqrt{a})^{2N+1}}\bigg)
- \frac{e^{-\frac{1}{2}a \eta^{2}}}{\sqrt{2\pi a}}\sum_{j=0}^{N-1} \frac{c_{j}(\eta)}{a^{j}}
+ \bigO\bigg(\frac{1}{a^{N+\frac{1}{2}}}\bigg)
	\\
= 1 - \frac{e^{-\frac{a}{2}\eta^2}}{\sqrt{2\pi}} \sum_{j=0}^{N-1} \frac{1}{(\sqrt{a})^{2j+1}}\bigg(\frac{(-1)^j (1/2)_j 2^j}{\eta^{2j+1}} + c_{j}(\eta)\bigg)
 +  \bigO\bigg(\frac{1}{(a \eta^2)^{N+\frac{1}{2}}}\bigg) + \bigO\bigg(\frac{1}{a^{N+\frac{1}{2}}}\bigg).
\end{align*}
Since $(1/2)_j 2^j = \frac{1}{2} \cdot \frac{3}{2} \cdot\frac{5}{2} \cdots \frac{2j-1}{2} 2^j = (2j-1)!!$,  the desired conclusion follows from (\ref{cjvarphij}).

$(ii)$ The assumption $\lambda \leq 1- \frac{1}{\sqrt{a}}$ implies that $-\eta \sqrt{a} \geq c$ for some $c > 0$. Using (\ref{erfclargex}) and Lemma \ref{lemma: uniform}, the desired conclusion now follows as in the proof of $(i)$.
\end{proof}

\paragraph{Acknowledgements.} CC acknowledges support from the Swedish Research Council, Grant No. 2021-04626. JL acknowledges support from the Swedish Research Council, Grant No. 2021-03877 and the Ruth and Nils-Erik Stenb\"ack Foundation. We thank S.-S. Byun for help with Figure \ref{fig: ML with hard wall}.


\footnotesize
\begin{thebibliography}{99}
\bibitem{AR Infinite Ginibre} K. Adhikari and N.K. Reddy, Hole probabilities for finite and infinite Ginibre ensembles, \textit{Int. Math. Res. Not. IMRN} \textbf{2017} (2017), no. 21, 6694--6730.

\bibitem{AkemannSungsoo} G. Akemann, S.-S. Byun and M. Ebke, Universality of the number variance in rotational invariant two-dimensional Coulomb gases, \textit{J. Stat. Phys.} \textbf{190} (2023), no. 1, Paper No. 9, 34 pp.

\bibitem{AIE gap 2014} G. Akemann, J. Ipsen and E. Strahov, Permanental processes from products of complex and quaternionic induced Ginibre ensembles, \textit{Random Matrices Theory Appl.} \textbf{3} (2014), no. 4, 1450014, 54 pp.


\bibitem{APS gap 2009}
G.~Akemann, M.~J. Phillips and L.~Shifrin,  Gap probabilities in non-{H}ermitian random matrix theory,  \textit{J. Math. Phys.} \textbf{50} (2009), no. 6, 063504, 32 pp.

\bibitem{AK hole 2013}
G.~Akemann and E.~Strahov, Hole probabilities and overcrowding estimates for products of complex Gaussian matrices, \textit{J. Stat. Phys.} \textbf{151} (2013), no. 6, 987--1003.

\bibitem{A2021} Y. Ameur, A localization theorem for the planar Coulomb gas in an external field, \textit{Electron. J. Probab.} \textbf{26} (2021), Paper No. 46, 21 pp.

\bibitem{AC2021} Y. Ameur and J. Cronvall, Szeg\"{o} type asymptotics for the reproducing kernel in spaces of full-plane weighted polynomials, \textit{Comm. Math. Phys.} \textbf{398} (2023), 1291--1348.

\bibitem{AHM2} Y. Ameur, H. Hedenmalm and N. Makarov, Fluctuations of eigenvalues of random normal matrices, \textit{Duke J. Math.} \textbf{159} (2011), 1533--1584.

\bibitem{AM} Y. Ameur, H. Hedenmalm, N. Makarov, Ward identities and random normal matrices, \textit{Ann. Probab.} \textbf{43} (2015), 1157--1201.

\bibitem{AK2013} Y. Ameur and N.-G. Kang, On a problem for Ward's equation with a Mittag-Leffler potential, \textit{Bull. Sci. Math.} \textbf{137} (2013), no. 7, 968--975.

\bibitem{AKM2019} Y. Ameur, N.-G. Kang and N. Makarov, Rescaling Ward identities in the random normal matrix model, \textit{Constr. Approx.} \textbf{50} (2019), no. 1, 63--127.

\bibitem{AKMW2020} Y. Ameur, N.-G. Kang, N. Makarov, Nikolai and A. Wennman, Scaling limits of random normal matrix processes at singular boundary points, \textit{J. Funct. Anal.} \textbf{278} (2020), no. 3, 108340, 46 pp.

\bibitem{AKS2018} Y. Ameur, N.-G. Kang and S.-M. Seo, The random normal matrix model: insertion of a point charge, \textit{Potential Anal.} \textbf{58} (2023), no. 2, 331--372.

\bibitem{BBLM2015} F. Balogh, M. Bertola, S.-Y. Lee and K.T.-R. McLaughlin, Strong asymptotics of the orthogonal polynomials with respect to a measure supported on the plane, \textit{Comm. Pure Appl. Math.} \textbf{68} (2015), no. 1, 112--172.

\bibitem{BGM17} F. Balogh, T. Grava and D. Merzi, Orthogonal polynomials for a class of measures with discrete rotational symmetries in the complex plane, \textit{Constr. Approx.} \textbf{46} (2017), no. 1, 109--169.

\bibitem{BasMor} E. Basor and K.E. Morrison, The Fisher-Hartwig conjecture and Toeplitz eigenvalues, \textit{Linear Algebra Appl.} \textbf{202} (1994), 129--142.

\bibitem{BEG18} M. Bertola, J.G. Elias Rebelo and T. Grava, Painlev\'{e} IV critical asymptotics for orthogonal polynomials in the complex plane, \textit{SIGMA Symmetry Integrability Geom. Methods Appl.} \textbf{14} (2018), Paper No. 091, 34 pp.

\bibitem{BK2012} P.M. Bleher and A.B.J. Kuijlaars, Orthogonal polynomials in the normal matrix model with a cubic potential,
\textit{Adv. Math.} \textbf{230} (2012), no. 3, 1272--1321.

\bibitem{BDH2022} P. Bourgade, G. Dubach and L. Hartung, Fisher-Hartwig asymptotics for non-Hermitian random matrices, In preparation (2022).

\bibitem{ByunCharlier} S.-S. Byun and C. Charlier, On the characteristic polynomial of the eigenvalue moduli of random normal matrices, arXiv:2205.04298.

\bibitem{ByunKangSeo} S.-S. Byun, N.-G. Kang and S.-M. Seo, Partition Functions of Determinantal and Pfaffian Coulomb Gases with Radially Symmetric Potentials, \textit{Comm. Math. Phys.} \textbf{401} (2023), no.2, 1627--1663.

\bibitem{BS2021} S.-S. Byun and S.-M. Seo, Random normal matrices in the almost-circular regime, \textit{Bernoulli} \textbf{29} (2023), no. 2, 1615--1637.

\bibitem{CGZ2014} D. Chafa\"{i}, N. Gozlan and P.-A. Zitt, First-order global asymptotics for confined particles with singular pair repulsion, \textit{Ann. Appl. Probab.} \textbf{24} (2014), no. 6, 2371--2413.

\bibitem{CE2020} L. Charles and B. Estienne, Entanglement entropy and Berezin-Toeplitz operators, \textit{Comm. Math. Phys.} \textbf{376} (2020), no. 1, 521--554.

\bibitem{Charlier} C. Charlier, Asymptotics of Hankel determinants with a one-cut regular potential and Fisher-Hartwig singularities, \textit{Int. Math. Res. Not. IMRN} \textbf{2019} (2019), 7515--7576.

\bibitem{CharlierBessel} C. Charlier, Exponential moments and piecewise thinning for the Bessel point process, \textit{Int. Math. Res. Not. IMRN} \textbf{2021} (2021), 16009--16071.

\bibitem{Charlier 2d jumps} C. Charlier, Asymptotics of determinants with a rotation-invariant weight and discontinuities along circles, \textit{Adv. Math.} \textbf{408} (2022), Paper No. 108600, 36 pp.

\bibitem{Charlier 2d gap} C. Charlier, Large gap asymptotics on annuli in the random normal matrix model, to appear in \textit{Math. Ann.}, arXiv:2110.06908.



\bibitem{CD2019} C. Charlier and A. Doeraene, The generating function for the Bessel point process and a system of coupled Painlev\'e V equations, \textit{Random Matrices Theory Appl.} \textbf{8} (2019), no. 3, 1950008, 31 pp.

\bibitem{ChLe2022} C. Charlier and J. Lenells, Exponential moments for disk counting statistics of random normal matrices in the critical regime, \textit{Nonlinearity} \textbf{36} (2023), no. 3, 1593--1616.

\bibitem{CZ1998} L.-L. Chau and O. Zaboronsky, \textit{On the structure of correlation functions in the normal matrix model}, \textit{Comm. Math. Phys.} \textbf{196} (1998), no. 1, 203--247.

\bibitem{CK2015} T. Claeys and I. Krasovsky, Toeplitz determinants with merging singularities, \textit{Duke Math. J.} \textbf{164} (2015), no. 15, 2897--2987.

\bibitem{CK} T. Claeys and A.B.J. Kuijlaars, Universality in unitary random matrix ensembles when the soft edge meets the hard edge, \textit{Contemp. Math.} \textbf{458} (2008), 265-280.

\bibitem{CFLV} F.D. Cunden; P. Facchi; M. Ligab\`{o}; P. Vivo, Universality of the third-order phase transition in the constrained Coulomb gas, \textit{J. Stat. Mech. Theory Exp.} \textbf{2017} (2017), no. 5, 053303, 18 pp.

\bibitem{DXZ2022 bis} D. Dai, S.-X. Xu and L. Zhang, Gap probability for the hard edge Pearcey process, \textit{Ann. Henri Poincar\'{e}} (2023), \texttt{https://doi.org/10.1007/s00023-023-01266-5}.

\bibitem{DZ2022} D. Dai, Y. Zhai, Asymptotics of the deformed Fredholm determinant of the confluent hypergeometric kernel, \textit{Stud. Appl. Math.} (2022), \texttt{https://doi.org/10.1111/sapm.12528}.

\bibitem{DeanoSimm} A. Dea\~{n}o and N. Simm, Characteristic polynomials of complex random matrices and Painlev\'{e} transcendents, \textit{Int. Math. Res. Not. IMRN} \textbf{2022} (2022), no. 1, 210--264.

\bibitem{Deift} P. Deift, \textit{Orthogonal polynomials and random matrices: a Riemann-Hilbert approach}, Courant Lecture Notes in Mathematics, 3. New York University, Courant Institute of Mathematical Sciences, New York; American Mathematical Society, Providence, RI, 1999.

\bibitem{DIKreview} P. Deift, A. Its, and I. Krasovsky, Toeplitz matrices and Toeplitz determinants under the impetus of the Ising model: some history and some recent results, \textit{Comm. Pure Appl. Math.} \textbf{66} (2013), no. 9, 1360--1438.

\bibitem{DeiftKrasVasi} P. Deift, I. Krasovsky and J. Vasilevska, Asymptotics for a determinant with a confluent hypergeometric kernel, \textit{Int. Math. Res. Not.} \textbf{9} (2011), 2117--2160.

\bibitem{ElbauFelder} P. Elbau and G. Felder, Density of eigenvalues of random normal matrices, \textit{Comm. Math. Phys.} \textbf{259} (2005), no. 2, 433--450.

\bibitem{ES2020} B. Estienne and J.-M. St\'{e}phan, Entanglement spectroscopy of chiral edge modes in the quantum Hall effect, \textit{Physical Review B} \textbf{101} (2020), no. 11, 115136.

\bibitem{Fahs} B. Fahs, Uniform asymptotics of Toeplitz determinants with Fisher-Hartwig singularities, \textit{Comm. Math. Phys.} \textbf{383} (2021), no. 2, 685--730.

\bibitem{FenzlLambert} M. Fenzl and G. Lambert, Precise deviations for disk counting statistics of invariant determinantal processes, \textit{Int. Math. Res. Not. IMRN} \textbf{2022} (2022), no. 10, 7420--7494.
\bibitem{Fo} P.J. Forrester, \textit{Log-gases and Random Matrices} (LMS-34), Princeton University Press, Princeton 2010.


\bibitem{ForresterHoleProba} P.J. Forrester, Some statistical properties of the eigenvalues of complex random matrices, \textit{Phys. Lett. A} \textbf{169} (1992), no. 1-2, 21--24.

\bibitem{GN2018} S. Ghosh and A. Nishry, Point Processes, Hole Events, and Large Deviations: Random Complex Zeros and Coulomb Gases, \textit{Constr. Approx.} \textbf{48} (2018), no. 1, 101--136.

\bibitem{GN2019} S. Ghosh and A. Nishry, Gaussian Complex Zeros on the Hole Event: The Emergence of a Forbidden Region, \textit{Comm. Pure Appl. Math.} \textbf{72}, no. 1 (2019), 3--62.

\bibitem{GHS1988} R. Grobe, F. Haake, H.-J. Sommers, Quantum distinction of regular and chaotic dissipative motion, \textit{Phys. Rev. Lett.} \textbf{61} (1988), no. 17, 1899--1902.

\bibitem{HM2013} H. Hedenmalm and N. Makarov, Coulomb gas ensembles and Laplacian growth, \textit{Proc. Lond. Math. Soc.} (3) \textbf{106} (2013), no. 4, 859--907.

\bibitem{IT} A. Its, L. Takhtajan, \textit{Normal matrix models, $\dbar$-problem, and orthogonal polynomials in the complex plane}, arXiv:0708.3867 (2007).

\bibitem{JLM1993} B. Jancovici, J. Lebowitz and G. Manificat, Large charge fluctuations in classical Coulomb systems, \textit{J. Statist. Phys.} \textbf{72} (1993), no. 3-4, 773--787.

\bibitem{Johansson} K. Johansson, On fluctuations of eigenvalues of random {H}ermitian matrices, \textit{Duke Math. J.} \textbf{91} (1998), 151--204.

\bibitem{KS1999} M.K.-H. Kiessling and H. Spohn, A note on the eigenvalue density of random matrices, \textit{Comm. Math. Phys.} \textbf{199} (1999), no. 3, 683--695.

\bibitem{LGMS2018} B. Lacroix-A-Chez-Toine, A. Grabsch, S.N. Majumdar and G. Schehr, Extremes of 2d Coulomb gas: universal intermediate deviation regime, \textit{J. Stat. Mech. Theory Exp.} \textbf{2018} (2018), no. 1, 013203, 39 pp.

\bibitem{LMS2018} B. Lacroix-A-Chez-Toine, S.N. Majumdar and G. Schehr, Rotating trapped fermions in two dimensions and the complex Ginibre ensemble: Exact results for the entanglement entropy and number variance, \textit{Phys. Rev. A} \textbf{99} (2019), 021602.

\bibitem{L et al 2019} B. Lacroix-A-Chez-Toine, J.A.M. Garz\'{o}n, C.S.H. Calva, I.P. Castillo, A. Kundu, S.N. Majumdar, and G. Schehr, Intermediate deviation regime for the full eigenvalue statistics in the complex Ginibre ensemble, \textit{Phys. Rev. E} \textbf{100} (2019), 012137.

\bibitem{LM} S.-Y. Lee, N. Makarov, Topology of quadrature domains, \textit{J. Amer. Math. Soc.} \textbf{29} (2016), 333--369.

\bibitem{LeeRiser2016} S.-Y. Lee and R. Riser, Fine asymptotic behavior for eigenvalues of random normal matrices: ellipse case, \textit{J. Math. Phys.} \textbf{57} (2016), no. 2, 023302, 29 pp.

\bibitem{LeeYang3} S.-Y. Lee and M. Yang, Strong Asymptotics of Planar Orthogonal Polynomials: Gaussian Weight Perturbed by Finite Number of Point Charges, to appear in \textit{Comm. Pure Appl. Math.}, arXiv:2003.04401.

\bibitem{LCX2022} S. Lyu, Y. Chen and S.-X. Xu, Laguerre Unitary Ensembles with Jump Discontinuities, PDEs and the Coupled Painlev\'{e} V System, arXiv:2202.00943.

\bibitem{Mehta} M.L. Mehta, Random matrices. \textit{Pure and Applied Mathematics} (Amsterdam), Vol. 142, 3rd ed., Elsevier/Academic Press, Amsterdam, 2004.



\bibitem{NAKP2020} T. Nagao, G. Akemann, M. Kieburg and I. Parra, Families of two-dimensional Coulomb gases on an ellipse: correlation functions and universality, \textit{J. Phys. A} \textbf{53} (2020), no. 7, 075201, 36 pp.

\bibitem{NishryWennman} A. Nishry and A. Wennman, The forbidden region for random zeros: appearance of quadrature domains, to appear in \textit{Comm. Pure Appl. Math}, arXiv:2009.08774.

\bibitem{NIST} F.W.J. Olver, A.B. Olde Daalhuis, D.W. Lozier, B.I. Schneider, R.F. Boisvert, C.W. Clark, B.R. Miller and B.V. Saunders, NIST Digital Library of Mathematical Functions. http://dlmf.nist.gov/, Release 1.0.13 of 2016-09-16.
\bibitem{PH} D. Petz, F. Hiai, \textit{ Logarithmic energy as an entropy functional.} Advances in differential equations and mathematical physics (Atlanta, GA, 1997), 205-221, Contemp. Math., 217, Amer. Math. Soc., Providence, RI, 1998. 46L50 (60F10)
\bibitem{RV} B. Rider, B. Vir\'{a}g, The noise in the circular law and the Gaussian free field, \textit{Int. Math. Res. Not. IMRN} \textbf{2007} (2007), no. 2, Art. ID rnm006, 33 pp.

\bibitem{SaTo} E. B. Saff and V. Totik, {\em Logarithmic Potentials with External Fields},  Grundlehren der Mathematischen Wissenschaften, Springer-Verlag, Berlin, 1997.

\bibitem{Seo0} S.-M. Seo, Edge scaling limit of the spectral radius for random normal matrix ensembles at hard edge, \textit{J. Stat. Phys.} \textbf{181} (2020), no. 5, 1473--1489.

\bibitem{Seo2021} S.-M. Seo, Edge behavior of two-dimensional Coulomb gases near a hard wall, \textit{Ann. Henri Poincar\'{e}} \textbf{23} (2021), 2247--2275.

\bibitem{Shirai} T. Shirai, Ginibre-type point processes and their asymptotic behavior, \textit{J. Math. Soc. Japan} \textbf{67} (2015), no. 2, 763--787.

\bibitem{SDMS2020} N.R. Smith, P. Le Doussal, S.N. Majumdar and G. Schehr, Counting statistics for non-interacting fermions in a $d$-dimensional potential, \textit{Phys. Rev. E} \textbf{103}, L030105.

\bibitem{SDMS2021} N.R. Smith, P. Le Doussal, S.N. Majumdar and G. Schehr, Counting statistics for non-interacting fermions in a rotating trap, \textit{Phys. Rev. A} \textbf{105} (2022), 043315.

\bibitem{Temme} N.M. Temme, \textit{Special functions: An introduction to the classical functions of mathematical physics}, John Wiley \& Sons (1996).

\bibitem{Tricomi} F.~G. Tricomi, Asymptotische {E}igenschaften der unvollst\"{a}ndigen {G}ammafunktion, \textit{Math. Z.} \textbf{53} (1950), 136--148.

\bibitem{WebbWong} C. Webb and M.D. Wong, On the moments of the characteristic polynomial of a Ginibre random matrix, \textit{Proc. Lond. Math. Soc.} (3) \textbf{118} (2019), no. 5, 1017--1056.

\bibitem{ZS2000} K. \.{Z}yczkowski and H.-J. Sommers, Truncations of random unitary matrices, \textit{J. Phys. A} \textbf{33} (2000), no. 10, 2045--2057.
\end{thebibliography}
\end{document}